\documentclass[12pt, a4paper, reqno]{amsart}

\usepackage{lipsum}
\usepackage{amsmath}
\usepackage{amssymb}
\usepackage{amsthm}
\usepackage{mathtools}
\usepackage{fancyhdr}
\usepackage{graphicx}
\usepackage{verbatim}
\usepackage{latexsym}
\usepackage{mathrsfs}
\usepackage{mathdots}

\usepackage{tikz}
\usepackage{tikz-cd}

\usepackage{luatex85} 
\usepackage[all]{xy}

\usepackage{float}
\usepackage{authblk}
\usepackage{eqparbox}

\usepackage{array}   
\usepackage{mathtools}
\usepackage{cellspace}
\usepackage{colortbl}
\usepackage{multicol}

\usepackage[foot]{amsaddr}

\usepackage[shortlabels]{enumitem}
\usepackage{pdfpages}
\usepackage{xcolor}
\usepackage{stmaryrd}
\usepackage{wallpaper}
\usepackage{setspace}
\usepackage[top = 2 cm, bottom = 3 cm, left = 2 cm, right = 2 cm]{geometry}
\usepackage{bbm}

\usepackage[UKenglish]{isodate}
\cleanlookdateon

\usepackage{anyfontsize}
\usepackage{pgf,pgfplots}
\usepackage{textcomp}
\usepackage[margin = 1cm]{caption}
\usepackage{subcaption}
\usepackage{stackengine}

\usepackage{wasysym}
\usepackage{esint} 

\usepackage[style=alphabetic,
  sorting = nyt,
  sortcites=true,
  giveninits=true,
  date=year,
  isbn=false,
  maxbibnames=99,
  maxalphanames=6,
  backend=biber]{biblatex}
\DeclareFieldFormat*{titlecase}{\MakeSentenceCase{#1}}

\DeclareSourcemap{\maps[datatype=bibtex]{\map{\step[fieldset=shorthand, null]}}}

\DeclareSourcemap{
  \maps[datatype=bibtex]{
    \map[overwrite]{
      \step[fieldsource=doi, final]
      \step[fieldset=url, null]
      \step[fieldset=eprint, null]
    }  
  }
}

\DeclareSourcemap{
  \maps[datatype=bibtex, overwrite]{
    \map{
      \step[fieldset=language, null]
      \step[fieldset=month, null]
      \step[fieldset=urldate, null]
    }
  }
}

\AtEveryBibitem{%
  \clearlist{language}%
  \clearlist{month}%
  \clearlist{urldate}%
  \clearfield{urldate}%
}

\renewbibmacro*{in:}{%
  \setunit{\addcomma\space}%
  \ifentrytype{article}
    {}
    {\printtext{%
       \bibstring{in}\intitlepunct}}}

\renewcommand*{\intitlepunct}{\addspace}

\DeclareFieldFormat[misc]{title}{\mkbibquote{#1}}
\DeclareFieldFormat[article,inproceedings]{volume}{\mkbibbold{#1}}
\DeclareFieldFormat[inproceedings]{series}{\mkbibitalic{#1}}

\usepackage[hidelinks,unicode,breaklinks=true]{hyperref}
\usepackage[capitalise,noabbrev,sort]{cleveref}

\usepackage{bookmark}

\usepackage{seqsplit}
\usepackage[notcite,notref,final]{showkeys}

\pgfplotsset{compat=1.16}
\usetikzlibrary{arrows}
\usetikzlibrary{arrows.meta}
\usetikzlibrary{patterns}
\usetikzlibrary{calc}
\usetikzlibrary{math} 

\usetikzlibrary{positioning,decorations.pathreplacing} 

\usepackage{chngcntr}
\counterwithin{figure}{section}
\counterwithin{equation}{section}
\counterwithin{table}{section}

\let\C\relax

\newcommand{\R}{\mathbb{R}}
\newcommand{\C}{\mathbb{C}}

\newcommand{\N}{\mathbb{N}}
\newcommand{\Z}{\mathbb{Z}}

\newcommand{\mcA}{\mathcal{A}}
\newcommand{\mcB}{\mathcal{B}}
\newcommand{\mcN}{\mathcal{N}}

\newcommand{\mcK}{\mathcal{K}}
\newcommand{\mcS}{\mathcal{S}}

\newcommand{\mcH}{\mathcal{H}}
\newcommand{\mcF}{\mathcal{F}}
\newcommand{\mcV}{\mathcal{V}}

\newcommand{\mcE}{\mathcal{E}}
\newcommand{\mcW}{\mathcal{W}}
\newcommand{\mcQ}{\mathcal{Q}}
\newcommand{\mcU}{\mathcal{U}}
\newcommand{\mcI}{\mathcal{I}}
\newcommand{\mcP}{\mathcal{P}}

\newcommand{\mfh}{\mathfrak{h}}
\newcommand{\mfS}{\mathfrak{S}}
\newcommand{\mfe}{\mathfrak{e}}

\newtheoremstyle{theorems}
  {3pt}
  {3pt}
  {\itshape}
  {}
  {\bfseries}
  {.}
  { }
  {}

\newtheoremstyle{proofparts}
  {3pt}
  {0pt}
  {}
  {\parindent}
  {\scshape}
  {:}
  {\newline}
  {}

\newtheoremstyle{claims}
  {2pt}
  {2pt}
  {}
  {\parindent}
  {\bfseries}
  {.}
  { }
  {}

\theoremstyle{theorems}
\newtheorem{thm}{Theorem}[section]

\newtheorem{lemma}[thm]{Lemma}
\newtheorem*{lemma*}{Lemma}

\newtheorem{prop}[thm]{Proposition}

\newtheorem*{conj*}{Conjecture}

\theoremstyle{definition}
\newtheorem{defn}[thm]{Definition}

\newtheorem{remark}[thm]{Remark}
\newtheorem{notation}[thm]{Notation}

\theoremstyle{proofparts}

\makeatletter
\@addtoreset{proofpart}{thm}
\makeatother

\theoremstyle{claims}

\newtheorem*{claim*}{Claim}

\crefname{thm}{theorem}{theorems}
\crefname{problem}{problem}{problems}
\crefname{lemma}{lemma}{lemmas}
\crefname{cor}{corollary}{corollaries}
\crefname{prop}{proposition}{propositions}
\crefname{conj}{conjecture}{conjectures}
\crefname{defn}{definition}{definitions}
\crefname{note}{note}{notes}
\crefname{ex}{example}{examples}
\crefname{remark}{remark}{remarks}
\crefname{notation}{notation}{notations}
\crefname{assumption}{assumption}{assumptions}
\crefname{claim}{claim}{claims}
\crefname{claim*}{claim}{claims}

\makeatletter
\newcommand{\Biggg}{\bBigg@{3}}
\newcommand{\vast}{\bBigg@{4}}
\newcommand{\Vast}{\bBigg@{5}}
\makeatother

\newcommand{\norm}[1]{\left\Vert #1 \right\Vert}
\newcommand{\abs}[1]{\left\vert #1 \right\vert}

\DeclareMathOperator{\supp}{supp}
\DeclareMathOperator{\tr}{tr}
\DeclareMathOperator{\Tr}{Tr}

\DeclareMathOperator{\tRe}{Re}

\DeclareMathOperator{\Li}{Li}

\DeclareMathOperator{\range}{range}
\DeclareMathOperator{\erfc}{erfc}

\definecolor{emphcolor}{rgb}{0,0,1}           

\newcommand{\expect}[1]{\left\langle #1 \right\rangle}

\newcommand{\ud}{\,\textnormal{d}}
\newcommand{\dd}[1]{\frac{\textnormal{d}}{\textnormal{d} #1}}

\newcommand{\hc}{\textnormal{h.c.}}

\let\oldepsilon\epsilon
\let\epsilon\varepsilon
\let\varepsilon\oldepsilon

\let\eps\epsilon

\title{Pressure of a dilute spin-polarized Fermi gas: Upper bound}
\author{Asbjørn Bækgaard Lauritsen}
\email{alaurits@ist.ac.at}
\author{Robert Seiringer}
\email{rseiring@ist.ac.at}
\address{Institute of Science and Technology Austria, Am Campus 1, 3400 Klosterneuburg, Austria}
\date{\today}

\allowdisplaybreaks

\begin{document}

\begin{abstract}
We prove an upper bound on the pressure of a dilute fully spin-polarized Fermi gas 
capturing the leading correction to the pressure of a free gas 
resulting from repulsive interactions. 
This correction is of order $a^3\rho^{8/3}$,
with $a$ the $p$-wave scattering length of the interaction and 
$\rho$ the particle density, depends on the temperature
and matches the corresponding lower bound of [arXiv:2307.01113]. 
\end{abstract}

\maketitle

\section{Introduction}
Dilute quantum gases have been the focus of much recent study in the mathematical physics literature. 
A natural question is the determination 
of asymptotic formulas for the ground state energy (at zero temperature) 
or the free energy or pressure (at positive temperature) 
valid for small particle densities. 
Naturally these quantities are to leading order 
given by the corresponding quantities for a free (i.e. non-interacting) gas,
and the task is to determine the correction to the quantities for a free gas arising from repulsive interactions 
in such a dilute regime. 
The first works on such problems are the works of Dyson \cite{Dyson.1957} and Lieb and Yngvason \cite{Lieb.Yngvason.1998}, 
giving upper respectively lower bounds for the ground state energy of a dilute Bose gas in $3$ dimensions. 
Here, to leading order, the ground state energy density differs from the free gas 
by a term of order $a_s \overline{\rho}^2$ with $a_s$ the $s$-\emph{wave scattering length} 
of the interaction and $\overline{\rho}$ the particle density.
Similarly the spin-$\frac{1}{2}$ Fermi gas has been studied \cite{Lieb.Seiringer.ea.2005,Falconi.Giacomelli.ea.2021,Giacomelli.2023,Lauritsen.2023}, 
and also here the ground state energy density differs from that of the free gas by a term of order $a_s \overline{\rho}^2$.

We consider here a spinless Fermi gas (equivalently completely spin-polarized) with a repulsive interaction.
At zero temperature, we found in \cite{Lauritsen.Seiringer.2024,Lauritsen.Seiringer.2024a} 
that the ground state energy density of the interacting gas differs from that  
of a free gas by a term of order $a^3\overline{\rho}^{8/3}$, 
with $a$ the $p$-\emph{wave scattering length} of the interaction.
We consider here the extension of this result to positive temperature
and show that (at fixed $\beta\mu$)
\begin{equation}\label{eqn.formula.asymp.intro}
\psi(\beta,\mu) \leq \psi_0(\beta,\mu) - c(\beta\mu) a^3 \overline{\rho}^{8/3}[1 + o(1)]
\qquad \textnormal{as } a^3\overline{\rho} \to 0,
\end{equation}
where $\psi(\beta,\mu)$ and $\psi_0(\beta,\mu)$ are the pressures of the interacting respectively free 
gas at inverse temperature $\beta$ and chemical potential $\mu$, and $c(\beta\mu)$ is an explicit 
temperature-dependent coefficient. 
This upper bound matches the corresponding lower bound in \cite{Lauritsen.Seiringer.2023a}
and we thus conclude that the asymptotic formula in \eqref{eqn.formula.asymp.intro} 
holds with equality.

Notably for dilute gases, where $\overline{\rho}$ is small, 
the correction for the spinless Fermi gas is much smaller than that for a Bose gas or spin-$\frac{1}{2}$  Fermi gas. 
This can be understood from the Pauli exclusion principle. 
For fermions of the same spin, the Pauli exclusion principle suppresses the probability 
of particles being close enough to interact. 
This has the effect that the interaction now enters through the $p$-wave scattering length 
and its effect is much smaller.

Both the dilute Bose gas and spin-$\frac{1}{2}$ Fermi gas have also been studied at positive temperature 
\cite{Yin.2010,Seiringer.2008,Seiringer.2006}
and for the Bose gas, even the next order correction has been found 
at both zero \cite{Fournais.Solovej.2020,Fournais.Solovej.2022,Yau.Yin.2009,Basti.Cenatiempo.ea.2021}
and (suitably low) positive \cite{Haberberger.Hainzl.ea.2023,Haberberger.Hainzl.ea.2024} temperature. 
Also the analogous lower-dimensional problems have been studied, 
again both at zero \cite{Agerskov.Reuvers.ea.2022,Agerskov.2023,Lieb.Yngvason.2001,Fournais.Girardot.ea.2024,Lieb.Seiringer.ea.2005}
and positive temperature \cite{Seiringer.2006,Deuchert.Mayer.ea.2020,Mayer.Seiringer.2020}.

We consider here also the two-dimensional setting and prove that, for fixed $\beta\mu$,
\begin{equation}\label{2deq}
\psi(\beta,\mu) \leq \psi_0(\beta,\mu) - c(\beta\mu) a^2 \overline{\rho}^{3}[1 + o(1)]
\qquad \textnormal{as } a^2\overline{\rho} \to 0,
\end{equation}
with the quantities being here the two-dimensional analogues of those in \eqref{eqn.formula.asymp.intro}. 
Again, this matches the lower bound in \cite{Lauritsen.Seiringer.2023a} 
and the asymptotic formula thus holds with equality. Our method does not directly extend to the one-dimensional case, however, and it remains an open problem to establish the validity of the analogue of \eqref{eqn.formula.asymp.intro} and \eqref{2deq} in one spatial dimension.

\subsection{Precise statement of results}
Consider a system of fermions confined to some large box $\Lambda = [-L/2,L/2]^3$. 
We  impose periodic boundary conditions on the box $\Lambda$. 
The pressure of the gas in the thermodynamic limit $L\to \infty$ is independent of the boundary conditions \cite{Robinson.1971}
and periodic boundary conditions is a convenient choice. 
Define the one-particle space $\mfh = L^2(\Lambda;\C)$ 
and the (fermionic) Fock space $\mcF(\mfh) = \bigoplus_{n=0}^\infty \bigwedge^n \mfh$.
On the Fock space we consider the grand canonical interacting Hamiltonian  
\begin{equation*}
\mcH = \mcH_0 + \mcV,
\end{equation*}
with the free (non-interacting) Hamiltonian $\mcH_0$ and interaction $\mcV$ given by
\begin{equation}\label{eqn.def.H0.V}
\begin{aligned}
\mcH_0 & = \!\ud\Gamma(-\Delta - \mu) = \sum_k (|k|^2 - \mu) a_k^* a_k,
\qquad 
\mcV = \!\ud \Gamma(V) = \frac{1}{2} \iint_{\Lambda \times \Lambda} \ud x \ud y \, V(x-y) a_x^* a_y^* a_y a_x.
\end{aligned}
\end{equation}
Here $a_k^* = a^*(f_k)$ and $a_k = a(f_k)$ are the creation and annihilation operators in the one-particle state $f_k(x) = L^{-3/2} e^{ikx}$
and $a_x^*$ and $a_x$ are the operator-valued distributions formally given by $a_x = L^{-3/2} \sum_k e^{ikx} a_k$.
Further, $\mu\in \R$ is the chemical potential and $V \geq 0$ is the (repulsive) two-body interaction.
We assume in addition that $V$ is radial and compactly supported. 
We used the notation

\begin{notation}
In any sum the variables are summed over $\frac{2\pi}{L}\Z^3$ unless otherwise noted. 
That is, we denote $\sum_k = \sum_{k\in \frac{2\pi}{L}\Z^3}$.
\end{notation}

We consider the pressure $\psi(\beta,\mu)$ of the interacting system. It is defined as 
\begin{equation*}
\psi(\beta,\mu) = \lim_{L\to \infty} \frac{1}{\beta L^3} \Tr_{\mcF(\mfh)} e^{-\beta \mcH}.
\end{equation*}
The limit exists and is independent of the boundary conditions \cite{Robinson.1971}.
Similarly, the pressure of the free system is given by \cite[(8.63)]{Huang.1987}
\begin{equation*}
\psi_0(\beta,\mu) = \lim_{L\to \infty} \frac{1}{\beta L^3} \Tr_{\mcF(\mfh)} e^{-\beta \mcH_0} = \frac{1}{\beta(2\pi)^3} \int_{\R^3} \log \left(1 + z e^{-\beta |k|^2}\right) \ud k,
\end{equation*}
with $z = e^{\beta\mu}$ the \emph{fugacity}.

To state our main theorem we first define the ($p$-wave) scattering length of the interaction $V$. 
We define it as in \cite{Lauritsen.Seiringer.2024a}. Different-looking but equivalent definitions are given in 
\cite{Lauritsen.Seiringer.2023a,Lauritsen.Seiringer.2024,Seiringer.Yngvason.2020}.

\begin{defn}[{\cite[Definition 1.1]{Lauritsen.Seiringer.2024a}}]\label{defn.scat.fun}
Let $\varphi_0$ be the solution of the  $p$-\emph{wave scattering equation}
\begin{equation}\label{eqn.scat}
x \Delta \varphi_0 + 2 \nabla \varphi_0 + \frac{1}{2} x V (1-\varphi_0) = 0
\end{equation}
on $\R^3$, 
with $\varphi_0(x) \to 0 $ for $|x|\to \infty$. Then $\varphi_0(x) = a^3/|x|^3$ for $x\notin \supp V$
for some constant $a > 0 $ called the $p$-\emph{wave scattering length}.
\end{defn}

The function $\varphi_0$ satisfies $0\leq \varphi_0(x)\leq \min\{1,a^3 |x|^{-3}\}$ for all $x\in \R^3$. 
The dimensionless quantity measuring the diluteness of the gas is  given by $a^3\overline{\rho}$
with $\overline{\rho} = \partial_\mu \psi(\beta,\mu)$ the (infinite volume) particle density.\footnote{We abuse notation slightly
by assuming that the derivative $\partial_\mu \psi(\beta,\mu)$ exists. 
The function $\psi(\beta,\mu)$, being convex in $\mu$, has both left and right derivatives. 
Should these not coincide we may replace instances of $\overline{\rho}$ with either.}
The density $\overline{\rho}$ is to leading order the same as that of the free gas $\overline{\rho}_0$.
More precisely, $\overline{\rho} = \overline{\rho}_0(1 + O( (a^3\overline{\rho}_0)^{1/2}))$ assuming that $z = e^{\beta\mu}$ is bounded away from zero 
\cite[Corollary 1.4]{Lauritsen.Seiringer.2023a}. 
We thus formulate the diluteness assumption as the assumption that $a^3\overline{\rho}_0$ is small. 
The density of the free gas is given by (see for instance \cite[Lemma 3.6]{Lauritsen.Seiringer.2023a})
\begin{equation*}
\overline{\rho}_0 = \partial_\mu \psi_0(\beta,\mu) 
  = - \frac{1}{(4\pi\beta)^{3/2}} \Li_{3/2}(-z),
\end{equation*}
with $\Li_s$ the \emph{polylogarithm} satisfying \cite[(25.12.16)]{dlmf}
\begin{equation*}
 - \Li_s(-e^{x}) = \frac{1}{\Gamma(s)} \int_0^\infty \frac{t^{s-1}}{e^{t-x}+1} \ud t ,
 \qquad s > 0,
 \end{equation*}
 with $\Gamma$ the Gamma function.

The temperature is naturally measured using the dimensionless fugacity $z=e^{\beta\mu}$.  
A temperature on the order of the Fermi temperature (of the free gas) $T\sim T_F \sim \overline{\rho}_0^{2/3}$ corresponds to a fugacity $z\sim 1$. 
The relevant temperatures to consider are $T \lesssim T_F$. This corresponds to $z \gtrsim 1$.
At larger temperatures thermal fluctuation dominate the quantum effects, 
and the gas should behave more like a (high-temperature) classical gas.

With these definitions at hand we may then formulate our main theorem:

\begin{thm}\label{thm.main}
Let $V\in L^1(\R^3)$ be non-negative, radial and compactly supported. 
Then, for any $\eps > 0$ there exist a function $C(z) > 0 $ bounded uniformly on any compact subset of $(0,\infty)$ 
such that if $a^3\overline{\rho}_0$ is sufficiently small, then 
\begin{equation}\label{main:eq}
\psi(\beta,\mu) \leq \psi_0(\beta,\mu) 
  - 24\pi^2 \frac{- \Li_{5/2}(-z)}{(-\Li_{3/2}(-z))^{5/3} } a^3 \overline{\rho}_0^{8/3} \left[1 + C(z) (a^3\overline{\rho}_0)^{1/16 -\eps}\right].
\end{equation}
\end{thm}

\begin{remark}\label{rem.dimensional.consistency}
The appearance of the scattering length in the error term is for dimensional consistency. 
This error term depends on the interaction $V$ also through its range and $L^1$-norm, both of 
dimension length. 
We think of the scattering length as a constant of dimension length, 
and so write $\range (V) \leq C a$ and $\norm{V}_{L^1}\leq C a$
with the constants $C>0$ being then dimensionless. The function $C(z)$ in \eqref{main:eq} depends only on these constants $C$. 
\end{remark}

\begin{remark}
The upper bound in \Cref{thm.main} matches the corresponding lower bound in \cite{Lauritsen.Seiringer.2023a}. 
Contrary to the lower bound in \cite{Lauritsen.Seiringer.2023a}, however, 
\Cref{thm.main} is not uniform in temperatures $T\lesssim T_F$. 
We discuss the reason for this in detail in \Cref{rmk.non.unif.temp.details} below. 
One should rather expect the formula in \Cref{thm.main} to hold uniformly in temperatures $T\lesssim T_F$,
meaning that we expect that $C(z)$ is bounded uniformly in $z \geq z_0 > 0$ for any fixed $z_0 > 0$.
It remains an open problem to prove this. 
\end{remark}

\begin{remark}[Extension to less regular $V$]
We can a posteriori extend \Cref{thm.main} to less regular $V$ as in \cite{Lauritsen.Seiringer.2024a}.
(In particular to the case of hard spheres, where formally $V(x) = \infty$ for $|x|\leq a$ and $V(x) = 0$ otherwise.)
Indeed, any non-negative radial measurable function $V$ can be approximated from below by a 
(non-negative, radial and compactly supported) $L^1$-function $\widetilde{V}$. 
Then clearly $\psi(\beta,\mu) \leq \widetilde{\psi}(\beta,\mu)$ with $\widetilde{\psi}(\beta,\mu)$ the pressure 
with interaction $\widetilde{V}$. 
Assume in addition that $x\mapsto |x|^2 V(x)$ is integrable outside some ball. Then the $p$-wave scattering length of $V$ is well-defined. 
If $V$ is not of finite range or $L^1$-norm, the error term in \Cref{thm.main} necessarily blows up when $\widetilde{V} \to V$.
Choosing $\widetilde{V}$ converging to $V$ slowly enough (as $a^3\overline{\rho}_0 \to 0$) however, we may achieve that $a(\widetilde{V}) = a(V) (1+o(1))$, 
while still keeping the error term in \Cref{thm.main} small. 
That is, we conclude for any such $V$ that 
\begin{equation*}
\psi(\beta,\mu) \leq \psi_0(\beta,\mu) - 24\pi^2 \frac{- \Li_{5/2}(-z)}{(-\Li_{3/2}(-z))^{5/3} } a^3 \overline{\rho}_0 \left[1 + o_z(1)\right],
\end{equation*}
where $o_z(1)$ vanishes as $a^3\overline{\rho}_0 \to 0$ uniformly for $z$ in compact subsets of $(0,\infty)$. 
\end{remark}

\subsubsection{Two dimensions}
Consider now the two-dimensional setting. The scattering length is defined  as  in \Cref{defn.scat.fun}, only in this case $\varphi_0(x) = a^2/|x|^2$ outside the support 
of $V$, with $a>0$ then the two-dimensional $p$-wave scattering length. 
The two-dimensional analogue of \Cref{thm.main} 
 reads 
\begin{thm}[Two dimensions]\label{thm.main.2d}
Let $V\in L^1(\R^2)$ be non-negative, radial and compactly supported. 
Then, for any $\eps > 0$ there exist a function $C(z) > 0 $ bounded uniformly on any compact subset of $(0,\infty)$ 
such that if $a^2\overline{\rho}_0$ is sufficiently small, then 
\begin{equation*}
\psi(\beta,\mu) \leq \psi_0(\beta,\mu) 
  - 8\pi^2 \frac{- \Li_{2}(-z)}{(-\Li_{1}(-z))^{2} } a^2 \overline{\rho}_0^{3} \left[1 + C(z) (a^2\overline{\rho}_0)^{1/8 -\eps}\right].
\end{equation*}
\end{thm} 

\Cref{thm.main.2d} matches the lower bound of \cite{Lauritsen.Seiringer.2023a}.  
Again, the upper bound in \Cref{thm.main.2d} is not uniform in temperatures $T\lesssim T_F$, 
though one should expect a uniform bound to hold. 

\section{Overview of the proof}
\label{sec.overview}

To prove \Cref{thm.main} we introduce a unitary operator $e^\mcB$ implementing the correlations 
between the particles arising from the interaction. 
This is analogous to the strategy in \cite{Falconi.Giacomelli.ea.2021,Giacomelli.2023,Lauritsen.Seiringer.2024a}, see also \cite{Giacomelli.2023a}.
We first introduce the variational formulation of the pressure, however.

\subsection{Variational formulation}
The pressure  $\psi(\beta,\mu)$ may be characterized variationally.  
We define the pressure functional $\mcP$
of a state $\Gamma$ by  
\begin{equation*}
-L^3 \mcP(\Gamma) = \expect{\mcH}_\Gamma - \frac{1}{\beta }S(\Gamma) = \Tr_{\mcF(\mfh)} \mcH \Gamma + \frac{1}{\beta} \Tr_{\mcF(\mfh)} \Gamma \log \Gamma.
\end{equation*}
By a \emph{state} we mean a positive trace-class operator on $\mcF(\mfh)$ of unit trace. 
Here we introduced the entropy $S(\Gamma) = - \Tr_{\mcF(\mfh)} \Gamma \log \Gamma$ and used the notation

\begin{notation}
We denote by $\expect{\mcA}_{\Gamma} = \Tr_{\mcF(\mfh)} \mcA \Gamma$ 
the expectation of an operator $\mcA$ in a state $\Gamma$. 
\end{notation}  

The pressure $\psi(\beta,\mu)$ at inverse temperature $\beta > 0$ and chemical potential $\mu\in \R$ 
then satisfies 
\begin{equation*}
\psi(\beta,\mu) = \lim_{L\to \infty} \sup_{\Gamma} \mcP(\Gamma).
\end{equation*}
Similarly, the pressure of the free system satisfies
\begin{equation}\label{eqn.free.gas}
\psi_0(\beta,\mu) = \lim_{L\to \infty} \sup_{\Gamma} \mcP_0(\Gamma) = \lim_{L\to \infty} \mcP_0(\Gamma_0) = \lim_{L\to \infty} \frac{1}{\beta L^3} \log Z_0
\end{equation}
with $\mcP_0$ the pressure functional of the free system (defined as $\mcP$ above only with $\mcH$ replaced by $\mcH_0$), 
$\Gamma_0 = \frac{1}{Z_0} e^{-\beta \mcH_0}$ the free Gibbs state
and 
$Z_0 = \Tr_{\mcF(\mfh)} e^{-\beta \mcH_0}$ the free partition function.

\subsection{Momentum cut-offs}
As in the zero-temperature setting with or without spin \cite{Lauritsen.Seiringer.2024a,Giacomelli.2023,Falconi.Giacomelli.ea.2021}
we wish to distinguish between small and large momenta. 
At zero temperature a natural cut-off is the Fermi momentum. 
At positive temperature there is no clearly defined Fermi momentum
since the temperature ``smears out'' the Fermi ball. 
Instead we define the following:
\begin{defn}
Define the projections $P$ and $Q = 1-P$ onto low and high momenta by 
\begin{equation}
\label{eqn.def.P.Q.kF}
\hat P(k) = \chi_{(|k| < k_F^\kappa)},
\qquad \hat Q(k) = 1 - \hat P(k),
\qquad 
k_F^\kappa := \left[\mu + \beta^{-1}\kappa\right]_+^{1/2},
\end{equation}
where $[x]_+ = \max\{0,x\}$ denotes the positive part. The parameter $\kappa>0$ will be chosen large so that $k_F^\kappa > 0$. In fact, we shall choose 
 \begin{equation}\label{eqn.choice.kappa}
\kappa = \zeta (a^3\rho_0)^{-\alpha}, 
\qquad 
\zeta := 1 + \abs{\log z},
\end{equation}
for some (small) $\alpha > 0$; for small $a^3\rho_0$ then $k_F^\kappa \gg \rho_0^{1/3}$, see Eq.~\eqref{eqn.asymp.beta.mu.kappa} below.
(Recall that  $z = e^{\beta\mu}$ denotes the fugacity.)
The temperature dependence of the error terms are naturally given in terms of $\zeta$.
Here and in the following, we use the notation 
\end{defn}

\begin{notation}
We denote by $\rho_0$ the particle density of the free gas in finite volume. 
\end{notation}

The integral kernels of $P$ and $Q$ are, with a slight abuse of notation,
\begin{equation*}
\begin{aligned}
P(x-y)  & = \frac{1}{L^3} \sum_k \hat P(k) e^{ik(x-y)},
\\
Q(x-y) & = \frac{1}{L^3} \sum_k \hat Q(k) e^{ik(x-y)} = \delta(x-y) - P(x-y).
\end{aligned} 
\end{equation*}
The momentum $k_F^\kappa$ serves as an enlarged Fermi momentum. 
The operators $P$ and $Q$ are chosen so that the approximations 
\begin{equation*}
P\gamma_0 \approx \gamma_0,  
\qquad 
Q\approx \mathbbm{1}
\end{equation*}
are both good in appropriate senses. 
Here $\gamma_0$ is the one-particle density matrix of the free Gibbs state $\Gamma_0$.

As in the zero-temperature setting studied in \cite{Lauritsen.Seiringer.2024a,Falconi.Giacomelli.ea.2021,Giacomelli.2023} 
we introduce creation and annihilation operators 
for small and large momentum. They are defined as follows.
\begin{defn}
Define the operators 
\begin{equation*}
c_k = \hat P(k) a_k ,
\qquad 
b_k = \hat Q(k) a_k,
\end{equation*}
Define further the operator-valued distributions 
\begin{equation*}
c_x = \frac{1}{L^{3/2}} \sum_{k} e^{ikx} c_k = \frac{1}{L^{3/2}} \sum_{k} e^{ikx} \hat P(k) a_k,
\qquad 
b_x = \frac{1}{L^{3/2}} \sum_{k} e^{ikx} b_k = \frac{1}{L^{3/2}} \sum_{k} e^{ikx} \hat Q(k) a_k.
\end{equation*}
We will further need a regularized version of $b$. 
We define 
\begin{equation*}
b_k^r = \hat Q^r(k) a_k, \qquad 
b_x^r = \frac{1}{L^{3/2}} \sum_k e^{ikx} b_k^r,
\end{equation*}
with $\hat Q^r : \R^3 \to [0,1]$ a smooth function with 
\begin{equation*}
\hat Q^r(k) = \begin{cases}
0 & |k|\leq k_F^\kappa, \\ 1 & |k| \geq 2k_F^\kappa,
\end{cases}
\qquad 
\abs{\nabla Q^r(k)} \leq C (k_F^\kappa)^{-1} \chi_{(k_F^\kappa \leq |k| \leq 2k_F^\kappa)}.
\end{equation*}
\end{defn}

\begin{remark}[{Comparison to \cite{Lauritsen.Seiringer.2024a}}] 
Contrary to \cite{Lauritsen.Seiringer.2024a} we do not define $c_x$ with the opposite sign in the exponent. 
Since we do not conjugate by the particle hole transform 
(which is not readily available at positive temperature), 
it is more convenient to define $c_x$ as done here. 
\end{remark}

The function $\hat Q^r$ naturally defines an operator $Q^r$ on $\mfh = L^2(\Lambda;\C)$ with integral kernel
(with a slight abuse of notation) 
\begin{equation*}
Q^r(x-y) = \frac{1}{L^3} \sum_{k} \hat Q^r(k) e^{ik(x-y)}.
\end{equation*}
As in \cite[Proposition 4.2]{Falconi.Giacomelli.ea.2021} we have 
\begin{lemma}
Let $P^r = \delta - Q^r$. Then $\norm{P^r}_{L^1(\Lambda)}\leq C$ for sufficiently large $L$.
\end{lemma}

\subsection{The scattering function}
We define the scattering function $\varphi$ as in \cite{Lauritsen.Seiringer.2024a} (see also \cite{Giacomelli.2023}) 
only with $k_F$ replaced by $k_F^\kappa$.
For completeness we recall the details here.

Let $\chi_\varphi: [0,\infty) \to [0,1]$ be a smooth function with 
\begin{equation*}
\chi_\varphi (t) = \begin{cases}
1 & t \leq 1, 
\\
0 & t \geq 2.
\end{cases}
\end{equation*}
Define then
\begin{equation}\label{eqn.def.scat.fun.phi}
\varphi(x) = \varphi_0(x) \chi_\varphi(k_F^\kappa|x|)
\end{equation}
with $\varphi_0$ the $p$-wave scattering function defined in \Cref{defn.scat.fun}.
More precisely $\varphi$ is the periodization of the right-hand-side.

\begin{remark}
We will in general abuse notation slightly and refer to any (compactly supported) function and its periodization by the same name. 
For $L$ larger than the range of the function, 
at most one summand in the periodization is non-zero and so no issue will arise. 
\end{remark}

The function $\varphi$ does not satisfy the $p$-wave scattering equation \eqref{eqn.scat}
exactly. We define the discrepancy
\begin{equation}\label{eqn.defn.mcE_varphi.scat.eqn}
\begin{aligned}
\mcE_\varphi(x) 
  & = x \Delta \varphi(x) + 2 \nabla \varphi(x) + \frac{1}{2} x V(x) (1-\varphi(x))
  \\ & = 2 k_F^\kappa x \nabla^\mu \varphi_0(x) \nabla^\mu \chi_\varphi(k_F^\kappa |x|)
  + (k_F^\kappa)^2 x \varphi_0(x) \Delta\chi_\varphi(k_F^\kappa |x|)
  + 2 k_F^\kappa \varphi_0(x) \nabla \chi_\varphi(k_F^\kappa |x|).
\end{aligned}
\end{equation}
Again, more precisely, $\mcE_\varphi$ is the periodization of the right-hand-side.
Here we have used the notation:
\begin{notation}\label{not.Einstein}
We adopt the Einstein summation convention of summing over repeated indices denoting components of a vector,
i.e., $x^\mu y^\mu = \sum_{\mu=1}^3 x^\mu y^\mu = x\cdot y$ for vectors $x,y\in \R^3$.
\end{notation}

\subsection{The operator \texorpdfstring{$\mcB$}{B}}
With the definitions above we can define the operator $\mcB$ implementing the correlations.
It is defined as 
\begin{equation}\label{eqn.def.B}
  \mcB = - \frac{1}{2} \iint \varphi(z-z') (b_z^r)^* (b_{z'}^r)^* c_{z'} c_z \ud z \ud z' - \hc
    = - \frac{1}{2L^3} \sum_{p,k,k'} \hat \varphi(p) (b_{k+p}^r)^* (b_{k'-p}^r)^* c_{k'} c_k - \hc. 
\end{equation}
We expect that the Gibbs state of the interacting gas is then approximately given by $e^{\mcB}\Gamma_0 e^{-\mcB}$.
This is analogous to the zero-temperature setting \cite{Lauritsen.Seiringer.2024a},
where one expects that the interacting ground state is approximately given by $e^\mcB \Psi_0$,
with $\Psi_0$ the ground state of the free gas.

\begin{remark}[{Comparison to \cite{Lauritsen.Seiringer.2024a,Falconi.Giacomelli.ea.2021,Giacomelli.2023}}]
In \cite{Lauritsen.Seiringer.2024a,Falconi.Giacomelli.ea.2021,Giacomelli.2023} an analogous operator $B$ is constructed. 
There, however, also a particle-hole transformation is used and so the operator $B$ is not particle number preserving. 
Conjugating the operator $B$ by the particle-hole transformation, however, we recover the operator $\mcB$ 
(up to the precise momentum cut-offs in the operators $c_k, b_k$, and the inclusion of spin considered in \cite{Falconi.Giacomelli.ea.2021,Giacomelli.2023}, in which case a different function $\varphi$ needs to be used).
\end{remark}

With the operator $\mcB$ above we  define for any state $\Gamma$
\begin{equation}\label{eqn.def.Gamma.lambda}
\Gamma^{\lambda} := e^{-\lambda \mcB} \Gamma e^{\lambda \mcB}, 
\qquad 0\leq \lambda \leq 1.
\end{equation}
Any state $\Gamma$ thus satisfies  $\Gamma = e^{\mcB} \Gamma^{1} e^{-\mcB}$.

\subsection{Implementation}\label{ss:imp}

To prove \Cref{thm.main} we shall consider the pressure functional evaluated on appropriate states. Let $\mcN=\!\ud\Gamma(1)$ denote the particle number operator. 
We shall consider the following states: 

\begin{defn}\label{def:ap}
A state $\Gamma$ is said to be an \emph{approximate Gibbs state} 
if it is translation invariant and satisfies 
\begin{equation}\label{eqn.approx.Gibbs.define}    
  \expect{\mcN}_\Gamma \leq C L^3 \rho_0 \quad , \quad 
  \expect{\mcH}_\Gamma - \frac{1}{\beta}S(\Gamma) \leq - \frac{1}{\beta} \log Z_0 + C L^3 a^3 \rho_0^{8/3}
\end{equation}
for some $C>0$ (independent of $L$, $a^3\rho_0$ and $z$).
\end{defn}

\begin{remark}
We note that the Gibbs state of the interacting gas is indeed an approximate Gibbs state (for $L$ large enough, $a^3\rho_0$ small enough and $1/z$ bounded).
Indeed, this follows from \cite[Theorem~1.2 and Corollary~1.4]{Lauritsen.Seiringer.2023a}.
\end{remark}

An immediate consequence of Definition~\ref{def:ap} is that in an approximate Gibbs state $\Gamma$ we have
\begin{equation}\label{apV}
\langle \mcV \rangle_{\Gamma} \leq C L^3 a^3 \rho_0^{8/3}\,.
\end{equation}

We shall now sketch the proof of \Cref{thm.main}.
Using \eqref{eqn.def.Gamma.lambda} above, 
any state $\Gamma$ can be written as $\Gamma = e^{\mcB} \Gamma^{1} e^{-\mcB}$. 
As discussed above, 
we expect that if $\Gamma$ is the (interacting) Gibbs state, then $\Gamma^{1}$ is approximately given by the free Gibbs state $\Gamma_0$. 
Hence, it is natural to write $\expect{\mcH}_{\Gamma} = \expect{e^{-\mcB}\mcH e^{\mcB}}_{\Gamma^{1}}$. 
Heuristically to compute $e^{-\mcB} \mcH e^{\mcB}$ we would expand to second order in a Baker--Campbell--Hausdorff expansion. 
To do this rigorously, we instead do a Duhamel expansion as follows.

Note that for any operator $\mcA$ we have 
(for any $0\leq \lambda,\lambda' \leq 1$, recall \eqref{eqn.def.Gamma.lambda})
\begin{equation}\label{eqn.duhamel}
\expect{\mcA}_{\Gamma^\lambda} = \expect{\mcA}_{\Gamma^{\lambda'}} - \int_\lambda^{\lambda'} \ud \lambda'' \, \partial_{\lambda''}\expect{\mcA}_{\Gamma^{\lambda''}}
  = \expect{\mcA}_{\Gamma^{\lambda'}} + \int_\lambda^{\lambda'} \ud \lambda'' \, \expect{[\mcA,\mcB]}_{\Gamma^{\lambda''}}.
\end{equation}
We decompose the interaction $\mcV$ as 
\begin{equation}\label{eqn.decompose.V}
\mcV = \! \ud  \Gamma(V) 
  = 
    \mcV_P + \mcV_Q + \mcV_{\textnormal{OD}}
  + \mcQ_V,
\end{equation}
with (denoting the two-body projection $P\otimes P$ by $PP$ and similar)
\begin{equation}\label{defVP}
\mcV_P = \!\ud\Gamma(PPVPP),
\qquad 
\mcV_Q = \!\ud\Gamma(QQVQQ),
\qquad 
\mcV_{\textnormal{OD}} = \!\ud\Gamma(PPVQQ + QQVPP).
\end{equation}
The operator $\mcQ_V$ containing all the remaining terms will not contribute to the pressure to the desired accuracy, i.e., it is an error term.

Decomposing as in \eqref{eqn.decompose.V} and employing  \eqref{eqn.duhamel} 
we find  for any state $\Gamma$ the formula
\begin{equation*}
\begin{aligned}
\expect{\mcH}_{\Gamma}
  & = \expect{\mcV_P}_{\Gamma}  + \expect{\mcH_0 + \mcV_Q + \mcV_{\textnormal{OD}}}_{\Gamma}  + \expect{\mcQ_V}_{\Gamma}
  \\ & = 
    \expect{\mcV_P}_{\Gamma} +  \expect{\mcH_0 + \mcV_Q}_{\Gamma^1}  + \expect{\mcQ_V}_{\Gamma} 
      + \int_0^1 \ud \lambda \expect{[\mcH_0 + \mcV_Q,\mcB] + \mcV_{\textnormal{OD}}}_{\Gamma^\lambda} 
  \\ & \quad 
      + \int_0^1 \ud \lambda \expect{[\mcV_{\textnormal{OD}},\mcB]}_{\Gamma^\lambda} 
      - \int_0^1 \ud \lambda \int_\lambda^1 \ud \lambda' \expect{[\mcV_{\textnormal{OD}},\mcB]}_{\Gamma^{\lambda'}}
  \\ & = \expect{\mcV_P}_{\Gamma} +  \expect{\mcH_0 + \mcV_Q}_{\Gamma^1}  + \expect{\mcQ_V}_{\Gamma} 
      + \int_0^1 \ud \lambda \expect{[\mcH_0 + \mcV_Q,\mcB] + \mcV_{\textnormal{OD}}}_{\Gamma^\lambda} 
  \\ & \quad 
      + \int_0^1 \ud \lambda \, (1-\lambda) \expect{[\mcV_{\textnormal{OD}},\mcB]}_{\Gamma^\lambda} .
\end{aligned}
\end{equation*}
We shall write $[\mcV_{\textnormal{OD}},\mcB] = 2(\mcW_P-\mcV_P) + \mcQ_{\textnormal{OD}}$, with $W = V(1-\varphi)$ and 
$\mcW_P = \ud \Gamma(PPWPP)$.
The term $\mcQ_{\textnormal{OD}}$ is then an error term. 
Again using \eqref{eqn.duhamel},
\begin{equation*}
\begin{aligned}
\expect{[\mcV_{\textnormal{OD}},\mcB]}_{\Gamma^\lambda} 
  & = 2 \expect{\mcW_P - \mcV_P}_{\Gamma} 
    - 2 \int_{0}^\lambda \ud \lambda' \expect{[\mcW_P-\mcV_P,\mcB]}_{\Gamma^{\lambda'}} 
    + \expect{\mcQ_{\textnormal{OD}}}_{\Gamma^{\lambda}}.
\end{aligned}
\end{equation*}
Defining $\mcQ_{V\varphi} = [\mcV_P-\mcW_P,\mcB]$ and  the  error terms 
\begin{align}
\mcE_V(\Gamma) & = \expect{\mcQ_V}_\Gamma,
\label{eqn.def.mcE.V}
\\
\mcE_{\textnormal{scat}}(\Gamma)
  & = \int_0^1 \ud \lambda \expect{[\mcH_0 + \mcV_Q,\mcB] + \mcV_{\textnormal{OD}}}_{\Gamma^\lambda} ,
\label{eqn.def.mcE.scat}
\\
\mcE_{\textnormal{OD}}(\Gamma)
  & = 
    \int_0^1 \ud \lambda \, (1-\lambda) \expect{\mcQ_{\textnormal{OD}}}_{\Gamma^{\lambda}} ,
\label{eqn.def.mcE.OD}
\\
\mcE_{V\varphi}(\Gamma)
  & = \int_0^1 \ud \lambda \, (1-\lambda)^2 \expect{\mcQ_{V\varphi}}_{\Gamma^{\lambda}} ,
\label{eqn.def.mcE.Vphi}
\end{align}
we find for any state $\Gamma$
\begin{equation*}
\begin{aligned}
\expect{\mcH}_{\Gamma}
  & = \expect{\mcH_0}_{\Gamma^1} 
    + \expect{\mcW_P}_{\Gamma}
    + \expect{\mcV_{Q}}_{\Gamma^1} 
    + \mcE_V(\Gamma) + \mcE_{\textnormal{scat}}(\Gamma)
    + \mcE_{\textnormal{OD}}(\Gamma) 
    + \mcE_{V\varphi}(\Gamma).
\end{aligned}
\end{equation*}
This is analogous to the similar formula \cite[(2.11)]{Lauritsen.Seiringer.2024a} in the zero-temperature setting. 

The desired leading contribution from the interaction is the term $\expect{\mcW_P}_{\Gamma}$. In fact, we expect that to leading order we can replace $\Gamma$ by $\Gamma_0$ in this expression, and drop the projection $P$ (i.e., replace it by $1$). 
With $\mcW = \!\ud\Gamma(W)$, let us consequently define the difference between $\expect{\mcW_P}_{\Gamma}$ and $\expect{\mcW}_{\Gamma_0}$ by  $\mcE_{\textnormal{pt}}$;  i.e.,    
\begin{equation}\label{eqn.def.mcE.1st}
\expect{\mcW_P}_{\Gamma} = \expect{\mcW}_{\Gamma_0} + \mcE_{\textnormal{pt}}(\Gamma)\,.
\end{equation}
The term $\mcE_{\textnormal{pt}}(\Gamma)$ is another error term. It quantifies, to some extent, the validity of \emph{first order perturbation theory}, as we shall explain in more detail \Cref{sec.1st.order}. 
Furthermore, by the Gibbs variational principle (applied to the free system) we have 
\begin{equation*}
\expect{\mcH_0}_{\Gamma^1} - \frac{1}{\beta} S(\Gamma^1) \geq \expect{\mcH_0}_{\Gamma_0} - \frac{1}{\beta}S(\Gamma_0)
  = -\frac{1}{\beta}\log Z_0.
\end{equation*}
Finally, $\mcV_{Q}\geq 0$ and $S(\Gamma^1) = S(\Gamma)$ since $e^{\mcB}$ is unitary.
Thus we find  for any state $\Gamma$ the lower bound
\begin{equation}\label{eqn.main}
\boxed{
\begin{aligned}
\expect{\mcH}_\Gamma - \frac{1}{\beta}S(\Gamma) 
  & 
  \geq -\frac{1}{\beta}\log Z_0
    + \expect{\mcW}_{\Gamma_0}
    + \mcE_{\textnormal{pt}}(\Gamma)
    + \mcE_V(\Gamma) 
    + \mcE_{\textnormal{scat}}(\Gamma)
    + \mcE_{\textnormal{OD}}(\Gamma) 
    + \mcE_{V\varphi}(\Gamma).
\end{aligned}
}
\end{equation}

To evaluate $\expect{\mcW}_{\Gamma_0}$ we use \Cref{lem.prop.rho0} below. 
Note that $\int |x|^2 W = 24 \pi a^3$.
Indeed, by the compact support of $V$ we have for small enough $a^3\rho_0$
\begin{equation*}
\int_{\R^3} |x|^2 W(x) \ud x = \int_{|x| < (k_F^\kappa)^{-1}} |x|^2 V(x) (1-\varphi_0(x)) \ud x = 24 \pi a^3.
\end{equation*}
The first `$=$' follows from the fact that $\varphi(x) = \varphi_0(x)$ for $|x|\leq (k_F^{\kappa})^{-1}$, noting the asymptotics for $k_F^\kappa$ in \eqref{eqn.asymp.beta.mu.kappa}.
The last `$=$' follows from a simple integration by parts using the scattering equation \eqref{eqn.scat}. 
Bounding also $\int |x|^4 W \leq C a^5$ 
(the integral $\int |x|^4 W$ is some constant of dimension $(\textnormal{length})^5$
and hence we write it as $C a^5$, cf. \Cref{rem.dimensional.consistency}) 
we have 
\begin{equation}\label{eqn.calc.expect.W}
\expect{\mcW}_{\Gamma_0} = 24\pi^2 \frac{-\Li_{5/2}(-z)}{(-\Li_{3/2}(-z))^{5/3}} L^3 a^3 \rho_0^{8/3}
  \left[1 + O(a^2 \rho_0^{2/3}) 
  + \mathfrak{e}_L
  \right].
\end{equation}
Here we have used the notation 
\begin{notation}
We denote by $\mathfrak{e}_L$ any term vanishing as $L^{-1}$ 
in the limit $L\to \infty$.
\end{notation}

The main part of this paper concerns bounding the error terms 
$\mcE_{\textnormal{pt}}(\Gamma)$, $\mcE_V(\Gamma)$, $\mcE_{\textnormal{scat}}(\Gamma)$, $\mcE_{\textnormal{OD}}(\Gamma)$ and 
$\mcE_{V\varphi}(\Gamma)$. They are bounded as follows. 
Recall our choice $\kappa = \zeta (a^3\rho_0)^{-\alpha}$ for some (small) $\alpha>0$ in \eqref{eqn.choice.kappa}.

\begin{prop}\label{prop.mcE.1st}
Let $\Gamma$ be an approximate Gibbs state. Then, for any $\alpha > 0$ sufficiently small,
\begin{equation*}
\mcE_{\textnormal{pt}}(\Gamma) \geq - C(z) L^3 a^{3} \rho_0^{8/3} 
(a^3\rho_0)^{1/16 - 21\alpha/32} 
- L^3 \mfe_L,
\end{equation*}
with the function $C(z)$ uniformly bounded on compact subsets of $(0,\infty)$.
\end{prop}

\begin{prop}\label{prop.mcE.V}
Let $\Gamma$ be an approximate Gibbs state. 
Then, for any $\alpha > 0$,
\begin{equation*}
\mcE_V(\Gamma) \geq - C L^3  a^{3} \rho_0^{8/3}  
(a^3\rho_0)^{1/2-5\alpha/2} - L^3\mfe_L.
\end{equation*}
\end{prop}

\begin{prop}\label{prop.mcE.scat}
  Let $\Gamma$ be an approximate Gibbs state. 
  Then, for any $\alpha > 0$ sufficiently small,
\begin{equation*}
\abs{\mcE_{\textnormal{scat}}(\Gamma)} \leq C L^3 a^3 \rho_0^{8/3} 
  (a^3\rho_0)^{1/5 - 2\alpha} \abs{\log a^3\rho_0}^{2/5}
  + L^3\mfe_L.
\end{equation*}
\end{prop}

\begin{prop}\label{prop.mcE.OD}
  Let $\Gamma$ be an approximate Gibbs state. 
  Then, for any $\alpha > 0$ sufficiently small,
\begin{equation*}
\abs{\mcE_{\textnormal{OD}}(\Gamma) } \leq 
  C L^3 a^3 \rho_0^{8/3}
  (a^3\rho_0)^{1/2 - 13\alpha/4}
  + L^3\mfe_L.
\end{equation*}
\end{prop}

\begin{prop}\label{prop.mcE.Vphi}
  Let $\Gamma$ be an approximate Gibbs state. 
  Then, for any $\alpha > 0$,
\begin{equation*}
\abs{\mcE_{V\varphi}(\Gamma)} \leq   C L^3 a^3 \rho_0^{8/3}
  (a^3\rho_0)^{2/3 - 4\alpha}
  + L^3\mfe_L.
\end{equation*}
\end{prop}

With these we may  give the 

\begin{proof}[{Proof of \Cref{thm.main}}]
We evaluate the first two terms in \eqref{eqn.main} using \eqref{eqn.free.gas} and \eqref{eqn.calc.expect.W}
and bound the error terms by Propositions \ref{prop.mcE.1st}--\ref{prop.mcE.Vphi}. 
Taking an infinite volume limit, we obtain the desired bound \eqref{main:eq}.
\end{proof}

\begin{remark}
We note that the error bounds in Propositions \ref{prop.mcE.V}--\ref{prop.mcE.Vphi} are uniform in the temperature, 
while the bound in \Cref{prop.mcE.1st} is not. 
\end{remark}

\begin{remark}
Propositions \ref{prop.mcE.V}--\ref{prop.mcE.Vphi} are analogous to \cite[Propositions 2.8--2.10]{Lauritsen.Seiringer.2024a}.
The main novelty of the present paper is \Cref{prop.mcE.1st}. 
\end{remark}

The remainder of the paper deals with the proofs of Propositions \ref{prop.mcE.1st}--\ref{prop.mcE.Vphi}.

\subsubsection*{Structure of the paper:}
First, in \Cref{sec.preliminaries}, we recall some useful results from \cite{Lauritsen.Seiringer.2023a,Lauritsen.Seiringer.2024a}, 
give some a priori bounds for approximate Gibbs states and conclude with the proof of \Cref{prop.mcE.V}. 
Next, in \Cref{sec.calc.commutators}, we calculate the commutators $[\mcH_0,\mcB]$, $[\mcV_Q,\mcB]$, $[\mcV_P - \mcW_P,\mcB]$, 
and $[\mcV_{\textnormal{OD}},\mcB]$ and extract the claimed leading terms. 
The error terms arising from the commutators are estimated in \Cref{sec.bdd.commutators}. 
In \Cref{sec.propagate.a.priori}, we propagate the a priori bounds to the states $\Gamma^{\lambda}$ 
and conclude the proof of \Cref{prop.mcE.OD,prop.mcE.scat,prop.mcE.Vphi}.
Finally, in \Cref{sec.1st.order}, we employ the method of \cite{Seiringer.2006a}  in order to estimate the validity of first order perturbation theory at positive temperature,   
and conclude the proof of \Cref{prop.mcE.1st} and thereby of \Cref{thm.main}.

In \Cref{sec.bdd.sums.riemann} we estimate certain Riemann sums needed in the proof, and in \Cref{sec.two.dimensions}
we sketch how to adapt the proof to the two-dimensional setting, and hence to prove \Cref{thm.main.2d}.

\section{Preliminaries}\label{sec.preliminaries}

\subsection{Reduced densities}

We recall from \cite{Lauritsen.Seiringer.2023a} the following properties for the reduced $k$-particle densities $\rho_0^{(k)}$ of the free gas.

\begin{lemma}[{\cite[Lemma 3.6]{Lauritsen.Seiringer.2023a}}]\label{lem.prop.rho0}
The reduced densities of the free Fermi gas satisfy 
\begin{align}
\rho_0^{(1)}(x_1)
  & = \rho_0= \frac{1}{(4\pi)^{3/2}} \beta^{-3/2} (-\Li_{3/2}(-z))
    \left[1 + O(L^{-1}\zeta \rho_0^{-1/3})\right],
  \label{eqn.rho1.free}
  \\
  \rho_0^{(2)}(x_1,x_2)
  & = 2\pi \frac{- \Li_{5/2}(-z)}{(-\Li_{3/2}(-z))^{5/3}} \rho_0^{8/3} |x_1-x_2|^2
    \left[1 
    +  O(\rho_0^{2/3}|x_1-x_2|^2)
    +O (L^{-1} \zeta \rho_0^{-1/3} ) 
    \right].
  \label{eqn.rho2.free}
\end{align}
\end{lemma} 
For bounded $1/z = e^{-\beta\mu}$, we have the asymptotic behavior
\begin{equation}
\label{eqn.asymp.beta.mu.kappa}
\beta \sim \zeta \rho_0^{-2/3},
\qquad 
|\mu| \leq C \rho_0^{2/3},
\qquad 
k_F^\kappa \sim \kappa^{1/2}\zeta^{-1/2} \rho_0^{1/3} 
\sim \rho_0^{1/3} (a^3\rho_0)^{-\alpha/2}.
\end{equation}
Indeed, the asymptotics for $\beta$ and for positive $\mu$ follows from \cite[Remark 3.7]{Lauritsen.Seiringer.2023a}. 
For $1/z$ bounded we have $\mu \geq - C \beta^{-1} \geq - C \zeta^{-1} \rho_0^{2/3}$ and the bound for $\mu$ follows for negative $\mu$.
The asymptotics for $k_F^\kappa$ follows from the first two by recalling \eqref{eqn.choice.kappa}
noting that for $a^3\rho_0$ small  we have $\beta^{-1}\kappa \gg -\mu$.

\subsection{The scattering function}

We recall from \cite{Lauritsen.Seiringer.2023a} that the scattering function $\varphi$ and the function $\mcE_\varphi$, 
defined in \eqref{eqn.def.scat.fun.phi} and \eqref{eqn.defn.mcE_varphi.scat.eqn}, satisfy 

\begin{lemma}[{\cite[Lemma 3.6 and Remark 3.7]{Lauritsen.Seiringer.2024a}}]
\label{lem.prop.phi.scattering.fun}
The scattering function $\varphi$ satisfies 
\begin{equation*}
\begin{aligned}
\norm{|\cdot|^n \varphi}_{L^1}
& \leq 
C a^3 (k_F^\kappa)^{-n}, 
& n &= 1,2,
\quad  
&
\norm{|\cdot|^n \nabla^n \varphi}_{L^1}
& \leq C a^3 \abs{\log a k_F^\kappa }, 
&n&=0,1,2
\quad  
\\
\norm{ |\cdot| \varphi}_{L^2 }
& \leq 
C a^{3/2+1}, 
&&
\quad 
&
\norm{|\cdot|^n \nabla^n \varphi}_{L^2}
& \leq C a^{3/2}, 
&n&=0,1.
\end{aligned}
\end{equation*}
Here $\nabla^n$ represents any combination of $n$ derivatives
and $\abs{\cdot}$ denotes the metric on the torus in the sense that $\abs{x}$ is the distance between $x$ and the point $0$.

Moreover, $\mcE_\varphi$ 
satisfies the same bounds as $\varphi$ only with an additional power $k_F^\kappa$.  
\end{lemma}

\subsection{Relative entropy}
The relative entropy of a state $\Gamma$ with respect to the free Gibbs state $\Gamma_0$ is given by 
\begin{equation*}
\begin{aligned}
\frac{1}{\beta } S(\Gamma, \Gamma_0) 
  = 
    \expect{\mcH_0}_\Gamma - \frac{1}{\beta} S(\Gamma) + \frac{1}{\beta} \log Z_0.
\end{aligned}
\end{equation*}
For a translation invariant state $\Gamma$ with one-particle density matrix $\gamma$, we have 
$S(\Gamma) \leq \mcS(\gamma)$, where the `one-particle-density' entropy is given in terms of the Fourier coefficients $\hat\gamma(k)$ of $\gamma$ as 
$\mcS(\gamma) = - \sum_k \left[\hat \gamma(k) \log \hat \gamma(k) + (1-\hat\gamma(k)) \log (1- \hat\gamma(k))\right]$. 
Moreover, $S(\Gamma_0) = \mcS(\gamma_0)$ with $\gamma_0$ the one-particle-density matrix of $\Gamma_0$, 
since $\Gamma_0$ is a quasi-free state.
In particular, using that $\hat\gamma_0(k) = (1 + e^{\beta(|k|^2-\mu)})^{-1}$,
\begin{equation*}
\begin{aligned}
\frac{1}{\beta}
  \mcS(\gamma, \gamma_0) 
  & :=  
  \frac{1}{\beta}
  \sum_k 
  \left[
  \hat \gamma(k) \log \frac{\hat\gamma(k)}{\hat\gamma_0(k)}
  + 
  (1-\hat\gamma(k)) \log \frac{1-\hat\gamma(k)}{1 - \hat\gamma_0(k)}
  \right]
  \nonumber
  \\ & 
  = \expect{\mcH_0}_\Gamma - \expect{\mcH_0}_{\Gamma_0}
    - \frac{1}{\beta} \mcS(\gamma) + \frac{1}{\beta} \mcS(\gamma_0)
  \leq \frac 1 \beta S(\Gamma, \Gamma_0).
\end{aligned}
\end{equation*}
Recalling \eqref{eqn.approx.Gibbs.define}, we thus have for an approximate Gibbs state $\Gamma$ 
\begin{align}
\mcS(\gamma, \gamma_0) 
  & \leq S(\Gamma, \Gamma_0)
  \leq C L^3 a^3 \rho_0^{8/3} \beta
  \leq C L^3 \zeta a^3  \rho_0^2,
\label{eqn.bdd.rel.entropy.a.priori}
\end{align}
where we used the asymptotics of $\beta$ from \eqref{eqn.asymp.beta.mu.kappa} in the last bound.

\subsection{Highly excited particles}
To state many of the intermediary bounds in the proof of Propositions~\ref{prop.mcE.V}--\ref{prop.mcE.Vphi} 
it is convenient to define 
\begin{equation}\label{def:NK}
\mcN_Q = \sum_k b_k^* b_k = \int b_x^* b_x \ud x  ,
\qquad 
\mcK_Q = \sum_k |k|^2 b_k^* b_k
= \int \nabla b_x^* \nabla b_x \ud x ,
\end{equation}
being the number and kinetic energy of highly excited particles. 
We shall prove the following lemma.

\begin{lemma}\label{lem.highly.excited.rel.entropy}
For any translation invariant state $\Gamma$ with one-particle density matrix $\gamma$, we have
\begin{equation*}
\begin{aligned}
\expect{\mcN_Q}_{\Gamma}
  & \lesssim \kappa^{-1} \left( \mcS(\gamma,\gamma_0) 
    + L^3 \beta^{-1} k_F^\kappa e^{-\kappa/3} + L^3 \mathfrak{e}_L\right)
\\
\expect{\mcK_Q}_{\Gamma}
  & \lesssim \left(\beta^{-1} + \rho_0^{2/3} \kappa^{-1}\right) 
    \left(\mcS(\gamma,\gamma_0) 
    + L^3 \beta^{-1} k_F^\kappa e^{- \kappa/3}
    + L^3 \mathfrak{e}_L
    \right)
\end{aligned}
\end{equation*}
\end{lemma}

Applying this Lemma for the choice $\kappa = \zeta (a^3\rho_0)^{-\alpha}$ in \eqref{eqn.choice.kappa}
and recalling \eqref{eqn.asymp.beta.mu.kappa} and \eqref{eqn.bdd.rel.entropy.a.priori}
we find for any approximate Gibbs state $\Gamma$
\begin{equation}\label{eqn.a.priori.NQ.KQ}
\expect{\mcN_Q}_\Gamma 
\leq C L^3 \zeta \kappa^{-1}  a^3 \rho_0^2 + L^3 \mathfrak{e}_L,
\qquad 
\expect{\mcK_Q}_{\Gamma}
  \leq C L^3  a^3 \rho_0^{8/3} + L^3\mfe_L
\end{equation}
with the constants $C>0$ depending only on $\alpha > 0$.

\begin{proof}
Define the functions
\begin{equation*}
s(t) = - t \log t - (1-t) \log (1-t),
\qquad 
s(t,t') = t\log \frac{t}{t'} + (1-t) \log \frac{1-t}{1-t'}.
\end{equation*}
Then for any translation-invariant density matrices $\gamma, \gamma'$ we have 
$\mcS(\gamma) = \sum_k s(\hat \gamma(k))$
and 
$\mcS(\gamma, \gamma') = \sum_k s(\hat\gamma(k), \hat \gamma'(k))$.

Let $h > 0$. We claim that for all $t\in [0,1]$
\begin{equation}\label{eqn.claim.bdd.KE.rel.entropy}
ht \leq 2 s(t,t_0) + \frac{h}{1+e^{h/2}}, \qquad t_0 = \frac{1}{1+e^h}
\end{equation}
To prove this we 
define $F(h) := - \log(1 + e^{-h}) = ht_0 - s(t_0)$ the `free energy' of the `Hamiltonian' $h$
and note that $s(t,t_0) = ht - s(t) - F(h)$. 
Then, by the Gibbs variational principle applied to the `Hamiltonian' $h/2$ we have 
\begin{equation*}
\begin{aligned}
2s(t,t_0) - ht 
& = 2 \left[\frac{1}{2} h t - s(t) - F(h)\right]
\geq 2 \left[F(h/2) - F(h)\right]
= - 2 h \int_{1/2}^1 F'( u h) \ud u.
\end{aligned}   
\end{equation*}
Clearly $|F'( u h)| \leq (1+ e^{h/2})^{-1}$ for any $u\in [1/2,1]$, 
yielding \eqref{eqn.claim.bdd.KE.rel.entropy}.

In order to prove the lemma we apply \eqref{eqn.claim.bdd.KE.rel.entropy} 
for $h = \beta(|k|^2 - \mu) \hat Q(k)$ and $t=\hat \gamma(k)$ and sum in $k$. 
Noting that $t_0 = (1 + e^h)^{-1} = \hat \gamma_0(k) \hat Q(k)$ on the support of $\hat Q$, and $s(\hat \gamma(k), \hat \gamma_0(k)) \geq 0$, we have 
\begin{equation*}
\sum_{|k| > k_F^\kappa} \beta(|k|^2 - \mu) \hat \gamma(k) 
\leq 2 \mcS(\gamma, \gamma_0) + \sum_{|k| > k_F^\kappa} \frac{\beta(|k|^2 - \mu)}{1 + e^{\frac{\beta}{2}(|k|^2 - \mu)}} .
\end{equation*}
The sum on the right hand side can be bounded as 
\begin{equation}\label{eqn.riemann.h/1+exp.h}
\sum_{|k| > k_F^\kappa} \frac{\beta(|k|^2 - \mu)}{1 + e^{\frac{\beta}{2}(|k|^2 - \mu)}}
  \leq C L^3 \beta^{-3/2} (1 + k_F^\kappa \beta^{1/2}) e^{-\kappa/3} + L^3 \mfe_L 
  \leq C L^3 \beta^{-1} k_F^\kappa e^{-\kappa/3} + L^3 \mfe_L.
\end{equation}
This follows by viewing the sum as a Riemann sum and computing the corresponding integral; we shall give the details in \Cref{sec.bdd.sums.riemann}.
The bounds in the lemma then follow by noting that 
$\beta((k_F^\kappa)^2 - \mu) = \kappa$ according to \eqref{eqn.def.P.Q.kF}
and 
\begin{equation*}
\begin{aligned}
\expect{\mcN_Q}_{\Gamma}
  & = \sum_k \hat Q(k) \hat \gamma(k) 
  \leq [\beta((k_F^\kappa)^2 - \mu)]^{-1} \sum_{|k| > k_F^\kappa} \beta(|k|^2-\mu) \gamma(k),
\\
\expect{\mcK_Q}_{\Gamma}
  & = \sum_k \hat Q(k) |k|^2 \hat \gamma(k) 
  = \beta^{-1} \sum_{|k|> k_F^\kappa}  \beta(|k|^2-\mu) \gamma(k)
  + \mu \expect{\mcN_Q}_\Gamma.
\end{aligned}
\end{equation*}
Using 
\eqref{eqn.asymp.beta.mu.kappa} to bound $\mu$ yields the statement of the Lemma.
\end{proof}

\subsection{The interaction \texorpdfstring{$V$}{V}}
We shall now give a priori bounds for the interaction $V$ and prove \Cref{prop.mcE.V}. 
This is analogous to what is done in \cite[Section 3.2]{Lauritsen.Seiringer.2024a}.

Recall the definition of $\mcE_{V}$ in \eqref{eqn.def.mcE.V}) and~\eqref{eqn.decompose.V}.

\begin{lemma}\label{lem.bdd.mcE.V.first}
For any state $\Gamma$  and any $\eps > 0$ we have the lower bound 
\begin{equation*}
\mcE_V(\Gamma)
  \geq - \eps \expect{\mcV_P}_\Gamma 
      - \eps \expect{\mcV_Q}_\Gamma 
    - C \eps^{-1} a^3 (k_F^\kappa)^3 \left[\expect{\mcK_Q}_\Gamma +  (k_F^\kappa)^2 \expect{\mcN_Q}_\Gamma \right]
\end{equation*}
\end{lemma}

\begin{proof}
Again denoting by $PP$ the two-body operator $P\otimes P$, we can write
\begin{equation*}
V = [PP + PQ + QP + QQ] V [PP + PQ + QP + QQ].
\end{equation*}
Expanding out the product we may bound the cross-terms as 
\begin{equation*}
\begin{aligned}
(QQ V (PQ+QP) + \hc) & \geq -\eps QQ V QQ - \frac{1}{\eps} (PQ + QP)V (PQ+QP)
\\
(PP V (PQ+QP) + \hc) & \geq -\eps PP V PP - \frac{1}{\eps} (PQ+QP)V (PQ+QP).
\end{aligned}
\end{equation*}
With the definitions in \eqref{defVP} we thus obtain the bound
\begin{equation*}
\mcV \geq (1-\eps) \mcV_P + (1-\eps) \mcV_Q + \mcV_{\textnormal{OD}} + \left(1 - \frac{2}{\eps}\right) \mcV_{\textnormal{X}}
\end{equation*}
with 
\begin{equation*}
\mcV_{\textnormal{X}} = \!\ud\Gamma( (PQ + QP) V (PQ + QP))
  = \frac{1}{2} \iint V(x-y) (b_x^* c_y^* + c_x^* b_y^*) (b_y c_x + c_y b_x) \ud x \ud y .
\end{equation*}
For any state $\Gamma$ define the function $\phi(x,y) = \expect{(b_x^* c_y^* + c_x^* b_y^*) (b_y c_x + c_y b_x)}_\Gamma$. 
We note that $\phi(x,x)=0$ and $\phi(x,y) = \phi(y,x)$. Hence, Taylor expanding in $y$ around $y=x$, the zeroth and first orders vanish. 
(The path from $x$ to $y$ used in the Taylor expansion should be interpreted as the shortest path on the torus $\Lambda$.)
Then 
\begin{equation*}
\phi(x,y) = (y-x)^{\mu} (y-x)^\nu \int_0^1 \ud t\, (1-t) [\nabla^\mu_2 \nabla^\nu_2 \phi](x,x+t(y-x))
\end{equation*}
where $\nabla^\mu_2$ denotes the derivative in the second variable. 
Changing variables to $z=y-x$, we have 
\begin{equation*}
\expect{\mcV_{\textnormal{X}}}_\Gamma
  = \frac{1}{2} \int_0^1\ud t \, (1-t) \int \ud z \, z^\mu z^\nu V(z) \int \ud x \,  [\nabla^\mu_2 \nabla^\nu_2 \phi](x,x+tz).
\end{equation*}
The $x$-integral is given by (with derivatives now being with respect to $x$ and the brackets indicating that the derivatives are of the product 
of operators in the brackets)
\begin{equation*}
\begin{aligned}
& \int \ud x \Bigl< 
b_x^* \nabla^\mu \nabla^\nu [c_{x+tz}^* b_{x+tz}] c_x 
  + c_x^* \nabla^\mu \nabla^\nu [b_{x+tz}^* b_{x+tz}] c_x 
\\ & \qquad 
  + b_x^* \nabla^\mu \nabla^\nu [c_{x+tz}^* c_{x+tz}] b_x 
  + c_x^* \nabla^\mu \nabla^\nu [b_{x+tz}^* c_{x+tz}] b_x
\Bigr>.
\end{aligned}
\end{equation*}
Integrating by parts once, so that no factor $b$ or $b^*$ carry two derivatives we may bound this using Cauchy--Schwarz by 
\begin{equation*}
\begin{aligned}
\abs{\int \ud x \,  [\nabla^\mu_2 \nabla^\nu_2 \phi](x,x+tz)}
  & \leq C (k_F^\kappa)^{3} \int \ud x \expect{\nabla b^*_x \nabla b_x}_{\Gamma} 
    + C (k_F^\kappa)^{5} \int \ud x \expect{b_x^* b_x}_\Gamma
  \\ & 
    = C (k_F^\kappa)^3 \expect{\mcK_Q}_\Gamma 
   + C (k_F^\kappa)^5 \expect{\mcN_Q}_\Gamma 
\end{aligned}
\end{equation*}
using that $\norm{\nabla^n c} \leq C (k_F^\kappa)^{3/2+n}$.
We conclude that 
\begin{equation*}
\abs{\expect{\mcV_{\textnormal{X}}}_\Gamma}
  \leq C \norm{|\cdot|^2 V}_{L^1} (k_F^\kappa)^3 \left[\expect{\mcK_Q}_\Gamma +  (k_F^\kappa)^2 \expect{\mcN_Q}_\Gamma \right].
\end{equation*}
Noting finally that $\int |x|^2 V \leq C a^3$ (it is some constant of dimension $(\textnormal{length})^3$, cf. \Cref{rem.dimensional.consistency}) 
we conclude the proof of the lemma. 
\end{proof}

\begin{lemma}[{A priori bound for $\mcV_{P}$}]
\label{lem.a.priori.VP}
For any state $\Gamma$ we have 
\begin{equation*}
\expect{\mcV_P}_\Gamma \leq C a^3 (k_F^\kappa)^{5} \expect{\mcN}_\Gamma. 
\end{equation*}

\end{lemma}

\begin{proof}

We have 
\begin{equation*}
\expect{\mcV_P}_\Gamma = \frac{1}{2} \iint \ud x \ud y V(x-y) \expect{c_x^* c_y^* c_y c_x}_\Gamma.
\end{equation*}
We Taylor expand the factors $c_y$ and $c_y^*$ as 
\begin{equation}\label{eqn.Taylor.cy.1st.order}
c_y = c_x + (y-x)^\mu \int_0^1 \ud t \, \nabla^\mu c_{x+t(y-x)},
\end{equation}
where again the path from $x$ to $y$ should be understood as the shortest on the torus $\Lambda$. 
Since $c_x$ is a bounded fermionic operator $c_x^2 = 0$. 
Doing this Taylor expansion for both $c_y$ and $c_y^*$ and changing variables to $z=y-x$ 
we have 
\begin{equation*}
\expect{\mcV_P}_\Gamma = \frac{1}{2} \iint \ud x \ud z \, V(z) z^\mu z^\nu \int_0^1 \ud t \int_0^1 \ud s \expect{c_x^* \nabla^\mu c_{x+tz}^* \nabla^\nu c_{x+sz} c_x}_\Gamma
\end{equation*}
To bound this we can use that $\norm{\nabla c}\leq C (k_F^\kappa)^{3/2+1}$. Thus, with $\| \, \cdot \, \|_{\mathfrak{S}_2}$ denoting the Hilbert--Schmidt norm, the Cauchy--Schwarz inequality implies that 
\begin{equation*}
\begin{aligned}
\expect{\mcV_P}_\Gamma
  & \leq C \iint \ud x \ud z \, V(z) |z|^2 (k_F^\kappa)^{5} \norm{\Gamma^{1/2} c_x^*}_{\mathfrak{S}_2} \norm{c_x \Gamma^{1/2}}_{\mathfrak{S}_2}.
\end{aligned}
\end{equation*}
Noting that $\norm{\Gamma^{1/2} c_x^*}_{\mathfrak{S}_2} = (\Tr [\Gamma^{1/2} c_x^* c_x \Gamma^{1/2}])^{1/2} = \expect{c_x^* c_x}_\Gamma^{1/2}$, 
that $\int \expect{c_x^* c_x}_\Gamma \ud x \leq \expect{\mcN}_\Gamma$ and that $\int |x|^2 V \leq C a^3$
we obtain the desired bound.
\end{proof}

\begin{lemma}[{A priori bound for $\mcV_{\textnormal{OD}}$}]\label{lem.a.priori.V.OD}
For any state $\Gamma$ and any $\delta > 0$ we have 
\begin{equation*}
\abs{\expect{\mcV_{\textnormal{OD}}}_\Gamma}
  \leq \delta \expect{\mcV_Q}_{\Gamma} + \delta^{-1} \expect{\mcV_P}_\Gamma
  \leq \delta \expect{\mcV_Q}_\Gamma + C \delta^{-1} a^3 (k_F^\kappa)^5 \expect{\mcN}_\Gamma.
\end{equation*}

\end{lemma}

\begin{proof}
By Cauchy--Schwarz we have for any $\delta > 0$
\begin{equation*}
\begin{aligned}
  \abs{\expect{\mcV_{\textnormal{OD}}}_\Gamma}
    & \leq \frac{1}{2}\iint V(x-y) \abs{\expect{b_x^* b_y^* c_y c_x + \hc}_\Gamma} \ud x \ud y 
    \leq 
      \delta \expect{\mcV_Q}_\Gamma + \delta^{-1} \expect{\mcV_P}_\Gamma.
\end{aligned}
\end{equation*}
Using the bound of $\expect{\mcV_P}_\Gamma$ in \Cref{lem.a.priori.VP} we conclude the desired bound. 
\end{proof}

\begin{lemma}[{A priori bound for $\mcV_{Q}$}]\label{lem.a.priori.VQ}
For any approximate Gibbs state $\Gamma$ we have 
\begin{equation*}
\expect{\mcV_Q}_\Gamma \leq C L^3 a^3 \rho_0 (k_F^\kappa)^5 + L^3 \mfe_L
\end{equation*}

\end{lemma}
\begin{proof}
Recalling \eqref{eqn.decompose.V} and applying \Cref{lem.bdd.mcE.V.first} and \Cref{lem.a.priori.V.OD} with $\delta = 1/2$ we have the lower bound
\begin{equation*}
\expect{\mcV}_\Gamma \geq -(1+\eps) \expect{\mcV_P}_\Gamma + \left(\frac{1}{2} - \eps\right) \expect{\mcV_Q}_\Gamma
  - \eps^{-1} C a^3 (k_F^\kappa)^3 \left[\expect{\mcK_Q}_\Gamma +  (k_F^\kappa)^2 \expect{\mcN_Q}_\Gamma \right].
\end{equation*}
Moreover, $\expect{\mcV}_\Gamma$ is bounded above as in \eqref{apV}. 
Bounding $\expect{\mcV_P}_\Gamma$ using \Cref{lem.a.priori.VP}, choosing $\eps = 1/4$, 
and using the bounds of $\mcN_Q$ and $\mcK_Q$ in \eqref{eqn.a.priori.NQ.KQ}
we conclude the proof of the lemma. 
\end{proof}

Finally we can the give the

\begin{proof}[{Proof of \Cref{prop.mcE.V}}] 
We start with \Cref{lem.bdd.mcE.V.first} and apply  the bounds of \Cref{lem.a.priori.VQ,lem.a.priori.VP} as well as \eqref{eqn.a.priori.NQ.KQ}. We conclude for any $\eps > 0$ the bound 
\begin{equation*}
\mcE_V(\Gamma) \geq 
- C \eps L^3 a^3 \rho_0 (k_F^\kappa)^5 
- C \eps^{-1} L^3 
  a^6 \rho_0^2 (k_F^\kappa)^5 - L^3\mfe_L
\end{equation*}
Choosing the optimal 
$\eps = a^{3/2}\rho_0^{1/2}$ 
and recalling the asymptotic behavior of  $k_F^\kappa$ in \eqref{eqn.asymp.beta.mu.kappa}
we arrive at the claimed bound.
\end{proof}

\section{Calculation of commutators}\label{sec.calc.commutators}

As detailed in Section~\ref{ss:imp}, our method of proof requires the calculation of the commutator of the operator $\mcB$ defined in \eqref{eqn.def.B} with various other operators. These commutators are analogous to the corresponding ones in the zero-temperature setting \cite{Lauritsen.Seiringer.2024a,Falconi.Giacomelli.ea.2021,Giacomelli.2023}. 
The calculations are essentially the same as those of \cite[Section 4]{Lauritsen.Seiringer.2024a}, only the particle-hole transformation used there is now absent. 
For completeness we give the details here.

\subsection{\texorpdfstring{$[\mcH_0,\mcB]$}{[H0,B]}:}
Recall the formulas for $\mcH_0$ and $\mcB$ from \eqref{eqn.def.H0.V} and \eqref{eqn.def.B}.
We have 
\begin{equation*}
\begin{aligned}
[\mcH_0,\mcB]
  & = - \frac{1}{2L^3} \sum_{k,k',p,q} (|q|^2 - \mu) \hat \varphi(p) [a_q^* a_q , (b_{k+p}^r)^* (b_{k'-p}^r)^* c_{k'} c_k] + \hc .
\end{aligned}     
\end{equation*}
To calculate the commutator we note that $[a_q^*a_q,a_k] = -\delta_{q,k} a_k$.
Thus, 
\begin{equation*}
\begin{aligned}
  [a_q^* a_q , (b_{k+p}^r)^* (b_{k'-p}^r)^* c_{k'} c_k] 
    & = (\delta_{q,k+p} + \delta_{q,k'-p} - \delta_{q,k'} - \delta_{q,k})  (b_{k+p}^r)^* (b_{k'-p}^r)^* c_{k'} c_k
\end{aligned}     
\end{equation*}
Noting further that 
\begin{equation*}
\begin{aligned}
(|k+p|^2 - \mu) + (|k-'p|^2 - \mu) - (|k'|^2 - \mu) - (|k|^2 - \mu)
  & = 2|p|^2 + 2p\cdot(k-k')
\end{aligned}
\end{equation*}
we have
\begin{equation*}
\begin{aligned}
[\mcH_0,\mcB]
  & = - \frac{1}{L^3} \sum_{k,k',p} (|p|^2 + p\cdot(k-k')) \hat \varphi(p) (b_{k+p}^r)^* (b_{k'-p}^r)^* c_{k'} c_k + \hc .
\end{aligned}     
\end{equation*}
Using the symmetry of interchanging $k\leftrightarrow k'$ and $p \to -p$ and writing in configuration space we have 
\begin{equation*}
[\mcH_0, \mcB]
   = \iint \ud x \ud y \, \Bigl(\Delta \varphi(x-y)  (b_x^r)^* (b_y^r)^* c_y c_x 
  + 2\nabla^\mu \varphi(x-y)  (b_x^r)^* (b_y^r)^* c_y \nabla^\mu c_x \Bigr) + \hc
\end{equation*}
As a first step we replace  the $b^r$'s by $b$'s and write 
\begin{equation}\label{eqn.def.H0.div.r}
[\mcH_0, \mcB]
   = \iint \ud x \ud y \, \Bigl(\Delta \varphi(x-y)  b_x^* b_y^* c_y c_x 
  + 2\nabla^\mu \varphi(x-y)  b_x^* b_y^* c_y \nabla^\mu c_x \Bigr) + \hc + \mcH_{0;\mcB}^{\div r},
\end{equation}
with $\mcH_{0;\mcB}^{\div r}$ defined so that this holds. 
In the first term in this expression we Taylor expand $c_x$ around $x=y$. 
More precisely we have (as in \cite[Equation (4.2)]{Lauritsen.Seiringer.2024a})
\begin{equation}\label{eqn.Taylor.cx.2nd.order}
c_x 
  = c_y + (x-y)^\mu \nabla^\mu c_x - (x-y)^\mu (x-y)^\nu \int_0^1 \ud t \, (1-t) \nabla^\mu \nabla^\nu c_{x+t(y-x)}.
\end{equation}
Note here that the first order term is evaluated at $x$ and not at $y$.
With this we have 
\begin{equation}\label{eqn.[H0.B].decompose}
[\mcH_0, \mcB]
  = \iint \ud x \ud y \, \Bigl([(\cdot)^\mu\Delta \varphi + 2\nabla^\mu \varphi](x-y)  b_x^* b_y^* c_y \nabla^\mu c_x \Bigr) + \hc
  + \mcH_{0;\mcB}^{\textnormal{Taylor}}
  + \mcH_{0;\mcB}^{\div r},
\end{equation}
with $\mcH_{0;\mcB}^{\textnormal{Taylor}}$ defined such that this holds, i.e. as the first term in \eqref{eqn.def.H0.div.r} 
only with $c_x$ replaced by the last term in \eqref{eqn.Taylor.cx.2nd.order}.

\subsection{\texorpdfstring{$[\mcV_Q,\mcB]$}{[VQ,B]}:}
The operator $\mcV_Q$ defined in \eqref{defVP} is given by 
\begin{equation*}
\begin{aligned}
\mcV_Q & = \frac{1}{2} \iint \ud x \ud y \, V(x-y) b_x^* b_y^* b_y b_x .
\end{aligned}
\end{equation*}
Thus, recalling again the formula for $\mcB$ in \eqref{eqn.def.B}, 
\begin{equation}\label{VQBinte}
\begin{aligned}
[\mcV_Q,\mcB]
  & = - \frac{1}{4} \iiiint \ud x \ud y \ud z \ud z' \, V(x-y) \varphi(z-z') 
      [b_x^* b_y^* b_y b_x, (b_z^r)^* (b_{z'}^r)^* c_{z'} c_z] + \hc.
\end{aligned}
\end{equation}
Since $b$'s and $c$'s anti-commute, the commutator is given by 
\begin{equation*}
[b_x^* b_y^* b_y b_x, (b_z^r)^* (b_{z'}^r)^* c_{z'} c_z]
  = b_x^* b_y^* [b_y b_x, (b_z^r)^* (b_{z'}^r)^*] c_{z'} c_z. 
\end{equation*}
Using that 
\begin{equation*}
\{b_x, (b_z^r)^*\}
  = \frac{1}{L^3}\sum_k \hat Q^r(k) e^{ik(x-z)} = Q^r(x-z),
\end{equation*}
one computes 
\begin{equation}\label{eqn.calc.[bb.bb]}
\begin{aligned}
[b_y b_x, (b_z^r)^* (b_{z'}^r)^*]
  & = Q^r (x-z) Q^r (y-z') - Q^r (y-z) Q^r (x-z') 
\\ & \quad 
  - Q^r (x-z) (b_{z'}^r)^* b_y 
  + Q^r(y-z) (b_{z'}^r)^* b_x 
\\ & \quad 
  + Q^r(x-z') (b_z^r)^* b_y 
  - Q^r(y-z') (b_z^r)^* b_x .
\end{aligned}
\end{equation}
The first two and the last four  terms give the same contribution to 
$[\mcV_Q,\mcB]$ when the integrations in \eqref{VQBinte} are computed, by the symmetries of interchanging $x\leftrightarrow y$ or $z\leftrightarrow z'$.
By writing $Q^r = \delta -P^r $ we thus obtain
\begin{equation}\label{eqn.[VQ.B].calc}
\begin{aligned}
[\mcV_Q,\mcB]
  & = -\frac{1}{2} \iint \ud x \ud y \, V(x-y) \varphi(x-y) b_x^* b_y^* c_y c_x + \hc + \mcQ_{[\mcV_Q,\mcB]}
\end{aligned}
\end{equation}
with 
\begin{equation}\label{eqn.Q.[V.B].define}
\begin{aligned}
\mcQ_{[\mcV_Q,\mcB]}
  & = \frac{1}{2} \iiiint 
    V(x-y) \varphi(z-z') 
    \Bigl[
      \left( 2 \delta(x-z) P^r(y-z') - P^r(x-z)P^r(y-z') \right) b_x^* b_y^* c_{z'} c_z
  \\ & \qquad 
      + (\delta(x-z) - P^r(x-z)) b_x^* b_y^* (b_{z'}^r)^* b_y c_{z'} c_z
    \Bigr]
  \ud x \ud y \ud z \ud z' + \hc.
\end{aligned}
\end{equation}
In the first term in \eqref{eqn.[VQ.B].calc} we again Taylor expand the factor $c_x$ using \eqref{eqn.Taylor.cx.2nd.order}. Then, 
\begin{equation}\label{eqn.[VQ.B].decompose}
\begin{aligned}
[\mcV_Q,\mcB]
  & = -\frac{1}{2} \iint \ud x \ud y \, (x-y)^\mu V(x-y) \varphi(x-y) b_x^* b_y^* c_y \nabla^\mu c_x + \hc + \mcV_{Q;\mcB}^{\textnormal{Taylor}} + \mcQ_{[\mcV_Q,\mcB]}
\end{aligned}
\end{equation}
with $\mcV_{Q;\mcB}^{\textnormal{Taylor}}$ as in the first term of \eqref{eqn.[VQ.B].calc} only with $c_x$ replaced by the last term in \eqref{eqn.Taylor.cx.2nd.order}.

\subsection{\texorpdfstring{$[\mcV_P - \mcW_P,\mcB]$}{[VP-WP,B]}:}
We start by noting that 
\begin{equation*}
\begin{aligned}
\mcV_P - \mcW_P & = \frac{1}{2} \iint \ud x \ud y \, V(x-y)\varphi(x-y) c_x^* c_y^* c_y c_x .
\end{aligned}
\end{equation*}
Thus, recalling again the formula for $\mcB$ in \eqref{eqn.def.B},
\begin{equation*}
\begin{aligned}
\mcQ_{V\varphi}
  & = 
[\mcV_P - \mcW_P, \mcB]
  \\ & 
  = - \frac{1}{4} \iiiint \ud x \ud y \ud z \ud z' \, V(x-y) \varphi(x-y) \varphi(z-z') [c_x^* c_y^* c_y c_x, (b_z^r)^* (b_{z'}^r)^* c_{z'} c_z] + \hc.
\end{aligned}
\end{equation*}
The commutator is given by
\begin{equation*}
\begin{aligned}
[c_x^* c_y^* c_y c_x, (b_z^r)^* (b_{z'}^r)^* c_{z'} c_z]
  & = (b_z^r)^* (b_{z'}^r)^*[c_x^* c_y^* , c_{z'} c_z] c_y c_x
\end{aligned}
\end{equation*}
To bound this term it is convenient to order the $c$'s and $b$'s not in normal order, but instead according to their indices $x,y$ and $z,z'$
such that factors $c$ and $b$ with indices $z,z'$ are first. 
This corresponds to anti-normal ordering the $c$-commutator.
Using that 
\begin{equation*}
\{c_x^*, c_z\} = \frac{1}{L^3} \sum_k \hat P(k) e^{ik(z-x)} = P(z-x) = P(x-z)
\end{equation*}
we calculate
\begin{equation}\label{eqn.calc.[cc.cc]}
\begin{aligned}
[c_{x}^* c_{y}^*, c_{z'} c_{z} ]
  & = P(x-z)P(y-z') - P(x-z')P(y-z)
  \\ & \quad 
    - P(x-z) c_{z'}c_y^* + P(y-z) c_{z'}c_x^* + P(x-z') c_z c_{y}^* - P(y-z') c_z c_x^*.
\end{aligned}
\end{equation}
From the symmetries of interchanging $x\leftrightarrow y$ or $z\leftrightarrow z'$ 
the first two as well as the last four terms give the same contribution to $[\mcV_P - \mcW_P, \mcB]$ when the integrations are performed.
Thus we find 
\begin{equation}\label{eqn.Q.Vphi.define}
\begin{aligned}
\mcQ_{V\varphi}
  & = - \frac{1}{2} \iiiint \ud x \ud y \ud z \ud z' \, V(x-y) \varphi(x-y) \varphi(z-z') 
  \\ & \quad \times 
    \Bigl[
      P(x-z) P(y-z') (b_z^r)^* (b_{z'}^r)^*  c_y c_x
      -
      2 P(x-z) (b_z^r)^* (b_{z'}^r)^* c_{z'} c_y^* c_y c_x
    \Bigr] + \hc.
\end{aligned}
\end{equation}

\subsection{\texorpdfstring{$[\mcV_{\textnormal{OD}},\mcB]$}{[VOD,B]}:}
The operator $\mcV_{\textnormal{OD}}$ defined in \eqref{defVP} takes the form 
\begin{equation*}
\begin{aligned}
\mcV_{\textnormal{OD}} & = \frac{1}{2} \iint \ud x\ud y \,  V(x-y) c_x^* c_y^* b_y b_x + \hc.  
\end{aligned}
\end{equation*}
Recalling again the formula for $\mcB$ in \eqref{eqn.def.B},  we have
\begin{equation*}
\begin{aligned}
[\mcV_{\textnormal{OD}},\mcB]
  & = - \frac{1}{4} \iiiint \ud x \ud y \ud z \ud z' \, V(x-y) \varphi(z-z') [c_x^* c_y^* b_y b_x, (b_z^r)^* (b_{z'}^r)^* c_{z'} c_z] + \hc.
\end{aligned}
\end{equation*}
The commutator is given by 
\begin{equation*}
\begin{aligned}
[c_x^* c_y^* b_y b_x, (b_z^r)^* (b_{z'}^r)^* c_{z'} c_z]
  & = c_x^* c_y^* [b_y b_x, (b_{z}^r)^* (b_{z'}^r)^*] c_{z'} c_z 
    - (b_z^r)^* (b_{z'}^r)^* [c_{z'} c_z, c_x^* c_y^*] b_y b_x.
\end{aligned}
\end{equation*}
The first term is the leading one. Using the formula in \eqref{eqn.calc.[bb.bb]} and writing $Q^r = \delta - P^r$
the main term is the term with two $\delta$'s. The rest are error terms. 
Thus we have 
\begin{equation*}
[\mcV_{\textnormal{OD}},\mcB] = 2(\mcW_P - \mcV_P) + \mcQ_{\textnormal{OD}}
\end{equation*}
with 
\begin{equation}\label{eqn.Q.OD.define}
\begin{aligned}
\mcQ_{\textnormal{OD}}
  & = \frac{- 1}{4} \iiiint 
    V(x-y) \varphi(z-z') 
    \Bigl[
      \left( - 4 \delta(x-z) P^r(y-z') + 2P^r(x-z)P^r(y-z') \right) c_x^* c_y^* c_{z'} c_z
  \\ & \qquad 
      - 4(\delta(x-z) - P^r(x-z)) c_x^* c_y^* (b_{z'}^r)^* b_y c_{z'} c_z
  \\ & \qquad 
      + (b_z^r)^* (b_{z'}^r)^* [c_x^* c_y^*, c_{z'} c_z] b_y b_x
    \Bigr]
  \ud x \ud y \ud z \ud z' 
  + \hc.
\end{aligned}
\end{equation}

\section{Bounding commutators}\label{sec.bdd.commutators}
To bound the error terms $\mcE_{\textnormal{scat}}$, $\mcE_{\textnormal{OD}}$ and $\mcE_{V\varphi}$ defined in \eqref{eqn.def.mcE.scat}--\eqref{eqn.def.mcE.Vphi}, 
in this section we first derive useful bounds on the expectation values of the operators $[\mcH_0 + \mcV_Q,\mcB] + \mcV_{\textnormal{OD}}$, 
$\mcQ_{\textnormal{OD}}$, and $\mcQ_{V\varphi}$ in general states. 
In the subsequent \Cref{sec.propagate.a.priori} we shall use these bound for the particular states $\Gamma^{\lambda}$ with $\Gamma$ an approximate Gibbs state 
to conclude the proof of \Cref{prop.mcE.scat,prop.mcE.OD,prop.mcE.Vphi}.

To show that $[\mcH_0 + \mcV_Q,\mcB] + \mcV_{\textnormal{OD}}$ is appropriately small, we write 
\begin{equation}\label{eqbe}
\begin{aligned}
\mcV_{\textnormal{OD}} & = \frac{1}{2} \iint \ud x\ud y \,  V(x-y) b_x^* b_y^* c_y c_x + \hc
  \\ & = \frac{1}{2} \iint \ud x\ud y \,  (x-y)^\mu V(x-y) b_x^* b_y^* c_y \nabla^\mu c_x + \hc + \mcV_{\textnormal{OD}}^{\textnormal{Taylor}}
\end{aligned}
\end{equation}
by Taylor expanding $c_x$ as in \eqref{eqn.Taylor.cx.2nd.order} and with $\mcV_{\textnormal{OD}}^{\textnormal{Taylor}}$ appropriately defined.
Recalling \eqref{eqn.[H0.B].decompose} and \eqref{eqn.[VQ.B].decompose} as well as \eqref{eqn.defn.mcE_varphi.scat.eqn},  we obtain
\begin{equation}\label{eqn.decompose.scat.eqn.operators}
\begin{aligned}
  [\mcH_0 + \mcV_Q,\mcB] + \mcV_{\textnormal{OD}}
    & = \mcQ_{\textnormal{scat}} + \mcH_{0;\mcB}^{\div r} + \mcH_{0;\mcB}^{\textnormal{Taylor}} 
    + \mcQ_{[\mcV_Q,\mcB]} + \mcV_{Q;\mcB}^{\textnormal{Taylor}} 
    + \mcV_{\textnormal{OD}}^{\textnormal{Taylor}},
\end{aligned}
\end{equation}
with 
\begin{equation*}
\mcQ_{\textnormal{scat}}
  = \iint \ud x \ud y \, \mcE_\varphi^\mu(x-y) b_x^* b_y^* c_y \nabla^\mu c_x + \hc. 
\end{equation*}

The main result of this section are bounds on the operators 
\begin{equation*}
\mcH_{0;\mcB}^{\div r},
\quad 
\mcQ_{\textnormal{Taylor}} := \mcH_{0;\mcB}^{\textnormal{Taylor}} + \mcV_{Q;\mcB}^{\textnormal{Taylor}} + \mcV_{\textnormal{OD}}^{\textnormal{Taylor}},
\quad 
\mcQ_{\textnormal{scat}},
\quad 
\mcQ_{[\mcV_Q,\mcB]},
\quad 
\mcQ_{V\varphi},
\quad 
\mcQ_{\textnormal{OD}}.
\end{equation*}
To state the bounds it will be convenient to define for any (small) $\delta > 0$ 
\begin{equation}\label{eqn.def.kF.delta}
k_F^{\delta^{-1}\kappa} = [\beta^{-1} + \mu \delta^{-1}\kappa]_+^{1/2},
\qquad
\mcN_{Q^>} = \sum_{|k| > k_F^{\delta^{-1}\kappa}} a_k^* a_k.
\end{equation}
That is, $k_F^{\delta^{-1}\kappa}$ is defined as $k_F^\kappa$ in \eqref{eqn.def.P.Q.kF} only with $\kappa$ replaced by $\delta^{-1}\kappa$.
Recall that $\kappa$ is chosen in \eqref{eqn.choice.kappa} as $\kappa = \zeta (a^3\rho_0)^{-\alpha}$ for some (small) $\alpha>0$. Recall also the definition of the operators $\mcN_Q$ and $\mcK_Q$ in \eqref{def:NK}.

\begin{remark}[{Bound on $\mcN_{Q^>}$}]
  \label{rmk.bdd.NQ>}
Noting that $\mcN_{Q^>}$ is defined as $\mcN_Q$ only with $\kappa$ replaced 
by $\delta^{-1}\kappa$ we can apply \Cref{lem.highly.excited.rel.entropy}
and conclude that $\mcN_{Q^>}$ satisfies the same bound as $\mcN_Q$
only with $\kappa$ replaced by $\delta^{-1}\kappa$. In particular, 
recalling \eqref{eqn.asymp.beta.mu.kappa} and \eqref{eqn.bdd.rel.entropy.a.priori}
we conclude as in \eqref{eqn.a.priori.NQ.KQ} that for an approximate 
Gibbs state $\Gamma$ we have 
\begin{equation*}
  \expect{\mcN_{Q^>}}_\Gamma \leq C L^3 \zeta \delta \kappa^{-1} a^3 \rho_0^2 + L^3 \mfe_L.
\end{equation*}
\end{remark}

We shall prove (analogously to \cite[Lemma 5.1]{Lauritsen.Seiringer.2024a})

\begin{lemma}\label{lem.bdd.list}
Let $\Gamma$ be any state. Then, for $\alpha > 0$ sufficiently small and any $0 < \delta < 1$, 
\begin{align}
\abs{\expect{\mcH_{0;\mcB}^{\div r}}_\Gamma}
  & \leq C a^3 (k_F^\kappa)^{4} \abs{\log a k_F^\kappa}\left[ k_F^{\kappa} \expect{\mcN_{Q}}_{\Gamma}^{1/2} + \expect{\mcK_{Q}}_{\Gamma}^{1/2}\right] \expect{\mcN}_\Gamma^{1/2},
\label{eqn.lem.bdd.H-r}
\\
\abs{\expect{\mcQ_{\textnormal{Taylor}}}_\Gamma} + \abs{\expect{\mcQ_{\textnormal{scat}}}_\Gamma}
  & \leq 
  C a^3 \abs{\log a k_F^\kappa} (k_F^{\delta^{-1}\kappa})^{3/2} (k_F^\kappa)^{3/2+2}  \expect{\mcN_{Q}}_\Gamma^{1/2} \expect{\mcN}_\Gamma^{1/2}
  \nonumber
\\ & \quad 
  +  C L^{3/2} a^{3/2} (k_F^\kappa)^{5} \expect{\mcN_{Q^>}}_{\Gamma}^{1/2},
\label{eqn.lem.bdd.Qscat.QTaylor}
\\
\abs{\expect{\mcQ_{[\mcV_Q,\mcB]}}_\Gamma}
  & \leq 
C a^3 (k_F^\kappa)^4 \left[\expect{\mcN_Q}_{\Gamma}^{1/2}  + (a k_F^\kappa)^{1/2} \expect{\mcN}_\Gamma^{1/2}\right] \expect{\mcV_Q}_{\Gamma}^{1/2},
\label{eqn.lem.bdd.Q[VQ.B]}
\\ 
\abs{\expect{\mcQ_{V\varphi}}_\Gamma}
  & \leq C a^{2 + 3/2} (k_F^{\kappa})^{4 + 3/2} \expect{\mcN_Q}_{\Gamma}^{1/2} \expect{\mcN}_\Gamma^{1/2},
\label{eqn.lem.bdd.QVphi}
\\ 
\abs{\expect{\mcQ_{\textnormal{OD}} }_\Gamma}
  & \leq 
  C a^{4+1/2} (k_F^\kappa)^{6+1/2}
   \expect{\mcN}_\Gamma
  +
  C a^2 (k_F^\kappa)^{3} \expect{\mcV_Q}_{\Gamma}^{1/2} \expect{\mcN_Q}_{\Gamma}^{1/2}.
\label{eqn.lem.bdd.QOD}
\end{align}
\end{lemma}

For the proof the following lemma  from \cite{Lauritsen.Seiringer.2024a} will be very useful.

\begin{lemma}[{\cite[Lemma 5.2]{Lauritsen.Seiringer.2024a}}]\label{lem.bdd.b(F)}
Let $F$ be a compactly supported function with $F(x)=0$ for $|x| \geq C (k_F^\kappa)^{-1}$.
Then, uniformly in $x\in \Lambda$ and $t\in [0,1]$, (with $\nabla^n$ denoting any $n$'th derivative)
\begin{align*}
\norm{\int F(x-y) a_y^* c_y \ud y} 
  & \leq C (k_F^\kappa)^{3/2} \norm{F}_{L^2} ,
  \\ 
\norm{\int F(x-y) a_y^* c_y \nabla^n c_{ty+(1-t)x} \ud y  }
  & \leq C (k_F^{\kappa})^{3+n} \norm{F}_{L^2}.
\end{align*}
\end{lemma}

As straightforward consequences we obtain as in \cite[Equations (5.5), (5.6)]{Lauritsen.Seiringer.2024a}
\begin{align}
\norm{\int \varphi(z-z') (b^r_{z'})^* c_{z'} \ud z'} 
& \leq C (k_F^\kappa)^3 \norm{\varphi}_{L^1} + C (k_F^\kappa)^{3/2} \norm{\varphi}_{L^2}
\leq C (ak_F^\kappa)^{3/2}
\label{eqn.bdd.phi*bc} 
\\
\label{eqn.bdd.phi*bcc}
\norm{\int \varphi(z-z') (b^r_{z'})^* c_{z'} c_z \ud z'}
& \leq C (k_F^\kappa)^{4+3/2} \norm{|\cdot|\varphi}_{L^1} + C (k_F^\kappa)^{4} \norm{|\cdot|\varphi}_{L^2}
\leq C (k_F^\kappa)^{3/2} (ak_F^\kappa)^{5/2}
\end{align}
where we used \Cref{lem.prop.phi.scattering.fun} in the last step. 

The remainder of this section is devoted to the proof of \Cref{lem.bdd.list}.

\begin{remark}
Many of the bounds and computations in the following are similar to the analogous bounds and computations in \cite{Lauritsen.Seiringer.2024a}. 
There, however, the operators are conjugated by a particle-hole transformation. 
At positive temperature there is no such natural `particle-hole transformation' --- there is no filled Fermi ball in the free system. Hence the bounds given here appear somewhat different than in the zero-temperature case studied in 
\cite{Lauritsen.Seiringer.2024a}.  
\end{remark}

\subsection{(Non-)regularization of \texorpdfstring{$[\mcH_0,\mcB]$}{[H0,B]}}
According to \eqref{eqn.def.H0.div.r}, $\mcH_{0;\mcB}^{\div r}$ is given by 
\begin{equation}\label{h0br}
\mcH_{0;\mcB}^{\div r}
  = 
  \iint \ud x \ud y \, \left((b_x^r)^* (b_y^r)^* - b_x^* b_y^*\right) 
  \Bigl(\Delta \varphi(x-y)   c_y c_x 
  + 2\nabla^\mu \varphi(x-y)  c_y \nabla^\mu c_x \Bigr) 
  + \hc.
\end{equation}
To bound this write $b^r_x = b_x + b_x^{\div r}$ with  $b_x^{\div r} = L^{-3/2} \sum_k e^{ikx} b_k^{\div r}$
and $b_k^{\div r} =  (\hat Q^r(k) - \hat Q(k)) a_k$.
We note that $\hat Q^r - \hat Q$ is supported on the set $k_F^\kappa \leq |k| \leq 2k_F^\kappa$.
In particular $\norm{b_x^{\div r}} \leq C (k_F^\kappa)^{3/2}$.
We further Taylor expand the factor $c_x$ in the first term in the integrand in \eqref{h0br} around $x=y$ as in \eqref{eqn.Taylor.cy.1st.order} (with $x$ and $y$ interchanged) 
and change variables to $z=x-y$. 
Then for any state $\Gamma$ we obtain the identity 
\begin{equation}\label{eqn.H-r.decompose}
\begin{aligned}
\expect{\mcH_{0;\mcB}^{\div r}}_{\Gamma}
  & = 2 \tRe \iint \ud z \ud y \, z^\mu \Delta \varphi(z) \int_0^1 \ud t \, 
    \expect{\left(b_{y+z}^* (b_y^{\div r})^* + (b_{y+z}^{\div r})^* b_y^* + (b_{y+z}^{\div r})^* (b_y^{\div r})^*\right) c_y \nabla^\mu c_{y+tz}}_{\Gamma}
  \\ & \quad +
      2 \tRe \iint \ud z \ud y \, \nabla^\mu \varphi(z) 
        \expect{\left(b_{y+z}^* (b_y^{\div r})^* + (b_{y+z}^{\div r})^* b_y^* + (b_{y+z}^{\div r})^* (b_y^{\div r})^*\right) c_y \nabla^\mu c_{y+z}}_{\Gamma}.
\end{aligned}
\end{equation}
The two terms are treated similarly. We start with the first. 
Define the function 
$\phi_y(z) = \expect{\left(b_{y+z}^* (b_y^{\div r})^* + (b_{y+z}^{\div r})^* b_y^* + (b_{y+z}^{\div r})^* (b_y^{\div r})^*\right) c_y \nabla^\mu c_{y+tz}}_{\Gamma}$. Note that it vanishes at $z=0$, and hence 
\begin{equation*}
\phi_y(z) = z^\nu \int_0^1 \ud s \, \nabla^\nu \phi_y(sz).
\end{equation*}
The derivative hits either a factor $\nabla c$, $(b^{\div r})^*$ or $b^*$. 
In the case where the derivative hits either $\nabla c$ or a factor $(b^{\div r})^*$ we bound the factor $\nabla c$ and a factor $(b^{\div r})^*$ in norm 
(in the term with two $(b^{\div r})^*$'s we keep one factor $(b^{\div r})^*$ without a derivative). 
These terms are then bounded by 
\begin{equation*}
\begin{aligned}
& (k_F^\kappa)^{5}
  \iint \ud z \ud y \, |z|^2 |\Delta \varphi(z)| \left(\norm{\Gamma^{1/2} b^*_y}_{\mathfrak{S}_2} + \norm{\Gamma^{1/2} (b_y^{\div r})^*}_{\mathfrak{S}_2}\right)
  \norm{c_y \Gamma^{1/2}}_{\mathfrak{S}_2}
 \\ & \quad 
  \leq C \norm{ |\cdot|^2 \Delta \varphi}_{L^1} (k_F^\kappa)^{5} \expect{\mcN_{Q}}_{\Gamma}^{1/2} \expect{\mcN}_\Gamma^{1/2}.
\end{aligned}
\end{equation*}
For the terms where the derivative hits a factor $b^*$ we again bound the factors $\nabla c$ and $(b^{\div r})^*$ in norm. 
These terms are then bounded by 
\begin{equation*}
(k_F^\kappa)^{4}
  \iint \ud z \ud y \, |z|^2 |\Delta \varphi(z)| \norm{\Gamma^{1/2} \nabla b^*_y}_{\mathfrak{S}_2} \norm{c_y \Gamma^{1/2}}_{\mathfrak{S}_2}
  \leq C \norm{ |\cdot|^2 \Delta \varphi}_{L^1} (k_F^\kappa)^{4} \expect{\mcK_{Q}}_{\Gamma}^{1/2} \expect{\mcN}_\Gamma^{1/2}.
\end{equation*}
The second term in \eqref{eqn.H-r.decompose} is treated analogously, only the bound has a factor $\norm{|\cdot|\nabla \varphi}_{L^1}$ instead of $\norm{|\cdot|^2\Delta\varphi}_{L^1}$.
Together with the bounds of \Cref{lem.prop.phi.scattering.fun} this completes the proof of \eqref{eqn.lem.bdd.H-r}.

\subsection{Taylor expansion errors}
According to \eqref{eqn.[H0.B].decompose}, \eqref{eqn.[VQ.B].decompose} and \eqref{eqbe} 
the term $\mcQ_{\textnormal{Taylor}}$ is given by 
\begin{equation}\label{eqn.formula.T.Taylor}
\mcQ_{\textnormal{Taylor}}
=
\iint\ud x \ud y \,  F^{\mu\nu}(x-y) b_x^* b_y^* c_y \int_0^1 \ud t \, (1-t) \nabla^\mu \nabla^\nu c_{x + t(y-x)}  + \hc,
\end{equation}
where 
\begin{equation*}
F^{\mu\nu}(x) 
= - \left[x^\mu x^\nu \Delta\varphi(x) + \frac{1}{2}  x^\mu x^\nu V(x) (1 - \varphi(x))\right]
= 2 x^\mu \nabla^\nu \varphi(x) - x^\mu \mcE_\varphi^\nu(x)\,,
\end{equation*}
using \eqref{eqn.defn.mcE_varphi.scat.eqn}. 
To bound this we write 
\begin{equation*}
b_x^* = (b_x^<)^* + (b_x^>)^*, 
\qquad 
b_x^> = L^{-3/2} \sum_k \hat Q^>(k) e^{ikx} a_k,
\qquad 
\hat Q^>(k) = \chi_{(|k| > k_F^{\delta^{-1}\kappa})}
\end{equation*}
for some $0 < \delta < 1$, 
with $k_F^{\delta^{-1}\kappa}$ defined  in \eqref{eqn.def.kF.delta}, 
i.e. as $k_F^\kappa$ only with $\kappa$  replaced by $\delta^{-1}\kappa$.
In the term with $(b_x^>)^*$ we shall use \Cref{lem.bdd.b(F)}. To do this we first write $b_y^* = a_y^* - c_y^*$, and 
 rewrite $\mcQ_{\textnormal{Taylor}}$ as 
\begin{equation*}
\begin{aligned}
\mcQ_{\textnormal{Taylor}}
  & = \iint\ud x \ud y \,  F^{\mu\nu}(x-y) 
  \left[(b_x^{<})^* b_y^*
      + (b_x^>)^* a_y^*
      - (b_x^>)^* c_y^*\right]
  c_y \int_0^1 \ud t \, (1-t) \nabla^\mu \nabla^\nu c_{x + t(y-x)}  
  \\ & \quad + \hc
  \\ & = : \mcQ_{\textnormal{Taylor}}^< + \mcQ_{\textnormal{Taylor}}^> + \mcQ_{\textnormal{Taylor}}^c.
\end{aligned}
\end{equation*}

First, for the term $\mcQ_{\textnormal{Taylor}}^<$ containing $(b_x^<)^*$, we note that $\norm{(b_x^<)^*}\leq C (k_F^{\delta^{-1}\kappa})^{3/2}$. 
Thus, bounding also $\nabla^2 c$ in norm we have 
\begin{equation*}
\begin{aligned}
\abs{\expect{\mcQ_{\textnormal{Taylor}}^<}_\Gamma}
  & \leq C \iint \ud x \ud y \, \abs{F(x-y)} \norm{\Gamma^{1/2} b_y^*}_{\mathfrak{S}_2}  (k_F^{\delta^{-1}\kappa})^{3/2} (k_F^\kappa)^{3/2+2} \norm{c_y \Gamma^{1/2}}_{\mathfrak{S}_2}
  \\ & \leq C \norm{F}_{L^1} (k_F^{\delta^{-1}\kappa})^{3/2} (k_F^\kappa)^{3/2+2}  \expect{\mcN_{Q}}_\Gamma^{1/2} \expect{\mcN}_{\Gamma}^{1/2}.
\end{aligned}
\end{equation*}
Next, for the term $\mcQ_{\textnormal{Taylor}}^c$ containing $c_y^*$, we bound $c_y^*$ and $\nabla^2 c$ in norm. Then this term is bounded by 
\begin{equation*}
\begin{aligned}
\abs{\expect{\mcQ_{\textnormal{Taylor}}^c}_\Gamma}
  & \leq C \iint \ud x \ud y \, \abs{F(x-y)} \norm{\Gamma^{1/2} (b_x^>)^*}_{\mathfrak{S}_2}  (k_F^\kappa)^{5} \norm{c_y \Gamma^{1/2}}_{\mathfrak{S}_2}
  \\ & \leq C \norm{F}_{L^1} (k_F^\kappa)^{5}  \expect{\mcN_{Q^>}}_\Gamma^{1/2} \expect{\mcN}_{\Gamma}^{1/2}.
\end{aligned}
\end{equation*}
Finally, for the term $\mcQ_{\textnormal{Taylor}}^>$ we use \Cref{lem.bdd.b(F)} to bound the $y$-integral.
This term is then bounded by 
\begin{equation*}
\begin{aligned}
\abs{\expect{\mcQ_{\textnormal{Taylor}}^>}_\Gamma}
  & \leq 2
  \int \ud x \int_0^1\ud t \norm{\Gamma^{1/2} (b_x^>)^*}_{\mathfrak{S}_2} \norm{\int\ud y\, F^{\mu\nu}(x-y) a_y^* c_y \nabla^\mu \nabla^\nu c_{x + t(y-x)}}
  \\ & \leq C L^{3/2} \norm{F}_{L^2} (k_F^\kappa)^{5} \expect{\mcN_{Q^>}}_{\Gamma}^{1/2}.
\end{aligned}
\end{equation*}
Combining the bounds for $\mcQ_{\textnormal{Taylor}}^<$, $\mcQ_{\textnormal{Taylor}}^>$ and $\mcQ_{\textnormal{Taylor}}^c$,
using \Cref{lem.prop.phi.scattering.fun} and noting that $\mcN_{Q^>}\leq \mcN_Q$ and $k_F^\kappa \leq k_F^{\delta^{-1}\kappa}$
we thus obtain the bound \eqref{eqn.lem.bdd.Qscat.QTaylor} for the expectation value of $\mcQ_{\textnormal{Taylor}}$.

\subsection{Scattering equation cancellation}
We can bound $\mcQ_{\textnormal{scat}}$ analogously to $\mcQ_{\textnormal{Taylor}}$. 
The only difference is that $\nabla^2 c$ is replaced by $\nabla c$, hence the bounds have one fewer power of $k_F^\kappa$, 
and $F$ is replaced by $\mcE_\varphi$, leading to an additional factor $k_F^\kappa$. Hence we arrive at the same bound.

\subsection{Error terms from \texorpdfstring{$[\mcV_Q,\mcB]$}{[VQ,B]}}
Recall the definition of  $\mcQ_{[\mcV_Q,\mcB]}$ in \eqref{eqn.Q.[V.B].define}.
We split the operator into $4$ terms and bound each separately: 
\begin{enumerate}[1.]
\item The quartic term with one $\delta$ and one $P^r$, 
\item The quartic term with two $P^r$'s, 
\item The sextic term with $\delta$, and 
\item The sextic term with $P^r$.
\end{enumerate}

\subsubsection{Quartic term with one $\delta$ and one $P^r$}
This term is of the form 
\begin{equation*}
\mcA_{4,\delta} = \iiint \ud x \ud y \ud z' \, V(x-y) \varphi(x-z') P^r(y-z') b_x^* b_y^* c_{z'} c_x
+\hc.
\end{equation*}
Taylor expanding $c_{z'}$ around $z'=x$ as in \eqref{eqn.Taylor.cy.1st.order} and bounding in norm $\norm{\nabla c}\leq C (k_F^\kappa)^{3/2+1}$ 
and $\norm{P^r}_{L^\infty} \leq C (k_F^\kappa)^{3}$ we obtain the bound 
\begin{align}
\nonumber
\abs{\expect{\mcA_{4,\delta}}_\Gamma} 
  & \leq C (k_F^\kappa)^{3/2+4} \norm{|\cdot|\varphi}_{L^1} \iint V(x-y) \norm{\Gamma^{1/2} b_x^* b_y^*}_{\mathfrak{S}_2} \norm{c_x \Gamma^{1/2}}_{\mathfrak{S}_2} \ud x \ud y 
  \\ & \leq C (k_F^\kappa)^{3/2+4} \norm{|\cdot|\varphi}_{L^1} \norm{V}_{L^1}^{1/2} \expect{\mcV_Q}_\Gamma^{1/2} \expect{\mcN}_{\Gamma}^{1/2}
  \leq C a^{3+1/2} (k_F^\kappa)^{3/2+3}  \expect{\mcV_Q}_\Gamma^{1/2} \expect{\mcN}_{\Gamma}^{1/2}
  \label{eqn.bdd.quartic.delta.[VQ.B]}
\end{align}
using \Cref{lem.prop.phi.scattering.fun} in the last step.

\subsubsection{Quartic term with two $P^r\!$'s}
This term is of the form 
\begin{equation*}
\mcA_{4,P}
  = - \frac{1}{2}\iiiint \ud x \ud y \ud z \ud z' \,
    V(x-y) \varphi(z-z') P^r(x-z)P^r(y-z') b_x^* b_y^* c_{z'} c_z
    +\hc.
\end{equation*}
As above we Taylor expand $c_{z'}$ and bound $\nabla c$ in norm. 
By Cauchy--Schwarz 
\begin{align}
\nonumber
& \abs{\expect{\mcA_{4,P}}_\Gamma}
  \\ & \quad  \leq C (k_F^\kappa)^{3/2+1} \iiiint \ud x \ud y \ud z \ud z' \,
    V(x-y) |z-z'| \varphi(z-z') |P^r(x-z)||P^r(y-z') | 
  \nonumber 
  \\ & \qquad \times 
  \norm{\Gamma^{1/2} b_x^* b_y^*}_{\mathfrak{S}_2}  \norm{c_z \Gamma^{1/2}}_{\mathfrak{S}_2}
  \nonumber 
  \\ & \quad \leq 
    C (k_F^\kappa)^{3/2+1} 
    \left[
    \iiiint \ud x \ud y \ud z \ud z' \,
    V(x-y) \norm{\Gamma^{1/2} b_x^* b_y^*}_{\mathfrak{S}_2}^2 |z-z'| \varphi(z-z') |P^r(x-z)|^2
    \right]^{1/2}    
  \nonumber 
  \\ & \qquad \times 
    \left[
      \iiiint \ud x \ud y \ud z \ud z' \,
      V(x-y) |z-z'| \varphi(z-z') |P^r(y-z') |^2   \norm{c_z \Gamma^{1/2}}_{\mathfrak{S}_2}^2
    \right]^{1/2}
  \nonumber
  \\ & \quad \leq 
    C (k_F^\kappa)^{3/2+1} \norm{|\cdot|\varphi}_{L^1} \norm{V}_{L^1}^{1/2} \norm{P^r}_{L^2}^2 \expect{\mcV_Q}_\Gamma^{1/2} \expect{\mcN}_\Gamma^{1/2}
  \nonumber
  \\ & \quad \leq 
    C a^{3+1/2} (k_F^\kappa)^{3/2+3} \expect{\mcV_Q}_\Gamma^{1/2} \expect{\mcN}_\Gamma^{1/2}
     \label{eqn.bdd.quartic.P.[VQ.B]}
\end{align}
using again \Cref{lem.prop.phi.scattering.fun}.

\subsubsection{Sextic term with $\delta$}
This term is of the form 
\begin{equation*}
\mcA_{6,\delta}
  = \frac{1}{2} \iiint \ud x \ud y \ud z' \, V(x-y) \varphi(x-z') b_x^* b_y^* (b_{z'}^r)^* b_y c_{z'} c_x
  +\hc.
\end{equation*}
We use \eqref{eqn.bdd.phi*bcc} to bound the $z'$-integral. Thus, by Cauchy--Schwarz,
\begin{align}
\abs{\expect{\mcA_{6,\delta}}_\Gamma}
  & \leq 
    \iint \ud x \ud y \, V(x-y)  \norm{\Gamma^{1/2} b_x^* b_y^*}_{\mathfrak{S}_2} \norm{\int \ud z' \,  \varphi(x-z') (b_{z'}^r)^*  c_{z'} c_x} \norm{b_y \Gamma^{1/2}}_{\mathfrak{S}_2}
  \nonumber
  \\ & \leq C (k_F^{\kappa})^{3/2} (ak_F^\kappa)^{5/2} \norm{V}_{L^1}^{1/2} \expect{\mcV_Q}_{\Gamma}^{1/2} \expect{\mcN_Q}_{\Gamma}^{1/2}
  \leq C a^3 (k_F^\kappa)^{4} \expect{\mcV_Q}_{\Gamma}^{1/2} \expect{\mcN_Q}_{\Gamma}^{1/2}.
       \label{eqn.bdd.sextic.delta.[VQ.B]}
\end{align}

\subsubsection{Sextic term with $P^r$}
This term is of the form 
\begin{equation*}
\mcA_{6,P}
  = -\frac{1}{2} \iiiint \ud x \ud y \ud z \ud z' \, 
    V(x-y) \varphi(z-z')
    P^r(x-z) b_x^* b_y^* (b_{z'}^r)^* b_y c_{z'} c_z
    +\hc.
\end{equation*}
As above we use \eqref{eqn.bdd.phi*bcc} to bound the $z'$-integral. The $z$-integral afterwards is then bounded by $\norm{P^r}_{L^1}\leq C$. 
As above then 
\begin{equation}
  \abs{\expect{\mcA_{6,P}}_\Gamma}
    \leq C a^3 (k_F^\kappa)^{4} \expect{\mcV_Q}_{\Gamma}^{1/2} \expect{\mcN_Q}_{\Gamma}^{1/2}.
         \label{eqn.bdd.sextic.P.[VQ.B]}
\end{equation}
Combining \eqref{eqn.bdd.quartic.delta.[VQ.B]}--\eqref{eqn.bdd.sextic.P.[VQ.B]} we conclude the proof of \eqref{eqn.lem.bdd.Q[VQ.B]}.

\subsection{Error terms from \texorpdfstring{$[\mcV_P-\mcW_P,\mcB]$}{[VP-WP,B]}}

Recall the definition of $\mcQ_{V\varphi}$ in \eqref{eqn.Q.Vphi.define}.
We consider the two terms separately:
\begin{enumerate}[1.]
\item The quartic term,  
\item The sextic term.
\end{enumerate}

\subsubsection{Quartic term}
This term is of the form 
\begin{equation*}
\mcA_4 = 
- \frac{1}{2} \iiiint \ud x \ud y \ud z \ud z' \, V(x-y) \varphi(x-y) \varphi(z-z') 
      P(x-z) P(y-z') (b_z^r)^* (b_{z'}^r)^*  c_y c_x + \hc.
\end{equation*}
We bound its expectation value in an arbitrary state $\Gamma$ as 
\begin{equation*}
\begin{aligned}
\abs{\expect{\mcA_4}_\Gamma}
  & \leq \iiint \ud x \ud y \ud z \, V(x-y) \varphi(x-y) 
      |P(x-z)|
      \norm{\Gamma^{1/2}(b_z^r)^*}_{\mathfrak{S}_2}
  \\ & \quad \times  
      \norm{ \int \ud z' \, P(y-z') \varphi(z-z') (b_{z'}^r)^*} \norm{c_y c_x \Gamma^{1/2}}_{\mathfrak{S}_2}.
\end{aligned}
\end{equation*}
To bound this we note that since $\hat Q^r \leq 1$ we have 
\begin{equation*}
\begin{aligned}
\norm{ \int \ud z' \, P(y-z') \varphi(z-z') (b_{z'}^r)^*} 
  & \leq \norm{ \int \ud z' \, P(y-z') \varphi(z-z') a_{z'}^*}
  = (P^2 * \varphi^2(y-z))^{1/2}
\end{aligned}
\end{equation*}
with $*$ denoting convolution. 
Then, noting that $\varphi \leq 1$,  we have by Cauchy--Schwarz
\begin{align}
\abs{\expect{\mcA_4}_\Gamma}
  & \leq 
  \left[
  \iiint \ud x \ud y \ud z \, V(x-y)  \abs{P(x-z)}^2 \norm{c_y c_x \Gamma^{1/2}}_{\mathfrak{S}_2}^2
  \right]^{1/2}
  \nonumber
  \\ & \quad \times 
  \left[
  \iiint \ud x \ud y \ud z \, V(x-y)  P^2 * \varphi^2(y-z) \norm{\Gamma^{1/2} (b_{z}^r)^*}_{\mathfrak{S}_2}^2 
  \right]^{1/2}
  \nonumber
  \\ & \leq 
  C (k_F^\kappa)^3 \norm{V}_{L^1}^{1/2} \norm{\varphi}_{L^2} \expect{\mcV_P}_\Gamma^{1/2} \expect{\mcN_Q}_{\Gamma}^{1/2}
  \nonumber
  \\& \leq 
  C  a^{2 + 3/2} (k_F^{\kappa})^{4 + 3/2} \expect{\mcN_Q}_{\Gamma}^{1/2} \expect{\mcN}_\Gamma^{1/2},
  \label{eqn.bdd.quartic.Vphi}
\end{align}
where we used \Cref{lem.a.priori.VP,lem.prop.phi.scattering.fun}.

\subsubsection{Sextic term}
This term is of the form 
\begin{equation*}
\begin{aligned}
\mcA_6 & = 
\iiiint \ud x \ud y \ud z \ud z' \, V(x-y) \varphi(x-y) \varphi(z-z') 
       P(x-z) (b_z^r)^* (b_{z'}^r)^*  c_{z'} c_y^* c_y c_x + \hc
       \\ & = 
       \iint \ud y \ud z \, P(y-z) 
       \left[\int \ud z' \,  \varphi(z-z') (b_z^r)^* (b_{z'}^r)^*  c_{z'}\right]
       \left[\int \ud x \, V(x-y) \varphi(x-y) c_y^* c_y c_x \right]
       + \hc.
\end{aligned}
\end{equation*}
Noting that as operators $P\leq \mathbbm{1}$ we thus have by Cauchy--Schwarz for any $\lambda > 0$
\begin{equation*}
\begin{aligned}
\pm \mcA_6 
  & \leq \lambda \int \ud z 
    \left[\int \ud z' \, \varphi(z-z') (b_z^r)^* (b_{z'}^r)^*  c_{z'}\right] 
    \left[\int \ud x \, \varphi(z-x) c_{x}^* b_{x}^r  b_z^r\right]
  \\ & \quad + \lambda^{-1} \int \ud y 
    \left[\int \ud z \, V(z-y) \varphi(z-y) c_z^* c_y^* c_y \right]
    \left[\int \ud x \, V(x-y) \varphi(x-y) c_y^* c_y c_x \right]
  \\ & =: \lambda \mcA_6^{\varphi} + \lambda^{-1} \mcA_6^{V}.
\end{aligned}
\end{equation*}
We bound $\mcA_6^\varphi$ using the bound in \eqref{eqn.bdd.phi*bc}. This yields 
\begin{equation}\label{eqn.bdd.sextic.Vphi.phi-part}
\begin{aligned}
\expect{\mcA_6^\varphi}_\Gamma
  & \leq \int \ud z \norm{\Gamma^{1/2} (b_z^r)^*}_{\mathfrak{S}_2} 
    \norm{\int \ud z' \, \varphi(z-z') (b_{z'}^r)^*  c_{z'}}
    \norm{\int \ud x \, \varphi(z-x) c_{x}^* b_{x}^r  }
    \norm{ b_z^r \Gamma^{1/2}}_{\mathfrak{S}_2}
  \\ & \leq 
    C (ak_F^\kappa)^3 \expect{\mcN_Q}_\Gamma.
\end{aligned}
\end{equation}
To bound $\mcA_6^V$ we again note that $\varphi \leq 1$. 
Then, by \Cref{lem.a.priori.VP},
\begin{equation*}
\begin{aligned}
\expect{\mcA_6^V}_\Gamma 
  & \leq \iiint \ud x \ud y \ud z 
    V(z-y)  V(x-y) 
    \norm{\Gamma^{1/2} c_z^* c_y^*}_{\mathfrak{S}_2} \norm{c_y c_y^*} \norm{c_y c_z \Gamma^{1/2}}_{\mathfrak{S}_2}  
  \\ & \leq C (k_F^\kappa)^3 \norm{V}_{L^1} \expect{\mcV_P}_\Gamma 
  \leq C a^4 (k_F^\kappa)^8 \expect{\mcN}_\Gamma. 
\end{aligned}
\end{equation*}
Choosing the optimal $\lambda$ we thus find 
\begin{equation}\label{eqn.bdd.sextic.Vphi}
\abs{\expect{\mcA_6}_\Gamma}
  \leq C a^{2+3/2} (k_F^\kappa)^{4+3/2} \expect{\mcN}_\Gamma^{1/2} \expect{\mcN_Q}_\Gamma^{1/2}.
\end{equation}
Combining \eqref{eqn.bdd.quartic.Vphi} and \eqref{eqn.bdd.sextic.Vphi} we conclude the proof of \eqref{eqn.lem.bdd.QVphi}.

\subsection{Error terms from \texorpdfstring{$[\mcV_{\textnormal{OD}}, \mcB]$}{[VOD,B]}}
Recall the definition of $\mcQ_{\textnormal{OD}}$ in \eqref{eqn.Q.OD.define}. 
Writing the $c$-commutator using \eqref{eqn.calc.[cc.cc]} we have four types of terms, which we treat separately. 

\begin{enumerate}[1.]
\item The quartic term with only $c$'s, 
\item The sextic term with four $c$'s and two $b$'s, 
\item The quartic term with only $b$'s, and 
\item The sextic term with four $b$'s and two $c$'s.
\end{enumerate}

\subsubsection{Quartic term with only $c$'s}
We have two terms, 
\begin{enumerate}[a.]
\item The term with one factor $\delta$ and one factor $P^r$, and 
\item The term with two factors $P^r$.
\end{enumerate}
They are of the form 
\begin{equation*}
\begin{aligned}
\mcA_{4,c,\delta} 
  & = \iiint \ud x \ud y \ud z' \, 
    V(x-y) \varphi(x-z') 
    P^r(y-z') c_x^* c_y^* c_{z'} c_x + \hc ,
  \\
\mcA_{4,c,P}
  & = - \frac{1}{2} \iiiint \ud x \ud y \ud z \ud z' \,
    V(x-y) \varphi(z-z') 
      P^r(x-z)P^r(y-z')  c_x^* c_y^* c_{z'} c_z + \hc.
\end{aligned}
\end{equation*}
To bound the first term we Taylor expand $c_{z'}$ around $z'=x$. Then by Cauchy--Schwarz
\begin{align}
\nonumber
\abs{\expect{\mcA_{4,c,\delta}}_{\Gamma}} 
  & \leq C (k_F^\kappa)^{3/2+1} 
    \iiint \ud x \ud y \ud z' \, 
    V(x-y) |x-z'| \varphi(x-z') 
    |P^r(y-z')| 
  \\ & \quad \times 
  \norm{\Gamma^{1/2} c_x^* c_y^*}_{\mathfrak{S}_2} \norm{c_x \Gamma^{1/2}}_{\mathfrak{S}_2}
  \nonumber
  \\ & \leq C (k_F^\kappa)^{3/2+1} \norm{V}_{L^1}^{1/2} \norm{|\cdot|\varphi}_{L^2} \norm{P^r}_{L^2} \expect{\mcV_P}_{\Gamma}^{1/2} \expect{\mcN}_{\Gamma}^{1/2}
    \leq C a^{4+1/2} (k_F^\kappa)^{6+1/2} \expect{\mcN}_\Gamma
  \label{eqn.bdd.OD.quartic.c.delta}
\end{align}
by \Cref{lem.a.priori.VP,lem.prop.phi.scattering.fun}. 
Similarly we have for the second term, again using \Cref{lem.prop.phi.scattering.fun,lem.a.priori.VP},
\begin{align}
& \abs{\expect{\mcA_{4,c,P}}}
  \nonumber
  \\ & \quad \leq 
    C (k_F^\kappa)^{3/2+1}
      \left[
        \iiiint \ud x \ud y \ud z \ud z' \, 
        V(x-y) |z-z'| \varphi(z-z') |P^r(x-z)|^2
        \norm{\Gamma^{1/2} c_x^* c_y^*}_{\mathfrak{S}_2}^2
      \right]^{1/2}
  \nonumber
  \\ & \qquad \times 
    \left[
      \iiiint \ud x \ud y \ud z \ud z' \, 
        V(x-y) |z-z'| \varphi(z-z') |P^r(y-z')|^2
        \norm{c_z \Gamma^{1/2}}_{\mathfrak{S}_2}^2  
    \right]^{1/2}
  \nonumber
  \\ & \quad \leq 
    C (k_F^\kappa)^{3/2+1} \norm{|\cdot|\varphi}_{L^1} \norm{P^r}_{L^2}^2 \norm{V}_{L^1}^{1/2} \expect{\mcV_P}_{\Gamma}^{1/2} \expect{\mcN}_{\Gamma}^{1/2}
  \leq 
    C a^{5} (k_F^\kappa)^{7}  \expect{\mcN}_{\Gamma}.
    \label{eqn.bdd.OD.quartic.c.P}
\end{align}

\subsubsection{Sextic term with four $c$'s and two $b$'s}
We have two terms, 
\begin{enumerate}[a.]
\item The term with a factor $\delta$, and 
\item The term with a factor $P^r$.
\end{enumerate}
They are of the form 
\begin{equation*}
\begin{aligned}
\mcA_{6,c,\delta}
  & = \iiint \ud x \ud y \ud z' \, 
    V(x-y) \varphi(x-z')
     c_x^* c_y^* (b_{z'}^r)^* b_y c_{z'} c_x + \hc ,
  \\
\mcA_{6,c,P}
  & = - \iiiint \ud x \ud y \ud z \ud z' \, 
    V(x-y) \varphi(z-z') 
    P^r(x-z) c_x^* c_y^* (b_{z'}^r)^* b_y c_{z'} c_z + \hc.
\end{aligned}
\end{equation*}
To bound the first term we use \eqref{eqn.bdd.phi*bcc} to bound the $z'$-integral. 
We have 
\begin{align}
\abs{\expect{\mcA_{6,c,\delta}}_\Gamma}
  & \leq 2 \iint \ud x \ud y \, 
    V(x-y) \norm{\Gamma^{1/2} c_x^* c_y^*}_{\mathfrak{S}_2} \norm{\int \ud z' \, \varphi(x-z') (b_{z'}^r)^* c_{z'} c_x } \norm{b_y \Gamma^{1/2}}_{\mathfrak{S}_2}
  \nonumber
  \\ & \leq 
    C (k_F^\kappa)^{3/2} (ak_F^\kappa)^{5/2} \norm{V}^{1/2}_{L^1} \expect{\mcV_P}_{\Gamma}^{1/2} \expect{\mcN_Q}_{\Gamma}^{1/2}
    \leq C a^{4+1/2} (k_F^\kappa)^{6+1/2} \expect{\mcN_Q}_{\Gamma}^{1/2} \expect{\mcN}_{\Gamma}^{1/2},
    \label{eqn.bdd.OD.sextic.c.delta}
\end{align}
where we used \Cref{lem.a.priori.VP}. 
The second term is bounded similarly, only the $z$-integral is  bounded by $\norm{P^r}_{L^1}\leq C$.
We conclude that 
\begin{equation}
  \abs{\expect{\mcA_{6,c,P}}_\Gamma}
    \leq C a^{4+1/2} (k_F^\kappa)^{6+1/2} \expect{\mcN_Q}_{\Gamma}^{1/2} \expect{\mcN}_{\Gamma}^{1/2}.
    \label{eqn.bdd.OD.sextic.c.P}
\end{equation}

\subsubsection{Quartic term with only $b$'s}
This term is of the form 
\begin{equation*}
\mcA_{4,b}
  = -\frac{1}{2} 
  \iiiint 
  \ud x \ud y \ud z \ud z' \, 
  V(x-y) \varphi(z-z') 
  P(x-z) P(y-z')
  (b_z^r)^* (b_{z'}^r)^* b_y b_x
  + \hc.
\end{equation*}
We bound its expectation value as 
\begin{equation*}
\begin{aligned}
\abs{\expect{\mcA_{4,b}}_\Gamma}
  & \leq \iiint 
  \ud x \ud y \ud z \, 
  V(x-y) |P(x-z)| 
  \\ & \quad \times 
  \norm{\Gamma^{1/2} (b_z^r)^*}_{\mfS_2} 
  \norm{\int \ud z' \, \varphi(z-z')P(y-z') (b_{z'}^r)^* }
  \norm{b_y b_x \Gamma^{1/2}}_{\mfS_2}
\end{aligned}
\end{equation*}
Then, since $0 \leq \hat Q^r(k) \leq 1$ we have 
\begin{equation*}
\norm{\int \ud z' \, \varphi(z-z')P(y-z') (b_{z'}^r)^* }
  \leq \norm{\int \ud z' \, \varphi(z-z')P(y-z') a_{z'}^* }
  = (\varphi^2 * P^2(y-z))^{1/2}
\end{equation*}
with $*$ denoting convolution. 
Hence, by Cauchy--Schwarz, using \Cref{lem.prop.phi.scattering.fun},
\begin{align}
\abs{\expect{\mcA_{4,b}}_\Gamma}
  & \leq 
  \left[
    \iiint 
  \ud x \ud y \ud z \, 
  V(x-y) |P(x-z)|^2  \norm{b_y b_x \Gamma^{1/2}}_{\mfS_2}^2
  \right]^{1/2}
  \nonumber
  \\ & \quad \times 
  \left[
    \iiint 
  \ud x \ud y \ud z \, 
  V(x-y) 
  \varphi^2 * P^2(y-z)
  \norm{\Gamma^{1/2} (b_z^r)^*}_{\mfS_2}^2 
  \right]^{1/2}
  \nonumber
  \\ & \leq 
    C \norm{V}_{L^1}^{1/2} \norm{P}_{L^2}^2 \norm{\varphi}_{L^2} \expect{\mcV_Q}_{\Gamma}^{1/2} \expect{\mcN_Q}_{\Gamma}^{1/2}
    \leq C a^2 (k_F^\kappa)^3 \expect{\mcV_Q}_{\Gamma}^{1/2} \expect{\mcN_Q}_{\Gamma}^{1/2}.
    \label{eqn.bdd.OD.quartic.b}
\end{align}

\subsubsection{Sextic term with four $b$'s and two $c$'s}
This term is of the form 
\begin{equation*}
\begin{aligned}
\mcA_{6,b}
  & = \iiiint 
  \ud x \ud y \ud z \ud z' \, 
  V(x-y) \varphi(z-z') 
  P(x-z) 
  (b_z^r)^* (b_{z'}^r)^* c_{z'} c_y^* b_y b_x + \hc
  \\ & 
  = \iint 
  \ud x \ud z \, 
  P(x-z) 
  \left[\int \ud z' \, 
  \varphi(z-z')
  (b_z^r)^* (b_{z'}^r)^* c_{z'}
  \right]
  \left[\int \ud y \, 
  V(x-y)  
   c_y^* b_y b_x 
  \right]
   + \hc.
\end{aligned}
\end{equation*}
To bound it we note that as operators $0\leq P \leq \mathbbm{1}$. Thus by Cauchy--Schwarz we have for any $\lambda > 0$ the bound 
\begin{equation*}
\begin{aligned}
\pm \mcA_{6,b} 
  & \leq \lambda \int \ud z 
    \left[\int \ud z' \, \varphi(z-z') (b_z^r)^* (b_{z'}^r)^*  c_{z'}\right] 
    \left[\int \ud x \, \varphi(z-x) c_{x}^* b_{x}^r  b_z^r\right]
  \\ & \quad + \lambda^{-1} \int \ud y 
    \left[\int \ud z \, V(z-y)  b_z^* b_y^* c_y \right]
    \left[\int \ud x \, V(x-y)  c_y^* b_y b_x \right]
  \\ & =: \lambda \mcA_{6,b}^{\varphi} + \lambda^{-1} \mcA_{6,b}^{V}.
\end{aligned}
\end{equation*}
Note that $\mcA_{6,b}^\varphi = \mcA_6^\varphi$ from the bound of $\mcQ_{V\varphi}$, see \eqref{eqn.bdd.sextic.Vphi.phi-part}.
In particular,  
\begin{equation*}
\begin{aligned}
\expect{\mcA_6^\varphi}_\Gamma
   & \leq 
    C (ak_F^\kappa)^3 \expect{\mcN_Q}_\Gamma.
\end{aligned}
\end{equation*}
With  Cauchy--Schwarz we bound $\mcA_{6,b}^V$ as 
\begin{equation*}
\begin{aligned}
\expect{\mcA_{6,b}^V}_{\Gamma}
  & \leq C (k_F^\kappa)^3 \iiint \ud x \ud y \ud z \, 
    V(z-y) V(x-y)
    \norm{\Gamma^{1/2} b_z^* b_y^*}_{\mfS_2} 
    \norm{b_y b_x \Gamma^{1/2}}_{\mfS_2}
  \\ & \leq C \norm{V}_{L^1} (k_F^\kappa)^3 \expect{\mcV_Q}_{\Gamma}
  \leq C a (k_F^\kappa)^3  \expect{\mcV_Q}_{\Gamma}.
\end{aligned}
\end{equation*}
Choosing the optimal $\lambda$ we find 
\begin{equation}
  \abs{\expect{\mcA_{6,b}}_\Gamma}
    \leq C a^2 (k_F^\kappa)^{3} \expect{\mcV_Q}_{\Gamma}^{1/2} \expect{\mcN_Q}_{\Gamma}^{1/2}.
    \label{eqn.bdd.OD.sextic.b}
\end{equation}
Combining \eqref{eqn.bdd.OD.quartic.c.delta}--\eqref{eqn.bdd.OD.sextic.b} we conclude the proof of \eqref{eqn.lem.bdd.QOD}. Hence the proof of \Cref{lem.bdd.list} is complete.

\section{Propagating a priori bounds --- Bounding \texorpdfstring{$\mcE_{\textnormal{scat}}, \mcE_{\textnormal{OD}}, \mcE_{V\varphi}$}{Escat, EOD, EVphi}}
\label{sec.propagate.a.priori}
To bound the error terms $\mcE_{\textnormal{scat}}, \mcE_{\textnormal{OD}}$ and $\mcE_{V\varphi}$  in \eqref{eqn.def.mcE.scat}--\eqref{eqn.def.mcE.Vphi}, 
we use \Cref{lem.bdd.list}. Thus, we need bounds on 
\begin{equation*}
\expect{\mcN}_{\Gamma^\lambda}, \qquad \expect{\mcN_Q}_{\Gamma^\lambda}, \qquad \expect{\mcK_Q}_{\Gamma^\lambda},
\qquad \expect{\mcV_Q}_{\Gamma^\lambda}
\end{equation*}
for  $0\leq \lambda \leq 1$ (recall the definition of $\Gamma^\lambda$ in \eqref{eqn.def.Gamma.lambda}).
For an approximate Gibbs state $\Gamma$ the states $\Gamma^\lambda$ are not necessarily approximate Gibbs states, 
so we cannot directly apply the bounds of \eqref{eqn.a.priori.NQ.KQ} and \Cref{lem.a.priori.VQ}.
Instead, we obtain the desired bounds by propagating the corresponding bounds for the state $\Gamma$.
This is similar to what is done in \cite{Falconi.Giacomelli.ea.2021,Giacomelli.2023,Lauritsen.Seiringer.2024a}.

\subsection{General propagation estimates}
To start, we consider propagation estimates for a general state $\Gamma$.
First, since $\mcB$ is particle number preserving we have 
\begin{equation*}
\expect{\mcN}_{\Gamma^\lambda} = \expect{\mcN}_\Gamma
\end{equation*}
for any state $\Gamma$. 
Next, we have

\begin{lemma}[{Propagation of $\mcN_Q$}]\label{lem.propagate.NQ}
Let $\Gamma$ be any state and let $\Gamma^\lambda$ be defined as in \eqref{eqn.def.Gamma.lambda}.
Then,
\begin{equation*}
\expect{\mcN_Q}_{\Gamma^{\lambda}} \leq C \expect{\mcN_Q}_{\Gamma^{\lambda'}} + C L^3 (k_F^\kappa)^3 (ak_F^\kappa)^5
\end{equation*}
for any $0 \leq \lambda, \lambda' \leq 1$.
\end{lemma}

\begin{proof}
This is an application of Gr\"onwall's lemma. We calculate 
\begin{equation*}
\dd{\lambda} \expect{\mcN_Q}_{\Gamma^\lambda} = - \expect{[\mcN_Q,\mcB]}_{\Gamma^\lambda}
  = 2 \tRe \iint \ud z \ud z' \, \varphi(z-z') \expect{(b_z^r)^* (b_{z'}^r)^* c_{z'} c_z}_{\Gamma^\lambda}.
\end{equation*}
Using \eqref{eqn.bdd.phi*bcc} to bound the $z'$-integral we find 
\begin{equation*}
\abs{\dd{\lambda} \expect{\mcN_Q}_{\Gamma^\lambda}}
  \leq C L^{3/2} (k_F^\kappa)^{3/2}(ak_F^\kappa)^{5/2} \expect{\mcN_Q}_{\Gamma^\lambda}^{1/2}.
\end{equation*}
By Gr\"onwall's lemma we easily obtain the desired bound. 
\end{proof}

\begin{lemma}[{Propagation of $\mcK_Q$}]\label{lem.propagate.KQ}
Let $\Gamma$ be any state and let $\Gamma^\lambda$ be defined as in \eqref{eqn.def.Gamma.lambda}.
Then,
\begin{equation*}
\expect{\mcK_Q}_{\Gamma^{\lambda}}
  \leq C \expect{\mcK_Q}_{\Gamma^{\lambda'}}
    + C a^3 (k_F^\kappa)^8 \left[L^3 + a^3 \abs{\log a k_F^\kappa}^2 \expect{\mcN}_\Gamma\right]
\end{equation*}
for any $0\leq \lambda,\lambda'\leq 1$.
\end{lemma}
\begin{proof}
This is again an application of Gr\"onwall's lemma. We calculate 
\begin{equation*}
\dd{\lambda} \expect{\mcK_Q}_{\Gamma^\lambda} = - \expect{[\mcK_Q,\mcB]}_{\Gamma^\lambda}
  = \tRe  \frac{1}{L^3} \sum_{k,k',p} (|k+p|^2 + |k'-p|^2) \hat\varphi(p) \expect{(b_{k+p}^r)^* (b_{k'-p}^r)^* c_{k'} c_k}_{\Gamma^\lambda}.
\end{equation*}
Noting the symmetry of interchanging $k$ with $k'$ and $p$ with $-p$ we write this in configuration space as 
\begin{equation*}
\dd{\lambda} \expect{\mcK_Q}_{\Gamma^\lambda} = - 2 \tRe \iint \varphi(x-y) \expect{(\Delta b_x^r)^* (b_y^r)^* c_y c_x}_{\Gamma^\lambda} \ud x \ud y. 
\end{equation*}
Taylor expanding $c_x$ around $x=y$ and integrating by parts once in the $x$-integral  we have (with the derivatives being with respect to $x$)
\begin{equation*}
\dd{\lambda} \expect{\mcK_Q}_{\Gamma^\lambda} 
  = 2\tRe \iint \expect{(\nabla^\nu  b_x^r)^* (b_y^r)^* c_y \int_0^1\ud t \, \nabla^\nu \left[(x-y)^\mu \varphi(x-y) \nabla^\mu c_{y+t(x-y)}\right]}_{\Gamma^\lambda} \ud x \ud y 
\end{equation*}
We bound the $y$-integral using \Cref{lem.bdd.b(F)}. Write first $(b_y^r)^* = a_y^* - (c_y^r)^*$. 
We bound the term with $(c_y^r)^*$ by, using \Cref{lem.prop.phi.scattering.fun}, 
\begin{equation*}
\begin{aligned}
& C (k_F^\kappa)^{4}
  \iint \ud x \ud y 
  \norm{(\Gamma^\lambda)^{1/2} (\nabla^\nu  b_x^r)^*}_{\mfS_2} 
  \left[(\varphi + |\cdot||\nabla\varphi|)(x-y) \right] 
  \norm{c_y (\Gamma^\lambda)^{1/2}}_{\mfS_2}
\\ & \quad 
  \leq C (k_F^\kappa)^4 
  \left[\norm{\varphi}_{L^1} + \norm{|\cdot|\nabla \varphi}_{L^1}\right] 
  \expect{\mcK_Q}_{\Gamma^\lambda}^{1/2} \expect{\mcN}_{\Gamma^\lambda}^{1/2}
  \leq a^3 (k_F^\kappa)^4 \abs{\log ak_F^\kappa} \expect{\mcK_Q}_{\Gamma^\lambda}^{1/2} \expect{\mcN}_{\Gamma^\lambda}^{1/2}.
\end{aligned}
\end{equation*}
Using \Cref{lem.bdd.b(F),lem.prop.phi.scattering.fun} the term with $a_y^*$ is bounded by 
\begin{equation*}
\begin{aligned}
& \int_0^1\ud t 
\int \ud x \norm{(\Gamma^\lambda)^{1/2}  (\nabla^\nu  b_x^r)^*}_{\mfS_2}
\Biggl\Vert 
  \int \ud y \Bigl[ \left(\delta^{\mu\nu} \varphi(x-y) + (x-y)^\mu \nabla^\nu \varphi(x-y)\right) a_y^* c_y \nabla^\mu c_{y+t(x-y)} 
\\ & \qquad 
  + t (x-y)^\mu \varphi(x-y) a_y^* c_y \nabla^\mu \nabla^\nu c_{y+t(x-y)}
\Bigr]\Biggr\Vert 
\\ & \quad 
  \leq C L^{3/2} (k_F^\kappa)^{4} \left(\norm{\varphi}_{L^2} + \norm{|\cdot|\nabla \varphi}_{L^2} + k_F^\kappa \norm{|\cdot|\varphi}_{L^2}\right)
  \expect{\mcK_Q}_{\Gamma^\lambda}^{1/2}
\\ & \quad \leq 
  C L^{3/2} a^{3/2} (k_F^\kappa)^{4} \expect{\mcK_Q}_{\Gamma^\lambda}^{1/2}.
\end{aligned} 
\end{equation*}
Using that $\expect{\mcN}_{\Gamma^\lambda} = \expect{\mcN}_{\Gamma}$ we obtain the bound 
\begin{equation*}
\abs{\dd{\lambda}\expect{\mcK_Q}_{\Gamma^\lambda}} 
  \leq C a^{3/2} (k_F^\kappa)^{4} \left[L^{3/2} + a^{3/2} \abs{\log ak_F^\kappa} \expect{\mcN}_\Gamma^{1/2} \right]  \expect{\mcK_Q}_{\Gamma^\lambda}^{1/2}.
\end{equation*}
By Gr\"onwall's lemma this proves the lemma.
\end{proof}

\begin{lemma}[{Propagation of $\mcV_Q$}]\label{lem.propagate.VQ}
Let $\Gamma$ be any state and let $\Gamma^\lambda$ be defined as in \eqref{eqn.def.Gamma.lambda}.
Then,
\begin{equation*}
\expect{\mcV_Q}_{\Gamma^{\lambda}}
  \leq C \expect{\mcV_Q}_{\Gamma^{\lambda'}}
    + C a^3 (k_F^\kappa)^5 
    \expect{\mcN}_\Gamma
\end{equation*}
for any $0\leq \lambda, \lambda' \leq 1$.
\end{lemma}

\begin{proof}
We again employ Gr\"onwall's lemma. Recalling \eqref{eqn.[VQ.B].calc} we have  
\begin{equation*}
\dd{\lambda}\expect{\mcV_Q}_{\Gamma^\lambda} = - \expect{[\mcV_Q,\mcB]}_{\Gamma^\lambda}
  = \tRe \iint \ud x \ud y \, V(x-y) \varphi(x-y) \expect{b_x^* b_y^* c_y c_x}_{\Gamma^\lambda} - \expect{\mcQ_{[\mcV_Q,\mcB]}}_{\Gamma^\lambda}.
\end{equation*}
The second term is bounded in \Cref{lem.bdd.list}. 
Using that $\varphi \leq 1$ we bound the first term with the aid of Cauchy--Schwarz by
\begin{equation*}
\begin{aligned}
& \left[\iint V(x-y)\varphi(x-y) \expect{b_x^* b_y^* b_y b_x}_{\Gamma^{\lambda}}\right]^{1/2}
\left[\iint V(x-y)\varphi(x-y) \expect{c_x^* c_y^* c_y c_x}_{\Gamma^{\lambda}}\right]^{1/2}
\\ & \quad \leq \expect{\mcV_Q}_{\Gamma^{\lambda} }^{1/2} \expect{\mcV_P}_{\Gamma^\lambda}^{1/2}.
\end{aligned}
\end{equation*}
Using \Cref{lem.bdd.list,lem.a.priori.VP} and recalling that $\expect{\mcN}_{\Gamma^\lambda} = \expect{\mcN}_\Gamma$ we obtain the bound
\begin{equation*}
\begin{aligned}
\abs{\dd{\lambda}\expect{\mcV_Q}_{\Gamma^\lambda} }
  & \leq C a^3 (k_F^\kappa)^4 \left[\expect{\mcN_Q}_{\Gamma^\lambda}^{1/2}  + (ak_F^\kappa)^{1/2} \expect{\mcN}_{\Gamma}^{1/2}\right] \expect{\mcV_Q}_{\Gamma^\lambda}^{1/2}
    + C a^{3/2} (k_F^\kappa)^{5/2} \expect{\mcN}_{\Gamma}^{1/2} \expect{\mcV_Q}_{\Gamma^\lambda}^{1/2}
  \\ & \leq 
    C a^{3/2} (k_F^\kappa)^{5/2} 
    \expect{\mcN}_{\Gamma}^{1/2} \expect{\mcV_Q}_{\Gamma^\lambda}^{1/2}.
\end{aligned}
\end{equation*}
By Gr\"onwall's lemma this yields  the desired bound.
\end{proof}

\subsection{Bounding \texorpdfstring{$\mcE_{\textnormal{scat}}, \mcE_{\textnormal{OD}}, \mcE_{V\varphi}$}{Escat, EOD, EVphi}}
We now apply the propagation estimates above to approximate Gibbs states $\Gamma$. 

\begin{lemma}\label{lem.bdd.propagate.a.priori}
Let $\Gamma$ be an approximate Gibbs state. Then, 
(for any $0\leq \lambda \leq 1$)
\begin{align}
  \expect{\mcN_Q}_{\Gamma^\lambda} 
    & \leq C L^3 \zeta \kappa^{-1}    a^3 \rho_0^2
    + L^3 \mfe_L,
  \label{eqn.bdd.propagate.a.priori.NQ}
 \\
  \expect{\mcK_Q}_{\Gamma^\lambda} 
    & 
    \leq C L^3 \zeta^{-4} \kappa^4 a^3  \rho_0^{8/3}
    + L^3\mfe_L,
    \label{eqn.bdd.propagate.a.priori.KQ}
\\
  \expect{\mcV_Q}_{\Gamma^\lambda} 
    & 
    \leq C L^3 \zeta^{-5/2}\kappa^{5/2} a^3 \rho_0^{8/3}
    + L^3\mfe_L.
    \label{eqn.bdd.propagate.a.priori.VQ}
\end{align}
\end{lemma}

\begin{proof} 
Apply \Cref{lem.propagate.NQ,lem.propagate.KQ,lem.propagate.VQ} for $\lambda'=0$, 
use the bounds in \eqref{eqn.a.priori.NQ.KQ} and \Cref{lem.a.priori.VQ}
and recall the bound in \eqref{eqn.asymp.beta.mu.kappa} for $k_F^\kappa$.
The statement of the lemma follows.
\end{proof}

The bounds of \Cref{lem.bdd.propagate.a.priori} can be used to prove \Cref{prop.mcE.scat,prop.mcE.OD,prop.mcE.Vphi}.

\begin{proof}[{Proof of \Cref{prop.mcE.scat}}]
Recall the definition of $\mcE_{\textnormal{scat}}$ in \eqref{eqn.def.mcE.scat} and the formula \eqref{eqn.decompose.scat.eqn.operators}. 
Combining \Cref{lem.bdd.list} and the bounds in
\eqref{eqn.asymp.beta.mu.kappa}, \eqref{eqn.bdd.propagate.a.priori.NQ}, \eqref{eqn.bdd.propagate.a.priori.KQ} 
and \eqref{eqn.bdd.propagate.a.priori.VQ} 
we have for any $0 < \delta \ll 1$
(using that $\mcN_{Q^>}$ satisfies the same bounds as $\mcN_Q$ only with $\kappa$ replaced by $\delta^{-1}\kappa$, 
recall \Cref{rmk.bdd.NQ>})
\begin{equation*}
\begin{aligned}
\abs{\expect{\mcH_{0;\mcB}^{\div r}}_{\Gamma^\lambda}}
  & \leq C L^3  a^{3} \rho_0^{8/3} \zeta^{-4}\kappa^{4} (a^3\rho_0)^{1/2} \abs{\log a^3\rho_0}
  + L^3\mfe_L,
\\ 
\abs{\expect{\mcQ_{\textnormal{Taylor}}}_{\Gamma^\lambda}} + \abs{\expect{\mcQ_{\textnormal{scat}}}_{\Gamma^\lambda}}
  & \leq 
  C L^3 a^3 \rho_0^{8/3} \left[
    \zeta^{-2}\kappa^2 \delta^{-3/4} (a^3\rho_0)^{1/2} \abs{\log a^3\rho_0}
  + \zeta^{-2} \kappa^{2} \delta^{1/2}
  \right]
  + L^3\mfe_L,
\\
\abs{\expect{\mcQ_{[\mcV_Q,\mcB]}}_{\Gamma^\lambda}}
  & \leq 
  C L^3 a^3 \rho_0^{8/3} \left[ \zeta^{-11/4}\kappa^{11/4} a^3\rho_0 + \zeta^{-4} \kappa^{4} (a^3 \rho_0)^{2/3}\right]
  + L^3\mfe_L.
\end{aligned}
\end{equation*}
Choosing the optimal $\delta = (a^3 \rho_0)^{2/5} \abs{\log a^3\rho_0}^{4/5}$ 
we conclude the desired statement for $\alpha > 0$ small enough. (Recall that $\kappa = \zeta (a^3\rho_0)^{-\alpha}$ as per \eqref{eqn.choice.kappa}.)
\end{proof}

\begin{proof}[{Proof of \Cref{prop.mcE.OD}}]  
Recall the definition of $\mcE_{\textnormal{OD}}$ in \eqref{eqn.def.mcE.OD}. 
Combining \Cref{lem.bdd.list} and the bounds in 
\eqref{eqn.asymp.beta.mu.kappa}, \eqref{eqn.bdd.propagate.a.priori.NQ} and \eqref{eqn.bdd.propagate.a.priori.VQ} 
we have
\begin{equation*}
\abs{\expect{\mcQ_{\textnormal{OD}} }_{\Gamma^\lambda}}
  \leq 
  C L^3 a^3 \rho_0^{8/3} \left[
  \zeta^{-13/4}\kappa^{13/4} (a^3\rho_0)^{1/2}
  +
  \zeta^{-4}\kappa^4   (a^3\rho_0)^{2/3}
  \right]
  + L^3\mfe_L.
\end{equation*}
We conclude the desired bound by taking $\alpha > 0$ small enough. 
\end{proof}

\begin{proof}[{Proof of \Cref{prop.mcE.Vphi}}]
Recall the definition of $\mcE_{V\varphi}$ in \eqref{eqn.def.mcE.Vphi}. 
Combining \Cref{lem.bdd.list} and the bounds in 
\eqref{eqn.asymp.beta.mu.kappa} and \eqref{eqn.bdd.propagate.a.priori.NQ}
we have 
\begin{equation*}
\abs{\expect{\mcQ_{V\varphi}}_\Gamma}
  \leq C L^3 a^3 \rho_0^{8/3} \zeta^{-4} \kappa^{4} (a^3\rho_0)^{2/3}
  + L^3\mfe_L\,,
\end{equation*}
from which the statement follows.
\end{proof}

\section{Validity of perturbation theory --- Bounding \texorpdfstring{$\mcE_{\textnormal{pt}}$}{E1st}}\label{sec.1st.order}
In this section we give the proof of \Cref{prop.mcE.1st}. 
In broad strokes this amounts to showing the validity of first order perturbation theory in an appropriate regime. 
The main idea is to use the a priori bound \eqref{eqn.bdd.rel.entropy.a.priori} on the relative entropy to show that the expectation value of $\mcW_P = \ud \Gamma(PPWPP)$ in an approximate Gibbs state $\Gamma$ is to leading order the same as in the non-interacting state $\Gamma_0$, 
which in turn can be replaced by the $ \expect{\mcW}_{\Gamma_0}$ 
(i.e., the projection $P$ can be dropped).

The main result of this section  is the following proposition.  A key ingredient in its proof is the method of \cite{Seiringer.2006a} for obtaining correlation estimates at positive temperature (see also \cite{Seiringer.2008}, and \cite{Graf.Solovej.1994} for earlier work at zero temperature.) 

Recall the definition of $\mcE_{\textnormal{pt}}(\Gamma)$ in \eqref{eqn.def.mcE.1st}.

\begin{prop}\label{lem.general.correlation.mcE.1st}
For any state $\Gamma$, any $a \leq R < L$, any $\rho_0^{-1/3} \lesssim d \leq L$, any $ 0 < q_F < k_F^\kappa$ with $\beta(k_F^\kappa - q_F)^2 \gtrsim 1$ any $n \in \N$
and any $z \gtrsim 1$ 
\begin{equation}\label{eq:prop71}
\begin{aligned}
\mcE_{\textnormal{pt}}(\Gamma)  
  & \gtrsim
    - L^3 a^3 R^2 \rho_0^{2+4/3} 
    -  (k_F^\kappa)^7 a^3 R^2 \expect{\mcN}_{\Gamma}
    - L^3 a^3 \rho_0^{8/3} R d^{-1}
      \\ & \quad 
    - L^{9/4} 
    \left[a^6R^{-5}  \rho_0^{2+2/3} d^{-3} + a^6 \rho_0^{3+4/3}(d^{-3} + \rho_0)\right]^{1/2} 
  \Bigl[
    d^3 S(\Gamma, \Gamma_0) 
  + L^3 \beta^4 (k_F^\kappa)^7 d^{-1}
  \\ & \qquad 
  + L^3 d^3 \beta^{-1} (k_F^\kappa - q_F) e^{- \kappa + \beta(2 k_F^\kappa q_F - q_F^2)} 
  + C_n d^3 (dq_F)^{-2n} \log Z_0
  \Bigr]^{1/4}
  - L^3\mathfrak{e}_L.
\end{aligned}
\end{equation}
The constant $C_n > 0$ depends only on $n$.
\end{prop}

\begin{remark}[Interpretation of parameters]
In \Cref{lem.general.correlation.mcE.1st} appear many free parameters. They may be understood heuristically as follows:
The length $R$ is the range of a new interaction. To prove \Cref{lem.general.correlation.mcE.1st}
we use in particular an argument that can be viewed as an analogue of the ``Dyson lemma'' (see \Cref{sec.dyson.lemma}) where the interaction $W$  gets replaced 
by a weaker and  longer ranged interaction $U$ of range $R$.
The length $d$ is a localization length. To prove \Cref{lem.general.correlation.mcE.1st} we split the box $\Lambda$ 
into many smaller regions of size of order $d$ and separated by a distance of order $d$, see \Cref{sec.local.rel.entropy}. 
The momentum cut-off $q_F$ and the integer $n$ are technical parameters introduced to estimate non-leading error terms.
\end{remark}

\begin{remark}[Non-uniformity in the temperature]\label{rmk.non.unif.temp.details}
We do now know how to prove the result of \Cref{thm.main} uniformly in temperatures $T\lesssim T_F$.
This is due to the temperature dependence of the error bounds in \Cref{lem.general.correlation.mcE.1st}.
We illustrate this here.

As in \cite{Lauritsen.Seiringer.2023a} we can, for sufficiently small temperatures, i.e. large $\zeta$, 
simply compare to the zero-temperature result in \cite{Lauritsen.Seiringer.2024a}. 
For  $\zeta \gg (a^3\rho_0)^{-1/2}$ the  contribution of the main term $\psi_0(\beta,\mu)$ resulting from the temperature  is negligible compared to the leading interaction term of order $a^3 \overline{\rho}_0^{8/3}$; see \cite[Section 1.2]{Lauritsen.Seiringer.2023a} for the details. 
Thus, to get a bound uniformly in temperatures $T\lesssim T_F$ 
it suffices to consider $\zeta \lesssim (a^3\rho_0)^{-1/2}$.

Further, one might consider changing the choice of $\kappa$ in \eqref{eqn.choice.kappa} to $\kappa = \zeta^{1-\eps}(a^3\rho)^{-\alpha}$ for some $\eps \geq 0$. 
As in \eqref{eqn.asymp.beta.mu.kappa} we then have $k_F^\kappa \sim (1 + \zeta^{-1/2}\kappa^{1/2})\rho_0^{1/3}$.

Using the bounds from \eqref{eqn.asymp.beta.mu.kappa} and \eqref{eqn.bdd.rel.entropy.a.priori} 
and this new choice of $\kappa$, 
two of the error terms in \Cref{lem.general.correlation.mcE.1st} are bounded as 
\begin{equation*}
\begin{aligned}
& L^{9/4} a^3 \rho_0^{8/3}
  \left[d^3 S(\Gamma,\Gamma_0)
  +
  L^3 \beta^4 (k_F^\kappa)^7 d^{-1}
  \right]^{1/4}
  \\ & \quad 
  \leq C L^3 a^3 \rho_0^{8/3} 
  \left[
    \zeta^{1/4} (a^3\rho_0)^{1/4} (d^3\rho_0)^{1/4}
    +
    \zeta \left[1 + \zeta^{-7/8}\kappa^{7/8}\right] (d^3\rho_0)^{-1/12}
  \right].
\end{aligned} 
\end{equation*}
For $\zeta \sim (a^3\rho_0)^{-1/2}$
these terms are of order 
\begin{equation*}
L^3 a^3 \rho_0^{8/3}
  \left[  
    (a^3\rho_0)^{1/8} (d^3\rho_0)^{1/4}
    +
    (a^3\rho_0)^{-1/2}\left[1 + (a^3\rho_0)^{7\eps/16 - 7\alpha/8}\right] (d^3\rho_0)^{-1/12}
  \right].
\end{equation*}
It is impossible to choose $d$ such that both of these terms are negligible compared to the leading term of order $L^3 a^3 \rho_0^{8/3}$. 
\end{remark}

With \Cref{lem.general.correlation.mcE.1st} at hand  we can give the

\begin{proof}[{Proof of \Cref{prop.mcE.1st}}]
We shall apply \Cref{lem.general.correlation.mcE.1st} with $q_F = \zeta^{-1/2} \rho_0^{1/3}$.
Recalling 
\eqref{eqn.asymp.beta.mu.kappa}
then 
\begin{equation*}
\beta k_F^\kappa q_F \sim \kappa^{1/2} \ll \kappa,
\qquad 
\beta q_F^2 \sim 1 \ll \kappa,
\qquad 
\beta(k_F^\kappa - q_F)^2 \sim \kappa \gg 1.
\end{equation*}
We shall choose $d \gg \zeta^{1/2} \rho_0^{-1/3}$. Then, by taking $n$ large enough, we 
can ensure that the last two summands in $[\ldots]^{1/4}$ in \eqref{eq:prop71} are negligible.

We can use \eqref{eqn.approx.Gibbs.define} to bound $\expect{\mcN}_{\Gamma}$, and \eqref{eqn.bdd.rel.entropy.a.priori} to bound $S(\Gamma,\Gamma_0)$.  
Thus, using \eqref{eqn.asymp.beta.mu.kappa},  
the leading terms in the lower bound on $\mcE_{\textnormal{pt}}(\Gamma)$ are
\begin{equation*}
\begin{aligned}
  & 
    - L^3 a^3\rho_0^{8/3} \left[\kappa^{7/2}\zeta^{-7/2} R^2\rho_0^{2/3} + Rd^{-1} 
    \right]
      \\ & \quad 
    - L^{3} 
    \left[a^3R^{-5/2}  \rho_0^{4/3}d^{-3/2} + a^3 \rho_0^{8/3}\right]
  \Bigl[
    \zeta d^3 a^3 \rho_0^2
  + \zeta^{1/2} \kappa^{7/2} \rho_0^{-1/3} d^{-1}
  \Bigr]^{1/4}
\end{aligned}
\end{equation*}
We shall restrict our attention to a compact set of $z$'s, i.e., we can assume that  $\zeta \sim 1$ and hence $\kappa \sim (a^3\rho_0)^{-\alpha}$. 
We shall choose 
\begin{equation*}
d = \rho_0^{-1/3} (a^3\rho_0)^{-s},
\qquad 
R = \rho_0^{-1/3} (a^3\rho_0)^{t}
\end{equation*}
for some $s,t> 0$ to be determined. Then, 
$\mcE_{\textnormal{pt}}(\Gamma) \geq - C(z) L^3 a^3\rho_0^{8/3} \left[(a^3\rho_0)^{\sigma} + \mathfrak{e}_L\right]$
with $C(z)$ bounded uniformly on compact sets of $z$'s and 
\begin{equation*}
\sigma = 
   \frac{1}{8} \min \Bigl\{
    16t - 28\alpha , \,  
    2 + 6s - 20t , \, 
    2 - 6s , \, 
    14s - 20t - 7\alpha , \, 
    2s - 7 \alpha
  \Bigr\}.
\end{equation*}
The choice 
\begin{equation*}
 s = \frac{1}{4} + \frac{7}{8}\alpha,\qquad 
 t = \frac{1}{32} + \frac{91}{64}\alpha
\end{equation*}
yields $\sigma = 1/16 - 21\alpha/32$, and hence the desired bound.
\end{proof}

The rest of this section is devoted to the proof of \Cref{lem.general.correlation.mcE.1st}. 
It is divided into three parts. 
We first introduce some convenient notation.

\begin{notation}\label{notation.tr.Tr}
We write 
\begin{equation*}
\Tr = \Tr_{\mcF(\mfh)},
\qquad 
\tr = \Tr_{\mfh}.
\end{equation*}
That is, 
$\tr$ for the trace of the one-particle space $\mfh = L^2(\Lambda;\C)$ and 
$\Tr$ for the trace over the Fock space $\mcF(\mfh)$.
\end{notation}

\begin{defn}[{See also \cite[Definition 5.5]{Lewin.Nam.ea.2021} and \cite[Appendix A]{Hainzl.Lewin.ea.2009}}]
\label{def.GammaX}
For any projection $X$ on $\mfh$ and state $\Gamma$ we define the operator 
$\Gamma_X$ as the state $\Gamma$ restricted to the space $\mcF(X\mfh)$ as follows:
Any operator $\mcA$ on $\mcF(X\mfh)$ may naturally be extended as $\mcA \otimes \mathbbm{1}_{\mcF(X^\perp \mfh)}$ 
on $\mcF(X\mfh) \otimes \mcF(X^\perp \mfh) \simeq \mcF(\mfh)$. 
With this extension and unitary equivalence hidden in the notation, 
the state $\Gamma_X$ is defined by duality satisfying $\Tr \mcA\Gamma_X = \Tr \mcA \Gamma$ for all $\mcA$ on $\mcF(X\mfh)$.
\end{defn}

\subsection{Regularizing the interaction}\label{sec.dyson.lemma}
As a first step we shall replace the effective interaction  $W=V(1-\varphi)$ by a 
much weaker interaction $U$ of longer range. This step is, at least in spirit, analogous to the  ``Dyson lemma'' \cite[Lemma 1]{Dyson.1957} (see also \cite[Lemma 1]{Lieb.Yngvason.1998}) applied first for dilute Bose gases. In our case, 
the presence of the projection $P$ 
effectively smears  out the particle coordinates on a length scale $(k_F^\kappa)^{-1} \gg a$,
allowing us to replace  $W=V(1-\varphi)$ by  
$U$ as long as the range of $U$ is $\ll (k_F^\kappa)^{-1}$
and $\int |x|^2 U = \int |x|^2 W = 24 \pi a^3$.
For any interaction $U$ we define 
\begin{equation*}
\mcU = \!\ud \Gamma(U),
\qquad 
\mcU_P = \!\ud\Gamma(PPUPP).
\end{equation*}
More concretely we show 
\begin{lemma}\label{lem.dyson}
Let $U$ be a radial function with $\int |x|^2 U = \int |x|^2 W$. 
Then for any state $\Gamma$, 
\begin{equation*}
 \expect{\mcW_P}_\Gamma = \expect{\mcU_P}_{\Gamma} + O\left( (k_F^\kappa)^{7} \norm{|\cdot|^4(U - W)}_{L^1} \expect{\mcN}_\Gamma\right).
 \end{equation*}
\end{lemma}

\begin{proof}
We start by writing 
\begin{equation*}
\mcW_P - \mcU_P = \!\ud\Gamma(PP (W - U) PP) = \frac{1}{2} \iint (W-U)(x-y) c_x^* c_y^* c_y c_x \ud x \ud y .
\end{equation*}
We Taylor expand $c_y$ and $c_y^*$ around $y=x$ as in \eqref{eqn.Taylor.cy.1st.order}. 
Changing  variables to $z=y-x$ we have 
\begin{equation*}
\mcW_P - \mcU_P
  = \frac{1}{2} \iint (W-U)(z)z^\mu z^\nu \int_0^1 \ud t \int_0^1 \ud s \, c_x^* \nabla^\mu c_{x+tz}^* \nabla^\nu c_{x+sz} c_x \ud x \ud z 
\end{equation*}
Evaluating in the state $\Gamma$ we define the function(s) $\phi_x^{\mu\nu}(z) = \expect{c_x^* \nabla^\mu c_{x+tz}^* \nabla^\nu c_{x+sz} c_x}_\Gamma$. 
We shall Taylor expand $\phi_x^{\mu\nu}$ around $z=0$. 
Noting that $\norm{\nabla^n c} \leq C (k_F^\kappa)^{3/2+n}$ we can bound 
\begin{equation*}
|\nabla^2\phi_x^{\mu\nu}(z)| 
  \leq C \norm{ \Gamma^{1/2} c_x^* }_{\mathfrak{S}_2} (k_F^\kappa)^{7} \norm{ c_x \Gamma^{1/2}}_{\mathfrak{S}_2} 
  = C (k_F^\kappa)^{7} \expect{c_x^* c_x}_{\Gamma}.
\end{equation*}
Consequently 
\begin{equation*}
\left| \phi_x^{\mu\nu}(z) - \phi_x^{\mu\nu}(0) - z^\lambda \nabla^\lambda \phi_x^{\mu\nu}(0) \right| \lesssim  |z|^2 (k_F^\kappa)^7 \expect{c_x^* c_x}_{\Gamma} \,.
\end{equation*}
Noting further that $\int z^\mu z^\nu U = \int z^\mu z^\nu W$ and $\int z^\mu z^\nu z^\lambda U = 0 = \int z^\mu z^\nu z^\lambda W$
we find
\begin{equation*}
\abs{\expect{\mcW_P - \mcU_P}_\Gamma}
  \leq 
    C (k_F^\kappa)^{7} \int |z|^4 \abs{U-W} \ud z  \int \expect{c_x^* c_x}_{\Gamma  } \ud x.
\end{equation*}
Since $ \int {c_x^* c_x} \ud x \leq  \mcN$ 
this proves the lemma. 
\end{proof}

\subsection{Localizing the relative entropy}\label{sec.local.rel.entropy}
Because of \Cref{lem.dyson}, our task is now to evaluate $\expect{\mcU_P}_\Gamma$ 
for an appropriate  choice of $U$ in an approximate Gibbs state $\Gamma$. Our goal is to show that we can replace $\Gamma$ by $\Gamma_0$ to leading order, which amounts  to showing the validity 
of first order perturbation theory. To be able to utilize the a priori bound \eqref{eqn.bdd.rel.entropy.a.priori} on the 
 relative entropy $S(\Gamma,\Gamma_0)$, we shall  
   localize the latter in a suitable way, following the method in 
\cite[Proof of Theorem 3.1]{Seiringer.2006a}. 
Due to the presence of the projection $P$, the case considered here is slightly different from the one considered in  \cite{Seiringer.2006a}, hence we cannot just quote the result but give the complete argument here.

A first step is localizing the relative entropy $S(\Gamma_P, \Gamma_d)$ for a state $\Gamma_d$ closely related to the free Gibbs state $\Gamma_0$, but with a cut-off the removes all correlations between well separated regions. 
The state $\Gamma_d$ is defined as follows.

Let $\eta:\R^3\to \R$ be a smooth function with 
\begin{itemize}
\item $\eta(0) = 1$ and $\eta(x)=0$ for $|x|\geq 1$,
\item $\hat\eta(p) = \int \ud x \, \eta(x) e^{-ipx} \geq 0$ for all $p\in \R^3$.
\end{itemize}
(To construct such a function $\eta$ simply take any smooth compactly supported function and convolve it with itself.)
Define then for $0 < d \leq L/2$ the function $\eta_d(x) = \eta(x/d)$ 
(more precisely its periodization $\eta_d(x) = \sum_{n\in \Z^3} \eta\left(\frac{x+nL}{d}\right)$). 
Let $\Gamma_d$ be the quasi-free state on $\mcF(\mfh)$ with one-particle density matrix 
\begin{equation*}
  \gamma_d(x;y) = \gamma_0(x;y) \eta_d(x-y),
\end{equation*}
with $\gamma_0$ the one-particle density matrix of the free Gibbs state. 
Define then $\overline{d}\geq d$ by $L/2\overline{d} = \lfloor L/2d \rfloor$ 
(here $\lfloor \cdot \rfloor$ denotes the integer part)
and define for $r\leq d/2$
\begin{equation*}
X_r(x) = \sum_{\xi \in 2\overline{d}\Z^3 \cap [0,L)^3} \chi_{r,\xi}(x).
\end{equation*}
Noting that $\eta_d$ vanishes outside a ball of radius $d \leq 2\overline{d}-2r$ we have 
\begin{equation*}
X_r \gamma_d X_r = \sum_{\xi \in 2\overline{d}\Z^3\cap[0,L)^3 } \chi_{r,\xi} \gamma_d \chi_{r,\xi}.
\end{equation*}
Thus, 
\begin{equation*}
(\Gamma_d)_{X_r} \simeq \bigotimes_{\xi \in 2\overline{d}\Z^3\cap[0,L)^3 } (\Gamma_d)_{\chi_{r,\xi}}
\end{equation*}
Next, we note that the relative entropy is monotone decreasing under restriction \cite[Theorem 1.5]{Ohya.Petz.2004}
and further that it is superadditive for its right argument being a product state (which follows easily from the subadditivity of the von Neumann entropy).  
Concretely this means that
\begin{equation*}
S(\Gamma_P,\Gamma_d) 
  \geq S\left( (\Gamma_P)_{X_r}, (\Gamma_d)_{X_r}\right)
  \geq \sum_{\xi \in 2\overline{d}\Z^3\cap[0,L)^3} S\left( (\Gamma_P)_{\chi_{r,\xi}}, (\Gamma_d)_{\chi_{r,\xi}}\right).
\end{equation*}
Replacing $X_r$ by $X_r(\cdot + a)$ for $a\in [0,2\overline{d}]^3$ and averaging over $a$ we find for any $r\leq d/2$
(this is \cite[(5.10)]{Seiringer.2006a})
\begin{equation}\label{eqn.GammaP.Gammad.localization}
S(\Gamma_P,\Gamma_d) 
  \geq \frac{1}{(2 \overline{d})^3} \int_\Lambda \ud \xi \,S\left( (\Gamma_P)_{\chi_{r,\xi}}, (\Gamma_d)_{\chi_{r,\xi}}\right).
\end{equation}

The next step is relating the relative entropy $S(\Gamma_P, \Gamma_d)$ to the relative entropy $S(\Gamma,\Gamma_0)$ for which we have a priori bounds. 
In the absence of the projection $P$, this is done in \cite[Section 5.2]{Seiringer.2006a}. 
We claim that (compare to \cite[Eq.~(5.31)]{Seiringer.2006a})
\begin{lemma}\label{lem.rel.entropy.localization}
Let $\Gamma$ be any state and let $0 < q_F < k_F^\kappa$ with $\beta(k_F^\kappa - q_F)^2 \gtrsim 1$. 
Then, for any $n \in \N$
\begin{equation*}
\begin{aligned}
S(\Gamma_P, \Gamma_d) 
  & \lesssim 
    S(\Gamma, \Gamma_0)
    + L^3 \beta^4 (k_F^\kappa)^7 d^{-4}
  \\ & \quad 
    + L^3 \beta^{-1} (k_F^\kappa - q_F) e^{- \kappa + \beta(2 k_F^\kappa q_F - q_F^2)} 
    + C_n (dq_F)^{-2n} \log Z_0
    + L^3 \mathfrak{e}_L.
\end{aligned}
\end{equation*}
The constant $C_n$ depends only on $n$.
\end{lemma}

Combining \Cref{lem.rel.entropy.localization,eqn.GammaP.Gammad.localization} we thus find that
\begin{equation}\label{eq.bdd.local.entropy}
\begin{aligned}
\int_\Lambda \ud \xi \, S\left( (\Gamma_P)_{\chi_{r,\xi}}, (\Gamma_d)_{\chi_{r,\xi}}\right)
   & \lesssim   d^3 S(\Gamma, \Gamma_0)
    + L^3 \beta^4 (k_F^\kappa)^7 d^{-1}
    + L^3 d^3 \beta^{-1} (k_F^\kappa - q_F) e^{- \kappa + \beta(2 k_F^\kappa q_F - q_F^2)} 
  \\ & \quad 
    + C_n d^3 (dq_F)^{-2n} \log Z_0
    + L^3 \mathfrak{e}_L.
\end{aligned}
\end{equation}
for any $0 < r\leq d/2 \leq L/4$.

\begin{proof}[{Proof of \Cref{lem.rel.entropy.localization}}]
Following \cite[Section 5.2]{Seiringer.2006a}, we first make the following observation. Denote by $\Gamma_\omega$ the quasi-free state with one-particle-density matrix $\omega$.
Then for any state $\Gamma$ (whose one-particle-density matrix we denote by $\gamma$) we have that 
\begin{equation}\label{eqn.formula.S.quasi-free}
S(\Gamma, \Gamma_\omega) = \Tr \Gamma \log \Gamma - \tr \gamma \log \omega - \tr (1-\gamma)\log (1-\omega) 
\end{equation}
is convex in $\omega$.
We may write $\gamma_d$ as a convex combination
\begin{equation*}
  \gamma_d = \frac{1}{|\Lambda|} \sum_{q \in \frac{2\pi}{L}\Z^3} \hat\eta_d(q) \gamma_{q},
  \qquad 
  \hat \gamma_{q}(p) = \frac{1}{2} \left(\hat\gamma_0(p+q) + \hat\gamma_0(p-q)\right).
\end{equation*}
(Note that for $q=0$ we indeed have $\gamma_q = \gamma_0$.)
Recalling that $\hat\eta_d \geq 0$ and $\frac{1}{|\Lambda|}\sum \hat \eta_d = \eta(0) =  1$ by construction, we thus have
\begin{equation}\label{eqn.S.GammaP.Gammad.q.sum}
S(\Gamma_P, \Gamma_d) \leq \frac{1}{|\Lambda|} \sum_{q\in \frac{2\pi}{L}\Z^3} \hat\eta_d(q) S(\Gamma_{P},  \Gamma_{\gamma_{q}})\,.
\end{equation}
We  claim that for any $t>0$ (compare to \cite[(5.14)]{Seiringer.2006a})
\begin{equation}\label{eqn.claimed.S.GammaP.Gammaq}
\begin{aligned}
S(\Gamma_{P}, \Gamma_{\gamma_{q}})
  & \leq (1 + t^{-1}) S(\Gamma_P, (\Gamma_0)_P) 
      - \tr Q \log (1-\gamma_q)
  \\ & \quad 
    + \tr P(h_q - h_0) \left(\frac{1}{1 + e^{(1+t)h_0 - th_q}} - \frac{1}{1 + e^{h_q}}\right)
\end{aligned}
\end{equation}
where $h_q = \log \frac{1-\gamma_q}{\gamma_q}$, i.e., $\gamma_q = (1+ e^{h_q})^{-1}$.
We will later choose $t=1$ but for now it is convenient to leave it as a variable.

To prove \Cref{eqn.claimed.S.GammaP.Gammaq} we first note that $(\Gamma_0)_P$ is a quasi-free state with one-particle density matrix $P\gamma_0$.
Using the formula in \Cref{eqn.formula.S.quasi-free} we have 
\begin{equation*}
\begin{aligned}
& (1 + t^{-1}) S(\Gamma_P, (\Gamma_0)_P) - S(\Gamma_{P},  \Gamma_{\gamma_{q}} )
  \\ & \quad = t^{-1} \Tr \Gamma_P \log \Gamma_P 
    + \tr Q \log (1-\gamma_q)
  \\ & \qquad 
    + \tr P \left[\gamma \log \gamma_q + (1-\gamma)\log (1-\gamma_q) - (1+t^{-1}) \gamma\log \gamma_0 - (1+t^{-1}) (1-\gamma) \log (1-\gamma_0)\right]
  \\ & \quad 
    = t^{-1} \left[\Tr \Gamma_P \log \Gamma_P + \tr P \gamma ( (1+t)h_0 - t h_q)\right]
      + \tr Q \log (1-\gamma_q) 
  \\ & \qquad 
       - \tr P \log (1+e^{-h_q}) 
      + (1+t^{-1})\tr P \log (1+e^{-h_0})  \,.
\end{aligned}
\end{equation*}
By the Gibbs variational principle applied to the one-body Hamiltonian $P((1+t)h_0 - t h_q)$ 
and Taylor expanding we have 
\begin{equation*}
\begin{aligned}
& \Tr \Gamma_P \log \Gamma_P + \tr P \gamma ( (1+t)h_0 - t h_q)
  \\ & \quad 
    \geq - \tr P \log \left(1 + e^{-(1+t)h_0 + th_q}\right)
  \\ & \quad 
    = -\tr P \log \left(1 + e^{-h_0}\right) + t \int_0^1 \ud s \, \tr P (h_0-h_q) \frac{1}{1 + e^{(1+st)h_0 - st h_q}}
\end{aligned}
\end{equation*}
Similarly we have 
\begin{equation*}
 - \tr P \log (1+e^{-h_q}) = - \tr P \log (1+e^{-h_0}) + \int_0^1 \ud s \, \tr P (h_q - h_0) \frac{1}{1 + e^{(1-s)h_0 + sh_q}}
\end{equation*}
Combining these equations and noting the the integrands in the $s$-integrals are monotone in $s$ we conclude the proof of \Cref{eqn.claimed.S.GammaP.Gammaq}.

To bound the first term on the right-hand-side of \eqref{eqn.claimed.S.GammaP.Gammaq} we note that $S(\Gamma_P, (\Gamma_0)_P) \leq S(\Gamma,\Gamma_0)$, 
since the relative entropy is decreasing in restrictions \cite[Theorem 1.5]{Ohya.Petz.2004}. 
Next, we bound the last term  on the right-hand-side of \eqref{eqn.claimed.S.GammaP.Gammaq} similarly to \cite[(5.27)]{Seiringer.2006a}.
Taylor expanding to first order in $t$ we find that 
\begin{equation*}
\abs{\tr P(h_q - h_0) \left(\frac{1}{1 + e^{(1+t)h_0 - th_q}} - \frac{1}{1 + e^{h_q}}\right)}
 \leq (1+t) \tr P (h_q-h_0)^2. 
\end{equation*}
To estimate the latter trace we recall the bound \cite[Lemma 5.1]{Seiringer.2006a}
\begin{equation*}
\abs{h_q(p)-h_0(p)} \leq C \beta q^2 ( \zeta + \beta p^2). 
\end{equation*}
Then clearly (recalling $\zeta \sim \beta \rho_0^{2/3}$ from \eqref{eqn.asymp.beta.mu.kappa})
\begin{equation*}
\tr P (h_q-h_0)^2 \leq C \sum_{|p|\leq k_F^\kappa} \beta^2 q^4 ( \zeta + \beta p^2)^2 
  \leq C L^3 \beta^4 (k_F^\kappa)^7 q^4.
\end{equation*}

To bound the second term on the right-hand-side of \eqref{eqn.claimed.S.GammaP.Gammaq} we can use convexity of $\gamma \mapsto -\log (1-\gamma) $
and the symmetry $p\to -p$ to conclude that 
\begin{equation*}
- \tr Q \log (1-\gamma_q) \leq  \sum_{|p| > k_F^\kappa} \log \left(1 + e^{-\beta((p+q)^2-\mu)}\right).
\end{equation*}
We treat this sum separately for large and small $q$'s. 
Let $0 < q_F < k_F^\kappa$ and assume that $|q|< q_F$. In this case 
\begin{equation}
\begin{aligned}
- \tr Q \log (1-\gamma_q) 
  & \leq  \sum_{|p| > k_F^\kappa - q_F} \log \left(1 + e^{-\beta(p^2-\mu)}\right)
  \\ & 
  \leq C L^3 \beta^{-3/2} [\beta^{1/2}(k_F^\kappa - q_F) + 1] e^{-\beta( (k_F^\kappa - q_F)^2 - \mu)}
  + L^3 \mathfrak{e}_L,
  \label{eqn.bdd.q<qF.sum}
\end{aligned}
\end{equation}
To prove this we view the sum as a Riemann sum for the corresponding integral and compute the integral explicitly.
We give the details in \Cref{sec.bdd.sums.riemann}.
For $|q| \geq q_F$ we shall simply bound 
\begin{equation*}
- \tr Q \log (1-\gamma_q) 
  \leq  \sum_{p \in \frac{2\pi}{L}\Z^3} \log \left(1 + e^{-\beta(p^2-\mu)}\right)
  =  \log Z_0,
\end{equation*}
where $Z_0$ is the partition function of the free gas. 
Combining these bounds and choosing $t=1$ we conclude that  
(for $\beta(k_F^\kappa - q_F)^2 \gtrsim 1$)
\begin{equation*}
\begin{aligned}
S(\Gamma_{P},  \Gamma_{\gamma_q})
  & \lesssim  S(\Gamma,\Gamma_0) 
  + L^3\beta^4 (k_F^\kappa)^7 q^4 
  \\ & \quad 
  + L^3 \beta^{-1} (k_F^\kappa - q_F) e^{-\beta( (k_F^\kappa - q_F)^2 - \mu)} 
  \chi_{(|q| < q_F)}
  + \log Z_0 \chi_{(|q| \geq q_F)}
  + L^3 \mathfrak{e}_L.
\end{aligned}
\end{equation*}

We insert this bound in \eqref{eqn.S.GammaP.Gammad.q.sum}. To evaluate  the $q$-summation we note that 
\begin{equation*}
\frac{1}{|\Lambda|}\sum_{q \in \frac{2\pi}{L}\Z^3}   \hat\eta_d(q) q^{2n} = (-\Delta)^n \eta(0) d^{-2n} = C_n d^{-2n}
\end{equation*}
for any integer $n\geq 0$. (Recall also that $\eta(0) = 1$.)
In particular for any $n\in \N$ 
\begin{equation*}
\frac{1}{L^3} \sum_{q} \hat\eta_d(q) \log Z_0 \chi_{(|q| > q_F)}
  \leq q_F^{-2n} \log Z_0  \frac{1}{L^3} \sum_{|q| > q_F} \hat\eta_d(q) q^{2n} 
  \leq C_n (q_F d)^{-2n} \log Z_0 . 
\end{equation*}
Noting that 
$-\beta((k_F^\kappa - q_F)^2 - \mu)
  = - \kappa + \beta(2 k_F^\kappa q_F - q_F^2)$
we conclude the proof of the lemma.
\end{proof}

\subsection{Localizing the interaction}

The bound \eqref{eq.bdd.local.entropy} allows us to conclude that the states $\Gamma_P$ and $\Gamma_d$ are suitably close when viewed on a ball of some radius $r$. As long as $r$ is large compared to $R$ (the range of $U$) this will thus allow us to obtain the desired estimate on the expectation value $\expect{\mcU_P}_{\Gamma}$, and to bound the difference $\expect{\mcU_P}_\Gamma - \expect{\mcU}_{\Gamma_0}$ by the localized relative entropy.

To do this precisely, we need to also localize the interaction into balls. The subsequent analysis  follows closely the corresponding analysis in \cite[Section 5.3]{Seiringer.2006a}. 
We recall that, for a state $\Gamma$, the state $\Gamma_X$ denotes its localization using the projection $X$, see \Cref{def.GammaX}. 
The state $\Gamma_0$ is the free Gibbs state and the state $\Gamma_d$ is defined in \Cref{sec.local.rel.entropy}.
We shall prove

\begin{lemma}\label{lem.local.int}
Let $\Gamma$ be any state and let $U \geq 0$ be compactly supported with range $R$. 
Then, for any $0 < r \leq d/2 \leq L/4$ with $d \gtrsim \rho_0^{-1/3}$ and $1/z$ bounded, we have 
\begin{align*}
\expect{\mcU_P}_{\Gamma}
  & 
  \geq \expect{\mcU}_{\Gamma_0}
    - C L^3 \norm{|\cdot|^2U}_{L^1} \rho_0^{8/3} \frac{R}{r}
  \\ 
  & \quad 
    - C L^{9/4} \left[\norm{|\cdot|U}_{L^2}^2  \rho_0^{2+2/3} + C \norm{|\cdot|^2U}_{L^1}^2 \rho_0^{3+4/3}(1 + r^3\rho_0)\right]^{1/2} 
  \\ & \qquad \times 
    r^{-3/2} 
    \left[\int_\Lambda \ud \xi \, S((\Gamma_P)_{\chi_{r,\xi}}, (\Gamma_d)_{\chi_{r,\xi}}) \right]^{1/4}
     - L^3\mfe_L. 
\end{align*}

\end{lemma}

\begin{proof}
Write for any $r>0$
\begin{equation*}
\begin{aligned}
U(x-y) & = \frac{3}{4\pi r^3} \int_\Lambda \ud \xi \, U(x-y) \chi_{r,\xi}(y)
  \\
    & = \frac{3}{4\pi r^3} \int_\Lambda \ud \xi \, \chi_{r,\xi}(x) U(x-y) \chi_{r,\xi}(y)
      + \frac{3}{4\pi r^3} \int_\Lambda \ud \xi \, (1-\chi_{r,\xi}(x)) U(x-y) \chi_{r,\xi}(y)
\end{aligned}
\end{equation*}
with $\chi_{r,\xi}$ the characteristic function of a ball of radius $r$ centered at $\xi$. 
As a lower bound we may keep only the first term since $U\geq 0$.
Define $U_{r,\xi}(x,y) = \chi_{r,\xi}(x) U(x-y) \chi_{r,\xi}(y)$ and $\mcU_{r,\xi} = \!\ud\Gamma(U_{r,\xi})$.
Then 
\begin{equation*}
\expect{\mcU_P}_\Gamma 
=
\expect{\mcU}_{\Gamma_P} 
  \geq \frac{3}{4\pi r^3} \int_\Lambda \ud \xi \, \Tr\left[ \mcU_{r,\xi} \Gamma_P \right] 
  = \frac{3}{4\pi r^3} \int_\Lambda \ud \xi \, \Tr\left[ \mcU_{r,\xi} (\Gamma_P)_{\chi_{r,\xi}}\right] 
\end{equation*}
since we may replace $\Gamma_P$ by $(\Gamma_P)_{\chi_{r,\xi}}$ since $U_{r, \xi}$ is localized to this domain.

For any $K$ we introduce 
\begin{equation*}
f_K(t) 
= \min \{t, K\}
= t- [t-K]_+ , 
\qquad 
[\cdot]_+ = \min\{\cdot,0\}.
\end{equation*}
Then $t\geq f_K(t)$ for any $t$. Using this on $t=\mcU_{r,\xi}$ we thus have (with $\Gamma_d$ as above)
\begin{align}
\expect{\mcU_P}_\Gamma
  & \geq \frac{3}{4\pi r^3} \int_\Lambda \ud \xi \, \Tr\left[f_K\left(\mcU_{r,\xi}\right) (\Gamma_P)_{\chi_{r,\xi}}\right]
  \nonumber
  \\
  & = \frac{3}{4\pi r^3} \int_\Lambda \ud \xi \, \Tr\left[f_K\left(\mcU_{r,\xi}\right) (\Gamma_d)_{\chi_{r,\xi}}\right]
  \nonumber
      + \frac{3}{4\pi r^3} \int_\Lambda \ud \xi \, \Tr\left[f_K\left(\mcU_{r,\xi}\right) \left((\Gamma_P)_{\chi_{r,\xi}} - (\Gamma_d)_{\chi_{r,\xi}}\right)\right]
  \nonumber
  \\ &  
    =: \mcI + \mcE_{\textnormal{state}}
  \label{eqn.expect.W.Gamma.two.terms}
\end{align}
The first term $\mcI$ is the main term. 
To evaluate it,  we use $f_K(t) = t - [t-K]_+$ and bound $[t-K]_+ \leq t^2/4K$. Thus 
\begin{equation*}
\begin{aligned}
\Tr\left[f_K\left(\mcU_{r,\xi}\right) (\Gamma_d)_{\chi_{r,\xi}}\right]
  & \geq  \Tr\left[\mcU_{r,\xi} (\Gamma_d)_{\chi_{r,\xi}}\right]
    - \frac{1}{4K} \Tr\left[\left(\mcU_{r,\xi}\right)^2 (\Gamma_d)_{\chi_{r,\xi}}\right].
\end{aligned}
\end{equation*}
For the first term we note that $(\Gamma_d)_{\chi_{r,\xi}}$ is quasi-free and has two-particle density 
\begin{equation*}
\begin{aligned}
\rho_d^{(2)}(x,y)\chi_{r,\xi}(x) \chi_{r,\xi}(y) 
  & = (\rho_0^2 - \abs{\gamma_d(x-y)}^2)\chi_{r,\xi}(x) \chi_{r,\xi}(y) 
  \\ & 
  \geq \rho_0^{(2)}(x,y) \chi_{r,\xi}(x) \chi_{r,\xi}(y)
\end{aligned}
\end{equation*}
since $|\gamma_d|\leq |\gamma_0|$ by construction. 
In particular,  
\begin{equation*}
\begin{aligned}
& \Tr\left[\mcU_{r,\xi} (\Gamma_d)_{\chi_{r,\xi}}\right]
  \\ &\quad  \geq \frac{1}{2}\iint_{(B_{r,\xi})^2} U(x-y) \rho_0^{(2)}(x,y) \ud x \ud y
  \\ & \quad 
  = \frac{1}{2}\iint_{\Lambda \times B_{r,\xi}} U(x-y) \rho_0^{(2)}(x,y) \ud x \ud y
    - \frac{1}{2} \iint_{(\Lambda \setminus B_{r,\xi}) \times B_{r,\xi}} U(x-y) \rho_0^{(2)}(x,y) \ud x \ud y.
\end{aligned}
\end{equation*}
Integrating this over $\xi$ we get 
\begin{equation*}
\begin{aligned}
\frac{3}{4\pi r^3} \int_\Lambda \ud \xi \, 
  \Tr\left[\mcU_{r,\xi} (\Gamma_d)_{\chi_{r,\xi}}\right] 
&
  \geq \Tr \left[\mcU \Gamma_0\right]
    - \frac{3}{8\pi r^3} \int_\Lambda \ud \xi \, 
    \iint_{(\Lambda \setminus B_{r,\xi}) \times B_{r,\xi}} \ud x \ud y \, U(x-y) \rho_0^{(2)}(x,y).
\end{aligned}
\end{equation*}
We conclude that 
\begin{equation}\label{eqn.fK.Gamma.d.eval}
\begin{aligned}
\mcI & \geq \expect{\mcU}_{\Gamma_0}
  \\ & \quad 
  - \underbrace{\frac{3}{16\pi Kr^3} \int_{\Lambda} \ud \xi \, \Tr\left[\mcU_{r,\xi}^2 (\Gamma_d)_{\chi_{r,\xi}}\right]}_{\mcE_{d} := }
    - \underbrace{\frac{3}{8\pi r^3} \int_{\Lambda} \ud \xi \, 
        \iint_{(\Lambda \setminus B_{r,\xi}) \times B_{r,\xi}} \ud x \ud y \, U(x-y) \rho_0^{(2)}(x,y)}_{\mcE_{\textnormal{local}} :=}.
\end{aligned}
\end{equation}
We are left with bounding the three error terms $\mcE_{\textnormal{local}}$, $\mcE_{d}$ and $\mcE_{\textnormal{state}}$.

First, we bound the error term $\mcE_{\textnormal{local}}$ from \eqref{eqn.fK.Gamma.d.eval}. 
Using that $U$ has compact support of range $R$ and recalling the formula in \eqref{eqn.rho2.free} for $\rho_0^{(2)}$ we have 
for $1/z$ bounded 
\begin{equation*}
\begin{aligned}
  \abs{\mcE_{\textnormal{local}}}
  &
    \leq \frac{C}{r^3} \int_{\Lambda} \ud \xi \int_{B_{r+R,\xi}\setminus B_{r,\xi}} \ud x \int_\Lambda \ud z \, |z|^2 U(z) \rho_0^{8/3}
    \left[1 + L^{-1} \zeta \rho_0^{-1/3} \right]
  \\ & 
  \leq C L^3 \rho_0^{8/3} \norm{|\cdot|^2 U}_{L^1} \frac{R}{r}
  	 + L^3 \mfe_L.
\end{aligned}
\end{equation*}

The error term $\mcE_d$ in \eqref{eqn.fK.Gamma.d.eval} can be evaluated explicitly since $\Gamma_d$ is a quasi-free state. 
We have 
\begin{equation}\label{eqn.Tr.U**2.Gammad}
\Tr\left[\left(\mcU_{r,\xi}\right)^2 (\Gamma_d)_{\chi_{r,\xi}}\right]
    = \iiiint_{(B_{r,\xi})^{4}} U(x-y) U(z-z') \Tr \left[a_x^* a_y^* a_y a_x a_z^* a_{z'}^* a_{z'} a_z \Gamma_d 
    \right]
    \ud x \ud y \ud z \ud z'.
\end{equation}
Normal ordering and using 
$\Tr \left[a_{x_1}^*\cdots a_{x_n}^* a_{x_n} \cdots a_{x_1} \Gamma_d 
\right] = \rho_d^{(n)}(x_1,\ldots,x_n) 
$
we find
\begin{align*}
\eqref{eqn.Tr.U**2.Gammad}
  & = 2 \iint_{(B_{r,\xi})^{2}} U(x-y)^2 \rho_d^{(2)}(x,y) \ud x \ud y 
  \\ & \quad 
    + 4 \iiint_{(B_{r,\xi})^3} U(x-y) U(x-z') \rho_d^{(3)}(x,y,z') \ud x \ud y \ud z' 
  \\ & \quad 
    + \iiiint_{(B_{r,\xi})^4} U(x-y) U(z-z') \rho_d^{(4)}(x,y,z,z') \ud x \ud y \ud z \ud z'.
\end{align*}
We may Taylor expand the reduced densities as in \cite[Lemma 3.6]{Lauritsen.Seiringer.2023a}, to conclude that as long as  $d \gtrsim \rho_0^{-1/3}$
and $1/z$ is bounded 
\[
\begin{aligned}
\rho_d^{(2)}(x,y) & \leq C \rho_0^{2+2/3} |x-y|^2 \left[1 + L^{-1} \zeta \rho_0^{-1/3}\right],
\\  
\rho_d^{(3)}(x,y,z') & \leq C \rho_0^{3+4/3} |x-y|^2 |x-z'|^2 \left[1 + L^{-1} \zeta \rho_0^{-1/3}\right], 
\\ 
\rho_d^{(4)}(x,y,z,z') & \leq C \rho_0^{4+4/3} |x-y|^2 |z-z'|^2 \left[1 + L^{-1} \zeta \rho_0^{-1/3}\right].
\end{aligned}
\]
This gives the bound
\begin{align*}
\eqref{eqn.Tr.U**2.Gammad}
  & \leq C r^3 \rho_0^{2+2/3} 
    \Biggl[ 
      \int |x|^2 U^2 \ud x 
      + \rho_0^{1+2/3} \left(\int |x|^2 U \ud x\right)^2
      + r^3 \rho_0^{2+2/3} \left(\int |x|^2 U \ud x\right)^2
      \Biggr]
  \\ & \quad \times 
      \left[1 + L^{-1} \zeta \rho_0^{-1/3}\right].
\end{align*}
Thus 
\begin{equation}\label{eqn.bdd.mcE.d}
  |\mcE_d|
    \leq C K^{-1} L^3 \left[ \norm{|\cdot|U}_{L^2}^2 \rho_0^{2+2/3} 
    + \norm{|\cdot|^2U}_{L^1}^2 \rho_0^{3+4/3}(1 + r^3\rho_0)
    + \mfe_L 
    \right].
\end{equation}

Finally  we bound the error term $\mcE_{\textnormal{state}}$ from \eqref{eqn.expect.W.Gamma.two.terms}. 
Since $f_K(t) \leq K$ for any $t$ we have 
\begin{equation*}
|\mcE_{\textnormal{state}}| \leq \frac{3K}{4\pi r^3} \int_\Lambda \ud \xi \, \norm{(\Gamma_P)_{\chi_{r,\xi}} - (\Gamma_d)_{\chi_{r,\xi}}}_{\mathfrak{S}_1} 
  \leq \frac{3\sqrt{2}K}{4  \pi r^3} L^{3/2} \left(\int_\Lambda \ud \xi \, S((\Gamma_P)_{\chi_{r,\xi}}, (\Gamma_d)_{\chi_{r,\xi}}) \right)^{1/2}
\end{equation*}
using the Cauchy--Schwarz inequality and that  
$\norm{(\Gamma_P)_{\chi_{r,\xi}} - (\Gamma_d)_{\chi_{r,\xi}}}_{\mathfrak{S}_1}^2 \leq 2 S((\Gamma_P)_{\chi_{r,\xi}}, (\Gamma_d)_{\chi_{r,\xi}})$ 
\cite[Theorem 1.15]{Ohya.Petz.2004}. 
Combining this with the bound for $\mcE_d$ and choosing the optimal $K$ we have 
\begin{align*}
\abs{\mcE_d} + \abs{\mcE_{\textnormal{state}}}
  & \leq C L^{9/4} \left[\norm{|\cdot|U}_{L^2}^2  \rho_0^{2+2/3} + \norm{|\cdot|^2U}_{L^1}^2 \rho_0^{3+4/3}(1 + r^3\rho_0) + \mfe_L\right]^{1/2} 
  \\ & \quad \times 
  r^{-3/2} 
    \left(\int_\Lambda \ud \xi \, S((\Gamma_P)_{\chi_{r,\xi}}, (\Gamma_d)_{\chi_{r,\xi}}) \right)^{1/4}.
  \end{align*}
Together with the bound of $\mcE_{\textnormal{local}}$ above we conclude the proof of the lemma. 
\end{proof}

\subsection{Combining the parts}
Finally, we combine the three parts above and give the

\begin{proof}[{Proof of \Cref{lem.general.correlation.mcE.1st}}]  
Consider the function $U$ given by 
\begin{equation*}
U(x) = 30 a^3 R^{-5} \chi_{|x|\leq R}.
\end{equation*}
This has 
\begin{equation*}
  \norm{|\cdot|^2U}_{L^1} = 24\pi a^3,
  \qquad 
  \norm{|\cdot|^4 U}_{L^1}   = C a^3 R^2,
  \qquad 
  \norm{|\cdot|U}_{L^{2}}^2 = C a^6 R^{-5}.
\end{equation*}
Note also that $  \norm{|\cdot|^4 W}_{L^1}   \leq C a^5$ and $a \lesssim R$. Combining  \Cref{lem.dyson,eq.bdd.local.entropy,lem.local.int} we have 
\begin{equation*}
\begin{aligned}
\expect{\mcW_P}_\Gamma  
  & \geq \expect{\mcU}_{\Gamma_0} 
    - C (k_F^\kappa)^7 a^3 R^2 \expect{\mcN}_\Gamma
    - C L^3 a^3 \rho_0^{8/3} R r^{-1}
      \\ & \quad 
    - C L^{9/4} 
    \left[a^6R^{-5}r^{-3}  \rho_0^{2+2/3} + a^6 \rho_0^{3+4/3}(r^{-3} + \rho_0) \right]^{1/2} 
  \Bigl[
    d^3 S(\Gamma, \Gamma_0) 
  + L^3 \beta^4 (k_F^\kappa)^7 d^{-1}
  \\ & \qquad 
  + L^3 d^3 \beta^{-1} (k_F^\kappa - q_F) e^{- \kappa + \beta(2 k_F^\kappa q_F - q_F^2)} 
  + C_n d^3 (dq_F)^{-2n} \log Z_0
  + L^3 \mathfrak{e}_L
  \Bigr]^{1/4}
  - L^3 \mfe_L.
\end{aligned}
\end{equation*}
Recalling the formula from $\rho_0^{(2)}$ from \eqref{eqn.rho2.free}, we have 
\begin{equation*}
\begin{aligned}
\expect{\mcU}_{\Gamma_0}
  & = \frac{1}{2} \iint U(x-y) \rho_0^{(2)}(x,y) \ud x\ud y 
  \\ & = 2\pi \frac{-\Li_{5/2}(-z)}{(-\Li_{3/2}(-z))^{5/3}} L^3 \rho_0^{2+2/3}  \int |x|^2 U(x) 
    \left[1 + O(|x|^2 \rho_0^{2/3}) + O(L^{-1} \zeta \rho_0^{-1/3})\right]
    \ud x 
  \\ & = \expect{\mcW}_{\Gamma_0}
    + O(L^3 a^3 R^2 \rho_0^{10/3}) 
    +L^3 \mathfrak{e}_L\,,
\end{aligned}
\end{equation*}
where we have used that $\int |x|^2 U(x) \ud x = \int |x|^2 W(x) \ud x$ and that $a\lesssim R$ in the last step. 
Choosing finally $r=d/2$ (as this is clearly the optimal choice) and 
recalling that $\mcE_{\textnormal{pt}}(\Gamma) = \expect{\mcW_P}_{\Gamma} - \expect{\mcW}_{\Gamma_0}$
this concludes the proof of \Cref{lem.general.correlation.mcE.1st}. 
\end{proof}

\appendix
\section{Bounding Riemann sums}
\label{sec.bdd.sums.riemann}
In this section we prove \eqref{eqn.riemann.h/1+exp.h} and \eqref{eqn.bdd.q<qF.sum}.
We do this by viewing the sums as Riemann sums for the corresponding integrals. To estimate their difference, note that for any function 
$F$ and any $k\in \frac{2\pi}L \Z^3$ we have 
\begin{equation}\label{eqn.riemann.general}
\frac{1}{L^3} 
F(k) 
  = \frac{1}{(2\pi)^3} 
  \int_{[-\frac{\pi}{L}, \frac{\pi}{L}]^3} \ud \xi \, 
    \left[F(k+\xi) - \int_0^1 \ud t \, \partial_t F(k+t\xi)\right].
\end{equation}
When summed over $k$, the first term gives the integral $\int F$ and the second term will in general be of order $L^{-1}$ for large $L$.

\begin{proof}[{Proof of \eqref{eqn.riemann.h/1+exp.h}}]
We shall apply the bound $h (1+e^{h/2})^{-1} \leq C e^{-h/3}$ to 
 $h = \beta(|k|^2 - \mu)$,  
obtaining
\begin{equation*}
\sum_{|k| > k_F^\kappa} \frac{\beta(|k|^2 - \mu)}{1 + e^{\frac{\beta}{2}(|k|^2 - \mu)}}
  \leq C  \sum_{|k| > k_F^\kappa}  z^{1/3} e^{-\beta |k|^2 / 3}
\end{equation*}
Let $F(k) = z^{1/3}e^{-\beta|k|^2/3}$ and note that 
\begin{equation*}
\partial_t F(k+t\xi)  
  = - \frac{2}{3} z^{1/3} \beta (k+t\xi) \xi e^{-\beta (k+t\xi)^2/3}.
\end{equation*}
Using that $|\xi|\leq CL^{-1}$ for $\xi \in [-\frac{\pi}{L}, \frac{\pi}{L}]^3$,
we may bound this by
\begin{equation*}
\abs{\partial_t F(k+t\xi)}
  \leq C z^{1/3 }\beta \left[L^{-1}|k+\xi| + L^{-2}\right] e^{-\beta (k+\xi)^2/3} e^{C\beta( |k+\xi| L^{-1} + L^{-2})}.
\end{equation*}
This is clearly integrable in $k+\xi$ and subleading in $L$. (More precisely it is of the order $L^{-1}$.)
Thus, using \eqref{eqn.riemann.general},
\begin{equation*}
\sum_{|k| > k_F^\kappa } z^{1/3 }e^{-\beta k^2/3}
= \frac{L^3}{(2\pi)^3} \int_{|k|\geq k_F^\kappa} z^{1/3} e^{-\beta k^2/3} \ud k + L^3\mathfrak{e}_L.
\end{equation*}
The integral is explicitly given by
\begin{align*}
\int_{|k|\geq k_F^\kappa } z^{1/3} e^{-\beta |k|^2/3} \ud k
  & 
  = \pi z^{1/3}(\beta/3)^{-3/2}  \left(\sqrt{\pi} \erfc((\beta/3)^{1/2}k_F^\kappa) + 2 (\beta/3)^{1/2}k_F^\kappa e^{-\beta(k_F^\kappa)^2/3}\right),
\end{align*}
with $\erfc$ the complementary error function \cite[(7.2.2)]{dlmf}. 
Bounding $\erfc(z) \leq \frac{1}{z+1} e^{-z^2}$ \cite[(7.8.3)]{dlmf} 
we thus have 
\begin{equation*}
\int_{|k|\geq k_F^\kappa - q_F} z^{1/3} e^{-\beta k^2} \ud k
\leq C \beta^{-3/2}  \left(\beta^{1/2}k_F^\kappa + 1\right) e^{-\beta((k_F^\kappa)^2 - \mu)/3}
= C \beta^{-3/2}  \left(\beta^{1/2}k_F^\kappa + 1\right) e^{-\kappa/3}
\end{equation*}
proving  \eqref{eqn.riemann.h/1+exp.h}. 
\end{proof}

\begin{proof}[{Proof of \eqref{eqn.bdd.q<qF.sum}}]
For any $t>-1$ we have $\log(1+t)\leq t$. Thus, in order 
to prove \eqref{eqn.bdd.q<qF.sum} we need to bound $\sum_{|k|>k_F^\kappa - q_F} z e^{-\beta|k|^2}$.
Define $F(k) = z e^{-\beta|k|^2}$ and note that 
\begin{equation*}
\partial_t F(k+t\xi)  
  = -2 z  \beta (k+t\xi) \xi e^{-\beta (k+t\xi)^2}.
\end{equation*}
Using that $|\xi|\leq CL^{-1}$ for $\xi \in [-\frac{\pi}{L}, \frac{\pi}{L}]^3$,
we may bound this by, 
\begin{equation*}
\abs{\partial_t F(k+t\xi)}
  \leq C z \beta \left[L^{-1}|k+\xi| + L^{-2}\right] e^{-\beta (k+\xi)^2} e^{C\beta( |k+\xi| L^{-1} + L^{-2})}.
\end{equation*}
As in the proof of \eqref{eqn.riemann.h/1+exp.h} above, this is clearly integrable in $k+\xi$ and of the order $L^{-1}$.
Thus, using \eqref{eqn.riemann.general},
\begin{equation*}
\sum_{|k| > k_F^\kappa - q_F} z e^{-\beta k^2}
= \frac{L^3}{(2\pi)^3} \int_{|k|\geq k_F^\kappa - q_F} z e^{-\beta k^2} \ud k + L^3\mathfrak{e}_L.
\end{equation*}
The integral is explicitly given by
\begin{align*}
\int_{|k|\geq k_F^\kappa - q_F} z e^{-\beta |k|^2} \ud k
  & 
  = \pi \beta^{-3/2}z \left(\sqrt{\pi} \erfc(\beta^{1/2}(k_F^\kappa - q_F)) + 2 \beta^{1/2}(k_F^\kappa - q_F)e^{-\beta(k_F^\kappa-q_F)^2}\right),
\end{align*}
with $\erfc$ the complementary error function \cite[(7.2.2)]{dlmf}. 
Bounding $\erfc(z) \leq \frac{1}{z+1} e^{-z^2}$ \cite[(7.8.3)]{dlmf} 
as in the proof of \eqref{eqn.riemann.h/1+exp.h} above we thus have 
\begin{equation*}
\int_{|k|\geq k_F^\kappa - q_F} z e^{-\beta k^2} \ud k
\leq C \beta^{-3/2}z \left(\beta^{1/2}(k_F^\kappa - q_F) + 1 \right) e^{-\beta(k_F^\kappa - q_F)^2}.
\end{equation*}
This concludes the proof of \eqref{eqn.bdd.q<qF.sum}. 
\end{proof}

\section{Two dimensions}\label{sec.two.dimensions}
In this appendix we shall sketch the (straightforward) changes to the argument to get the bounds also in dimension $D=2$,
i.e., to prove \Cref{thm.main.2d}.
To illustrate the dimensional dependence and make clear which parts of the argument fails 
in dimension $D=1$ we consider here general dimension $D=1,2$.
The main steps in the proof are \Cref{lem.bdd.list} and \Cref{lem.general.correlation.mcE.1st} together with the a priori bounds of \Cref{lem.bdd.propagate.a.priori}. 
In general dimension $D$ they read 
\begin{lemma}[{\Cref{lem.bdd.list} in $D$ dimensions}]
\label{lem.bdd.list.dimensions}
Let $D=1,2$ and let $\Gamma$ be any state. Then, for $\alpha > 0$ sufficiently small and any $0 < \delta < 1$, 
\begin{align*}
\abs{\expect{\mcH_{0;\mcB}^{\div r}}_\Gamma}
  & \leq C a^D (k_F^\kappa)^{D+1} \abs{\log ak_F^\kappa} 
  \left[ k_F^\kappa \expect{\mcN_Q}_\Gamma^{1/2} + \expect{\mcK_Q}_\Gamma^{1/2}  \right]
  \expect{\mcN}_\Gamma^{1/2},
\\
\abs{\expect{\mcQ_{\textnormal{Taylor}}}_\Gamma} + \abs{\expect{\mcQ_{\textnormal{scat}}}_\Gamma}
  & \leq C a^D \abs{\log ak_F^\kappa} (k_F^{\delta^{-1}\kappa})^{D/2} (k_F^\kappa)^{D/2+2} \expect{\mcN_Q}_\Gamma^{1/2} \expect{\mcN}_\Gamma^{1/2}
    \\ & \quad + C L^{D/2} a^{D/2} (k_F^\kappa)^{D+2} \expect{\mcN_{Q^>}}_\Gamma^{1/2} 
\\
\abs{\expect{\mcQ_{[\mcV_Q,\mcB]}}_\Gamma}
  & \leq \begin{cases}
  C a^{1/2} (k_F^\kappa)^{3/2} \expect{\mcV_Q}_\Gamma^{1/2} \expect{\mcN}_{\Gamma}^{1/2} & D=1,
  \\
  C a^2 (k_F^\kappa)^3 \abs{\log ak_F^\kappa}^{1/2} \expect{\mcV_Q}_{\Gamma}^{1/2} \expect{\mcN}_{\Gamma}^{1/2} & D=2,
  \end{cases}
\\ 
\abs{\expect{\mcQ_{V\varphi}}_\Gamma}
  & \leq C a^{3D/2-1} (k_F^\kappa)^{3D/2+1} \expect{\mcN_Q}_\Gamma^{1/2} \expect{\mcN}_\Gamma^{1/2},
\\ 
\abs{\expect{\mcQ_{\textnormal{OD}} }_\Gamma}
  & \leq \begin{cases}
  C a (k_F^\kappa)^3 \expect{\mcN}_\Gamma 
  + C k_F^\kappa \expect{\mcV_Q}_\Gamma^{1/2} \expect{\mcN_Q}_\Gamma^{1/2} & D=1,
  \\
  C a^3 (k_F^\kappa)^5 \abs{\log ak_F^\kappa}^{1/2} \expect{\mcN}_\Gamma 
  + C a (k_F^\kappa)^2 \expect{\mcV_Q}_\Gamma^{1/2} \expect{\mcN_Q}_\Gamma^{1/2} & D=2.
  \end{cases}
\end{align*}
\end{lemma}

\begin{lemma}[{\Cref{lem.general.correlation.mcE.1st} in $D$ dimensions}]
\label{lem.general.correlation.mcE.1st.dimensions}
Let $D=1,2$ and let $\Gamma$ be any state. Then, 
for any $0 < d,R < L$, any $ 0 < q_F < k_F^\kappa$ with $\beta(k_F^\kappa - q_F)^2 \gtrsim 1$ any $n \in \N$ and any $z \gtrsim 1$
\begin{align*}
\mcE_{\textnormal{pt}}(\Gamma)
  & \gtrsim 
    - L^D a^D R^2 \rho_0^{2+4/D} 
    - (k_F^\kappa)^{D+4} a^D R^4 \expect{\mcN}_\Gamma 
    - L^D a^D \rho_0^{2+2/D} R d^{-1}
  \\ & \quad 
    - L^{3D/4} \left[a^{2D} R^{-D-2} \rho_0^{2+2/D} d^{-D} + a^{2D} \rho_0^{3+4/D}(d^{-D} + \rho_0)\right]^{1/2}
      \Bigl[d^D S(\Gamma,\Gamma_0) 
  \\ & \qquad 
      + L^D \beta^4 (k_F^\kappa)^{D+4} d^{D-4}
      + L^D \beta^{-D/2} e^{-\kappa + \beta(2k_F^\kappa q_F - q_F^2)}
      + C_n (d q_F)^{-2n} \log Z_0\Bigr]^{1/4}
      - L^D \mfe_L.
\end{align*}
\end{lemma}

\begin{lemma}[{\Cref{lem.bdd.propagate.a.priori} in $D$ dimensions}]
\label{lem.bdd.propagate.a.priori.dimensions}
Let $D=1,2$ and let $\Gamma$ be an approximate Gibbs state. Then, 
for any $0\leq \lambda \leq 1$, 
\begin{align*}
  \expect{\mcN_Q}_{\Gamma^\lambda} 
    & \leq C L^D \zeta \kappa^{-1}    a^D \rho_0^2
    + L^D \mfe_L,
 \\
  \expect{\mcK_Q}_{\Gamma^\lambda} 
    & 
    \leq C L^D \zeta^{-D-1} \kappa^{D+1} a^D  \rho_0^{2+2/D}
    + L^D\mfe_L,
\\
  \expect{\mcV_Q}_{\Gamma^\lambda} 
    & 
    \leq C L^D \zeta^{-D/2-1}\kappa^{D/2+1} a^D \rho_0^{2+2/D}
    + L^D\mfe_L.
\end{align*}
\end{lemma}

With \Cref{lem.general.correlation.mcE.1st.dimensions} at hand we may prove the analogue of \Cref{prop.mcE.1st}:

\begin{lemma}[{\Cref{prop.mcE.1st} in $D$ dimensions}]\label{lem.mcE.1st.dimensions}
Let $\Gamma$ be an approximate Gibbs state. Then, for $\alpha > 0$ sufficiently small,
\begin{equation*}
\mcE_{\textnormal{pt}}(\Gamma) \geq 
 - C(z) L^D a^D \rho_0^{2+2/D} (a^D \rho_0)^{\sigma} - L^D \mfe_L,
 \qquad 
 \sigma = \begin{cases}
 \frac{2}{11} - \frac{30}{44} \alpha  
 & D=1,
 \\
 \frac{1}{8} - \frac{3}{8} \alpha
 & D=2,
 \end{cases}
\end{equation*}
with the function $C(z)$ uniformly bounded on compact subsets of $(0,\infty)$.
\end{lemma}

With these we may then give the

\begin{proof}[{Proof of \Cref{thm.main.2d}}] 
Recall \eqref{eqn.main}. 
Similarly as in \eqref{eqn.calc.expect.W} we evaluate (see also \cite[Lemma 3.6]{Lauritsen.Seiringer.2023a})
\begin{equation*}
\expect{\mcW}_{\Gamma_0}
  = 8\pi^2 \frac{-\Li_{2}(-z)}{(-\Li_{1}(-z))^{2}} a^2 \rho_0^{3} \left[1 + O(a^2 \rho_0) + \mfe_L\right].
\end{equation*}
We bound $\mcE_V(\Gamma)$ similarly as in \Cref{prop.mcE.V},
$\mcE_{V\varphi}(\Gamma)$, $\mcE_{\textnormal{scat}}(\Gamma)$ and $\mcE_{\textnormal{OD}}(\Gamma)$ 
using \Cref{lem.bdd.list.dimensions,lem.bdd.propagate.a.priori.dimensions}
and choosing the optimal $\delta = (a^2\rho_0)^{1/2}\abs{\log a^2\rho_0}$, 
and $\mcE_{\textnormal{pt}}(\Gamma)$ using \Cref{lem.mcE.1st.dimensions}.
\end{proof}

\begin{remark}[{The case of $D=1$ dimensions}]
The method presented here does not allow proving the one-dimensional analogue of \Cref{thm.main,thm.main.2d}. 
Indeed, combining the propagated a priori bounds of \Cref{lem.bdd.propagate.a.priori.dimensions}
with the bounds in \Cref{lem.bdd.list.dimensions}, 
the bounds on the error terms $\mcQ_{[\mcV_Q,\mcB]}$, $\mcQ_{V\varphi}$ and $\mcQ_{\textnormal{OD}}$ are too large. 
A similar issue occurs in the zero-temperature setting, see \cite[Remark A.4]{Lauritsen.Seiringer.2024a}.
\end{remark}

The proofs of Lemmas~\ref{lem.bdd.list.dimensions}--\ref{lem.mcE.1st.dimensions}
are as those of \Cref{lem.bdd.list,lem.bdd.propagate.a.priori} and \Cref{lem.general.correlation.mcE.1st,prop.mcE.1st} only with straightforward changes, which we shall sketch  here. 
Recall first the bounds on the scattering function:
\begin{lemma}[{\cite[Lemma A.2 and Remark A.4]{Lauritsen.Seiringer.2024a}}]
\label{lem.scat.D=1.2}
The scattering function $\varphi$ satisfies 
\begin{equation*}
\begin{aligned}
\norm{|\cdot|^n \varphi}_{L^1}
& \leq 
C a^D (k_F^\kappa)^{-n}, 
& n &= 1,2,
\quad  
&
\norm{|\cdot|^n \nabla^n \varphi}_{L^1}
& \leq C a^D \abs{\log a k_F^\kappa }, 
&n&=0,1,2
\quad  
\\
\norm{ |\cdot| \varphi}_{L^2 }
& \leq 
\begin{cases}
C a (k_F^\kappa)^{-1/2} & D=1,
\\
C a^2 \abs{\log ak_F}^{1/2} & D=2,
\end{cases}
\hspace*{-3em}
&&
\quad 
&
\norm{|\cdot|^n \nabla^n \varphi}_{L^2}
& \leq C a^{D/2}, 
&n&=0,1.
\end{aligned}
\end{equation*}
\end{lemma}
Furthermore, we have the analogues of \eqref{eqn.bdd.phi*bc} and \eqref{eqn.bdd.phi*bcc}:
\begin{align*}
\norm{\int \varphi(z-z') (b^r_{z'})^* c_{z'} \ud z'} 
& \leq C \left[(k_F^\kappa)^D \norm{\varphi}_{L^1} + (k_F^\kappa)^{D/2} \norm{\varphi}_{L^2}\right]
\\
\norm{\int \varphi(z-z') (b^r_{z'})^* c_{z'} c_z \ud z'}
& \leq C \left[(k_F^\kappa)^{3D/2+1} \norm{|\cdot|\varphi}_{L^1} + (k_F^\kappa)^{D+1} \norm{|\cdot|\varphi}_{L^2}\right].
\end{align*}

\begin{proof}[{Proof of \Cref{lem.bdd.list.dimensions}}]
The proof of \Cref{lem.bdd.list.dimensions} is exactly as that of \Cref{lem.bdd.list} given in \Cref{sec.bdd.commutators}. 
We simply state here all the intermediary bounds.

We may bound $\mcH_{0;B}^{\div r}$, $\mcQ_{\textnormal{Taylor}}$ and $\mcQ_{\textnormal{scat}}$ as 
\begin{align*}
\abs{\expect{\mcH_{0;B}^{\div r}}_{\Gamma}}
& \leq C 
(k_F^\kappa)^{D}
  \left[\norm{|\cdot|\nabla\varphi}_{L^1} + \norm{ |\cdot|^2 \Delta \varphi}_{L^1}\right] 
  \left[k_F^\kappa \expect{\mcN_Q}_\Gamma^{1/2} + \expect{\mcK_{Q}}_{\Gamma}^{1/2}\right] 
  \expect{\mcN}_\Gamma^{1/2},
\\
\abs{\expect{\mcQ_{\textnormal{Taylor}}}_{\Gamma}}
+
\abs{\expect{\mcQ_{\textnormal{scat}}}_\Gamma}
& \leq 
  C (k_F^{\delta^{-1}\kappa})^{D/2} (k_F^\kappa)^{D/2+} 
  \left[k_F^\kappa\norm{F}_{L^1} + \norm{\mcE_\varphi}_{L^1}\right]   
  \expect{\mcN_{Q}}_\Gamma^{1/2} \expect{\mcN}_{\Gamma}^{1/2}
\\ & \quad 
  +  C L^{D/2} (k_F^\kappa)^{D+1} 
  \left[k_F^\kappa \norm{F}_{L^2} + \norm{\mcE_\varphi}_{L^2}\right]  
  \expect{\mcN_{Q^>}}_{\Gamma}^{1/2}.
\end{align*}
To bound $\mcQ_{[\mcV_Q,\mcB]}$ we have the bounds 
\begin{align*}
\abs{\expect{\mcA_{4,\delta}}_{\Gamma}} + \abs{\expect{\mcA_{4,P}}_\Gamma}
  & \leq C (\kappa_F^\kappa)^{3D/2+1} \norm{|\cdot|\varphi}_{L^1} \norm{V}_{L^1}^{1/2} \expect{\mcV_Q}_{\Gamma}^{1/2} \expect{\mcN}_\Gamma^{1/2},
  \\
\abs{\expect{\mcA_{6,\delta}}_\Gamma} + \abs{\expect{\mcA_{6,P}}_\Gamma}  
  & \leq C \left[(k_F^\kappa)^{3D/2+1} \norm{|\cdot|\varphi}_{L^1} + (k_F^\kappa)^{D+1} \norm{|\cdot|\varphi}_{L^2}\right] 
  \\ & \quad \times \norm{V}_{L^1}^{1/2} \expect{\mcV_Q}_{\Gamma}^{1/2} \expect{\mcN}_\Gamma^{1/2}.
\end{align*}  
To bound $\mcQ_{V\varphi}$ we have the bounds 
\begin{align*}
\abs{\expect{\mcA_4}_\Gamma}
  & \leq C a^{D/2} (k_F^\kappa)^{3D/2+1} \norm{V}_{L^1}^{1/2} \norm{\varphi}_{L^2} \expect{\mcN_Q}_{\Gamma}^{1/2} \expect{\mcN}_\Gamma^{1/2},
\\
\abs{\expect{\mcA_6}_\Gamma}
  & \leq C a^{D/2} (k_F^\kappa)^{D+1}
    \left[(k_F^\kappa)^D \norm{\varphi}_{L^1} + (k_F^\kappa)^{D/2} \norm{\varphi}_{L^2}\right]
    \norm{V}_{L^1}^{1/2}
    \expect{\mcN_Q}_{\Gamma}^{1/2} \expect{\mcN}_\Gamma^{1/2}.
\end{align*}
To bound $\mcQ_{\textnormal{OD}}$ we have the bounds 
\begin{align*}
\abs{\expect{\mcA_{4,c,\delta}}_\Gamma} 
  & \leq C a^{D/2} (k_F^\kappa)^{3D/2+2} \norm{V}_{L^1}^{1/2} \norm{|\cdot|\varphi}_{L^2} \expect{\mcN}_\Gamma, 
\\
\abs{\expect{\mcA_{4,c,P}}_\Gamma}
  & \leq C a^{D/2} (k_F^\kappa)^{2D+2} \norm{|\cdot|\varphi}_{L^1} \norm{V}_{L^1}^{1/2} \expect{\mcN}_\Gamma,
\\
\abs{\expect{\mcA_{6,c,\delta}}_\Gamma} + \abs{\expect{\mcA_{6,c,P}}_\Gamma}
  & \leq C a^{D/2} (k_F^\kappa)^{D/2+1} 
  \left[(k_F^\kappa)^{3D/2+1} \norm{|\cdot|\varphi}_{L^1} + (k_F^\kappa)^{D+1} \norm{|\cdot|\varphi}_{L^2}\right]
  \\ & \quad \times 
  \norm{V}_{L^1}^{1/2} \expect{\mcN_Q}_{\Gamma}^{1/2} \expect{\mcN}_\Gamma^{1/2},
\\
\abs{\expect{\mcA_{4,b}}_\Gamma}
  & \leq C (k_F^\kappa)^D \norm{V}_{L^1}^{1/2} \norm{\varphi}_{L^2} \expect{\mcV_Q}_\Gamma^{1/2} \expect{\mcN_Q}_\Gamma^{1/2},
\\
\abs{\expect{\mcA_{6,b}}_\Gamma}
  & \leq C (k_F^\kappa)^{D/2}
  \left[(k_F^\kappa)^D \norm{\varphi}_{L^1} + (k_F^\kappa)^{D/2} \norm{\varphi}_{L^2}\right]
  \norm{V}_{L^1}^{1/2}
  \expect{\mcV_Q}_\Gamma^{1/2} \expect{\mcN_Q}_\Gamma^{1/2}.
\end{align*}
Combining with the bounds in \Cref{lem.scat.D=1.2} we conclude of the lemma. 
\end{proof}

\begin{proof}[{Proof of \Cref{lem.general.correlation.mcE.1st.dimensions}}]
The proof is (almost) the same as that for \Cref{lem.general.correlation.mcE.1st} given in \Cref{sec.1st.order}. 
The only non-immediate change is the bound for the sum $\sum_{|k|>k_F^\kappa - q_F} z e^{-\beta|k|^2}$. 
Similarly as in \Cref{sec.bdd.sums.riemann} we have 
\begin{equation*}
\sum_{|k|>k_F^\kappa - q_F} z e^{-\beta|k|^2} 
  \leq \begin{cases}
  C L \beta^{-1/2} z \frac{1}{1 + \beta^{1/2}(k_F^\kappa - q_F)} e^{-\beta(k_F^\kappa - q_F)^2} + L\mfe_L & D=1,
  \\
  C L^2 \beta^{-1} z  e^{-\beta(k_F^\kappa - q_F)^2} + L^2 \mfe_L & D=2,
  \end{cases}
\end{equation*}
This way we obtain the proof of the lemma.
\end{proof}

\begin{proof}[{Proof of \Cref{lem.bdd.propagate.a.priori.dimensions}}]
Same as that for \Cref{lem.bdd.propagate.a.priori}.
\end{proof}

\begin{proof}[{Proof of \Cref{lem.mcE.1st.dimensions}}]
We use \Cref{lem.general.correlation.mcE.1st.dimensions} with $q_F = \zeta^{-1/2} \rho_0^{1/D}$.
Choosing in addition $d \gg \zeta^{1/2} \rho_0^{-1/D}$ then, 
as in the proof of \Cref{prop.mcE.1st} above, 
by taking $n$ large enough, we 
can ensure that the last two summands in $[\ldots]^{1/4}$ are negligible. 
Then,
\begin{equation*}
\begin{aligned}
\mcE_{\textnormal{pt}}(\Gamma)  
  & \gtrsim
    - L^D a^D\rho_0^{2+2/D} \left[\kappa^{D/2+2}\zeta^{-D/2+2} R^2\rho_0^{2/D} + Rd^{-1} 
    \right]
      \\ & \quad 
    - L^{D} 
    \left[a^DR^{-D/2-1}  \rho_0^{4/D}d^{-D/2} + a^D \rho_0^{2+2/D}\right]
  \Bigl[
    \zeta d^D a^D \rho_0^2
  + \zeta^{-D/2+2} \kappa^{D/2+2} \rho_0^{1-4/D} d^{D-4}
  \Bigr]^{1/4}
    \\ & \quad 
    - L^D \mathfrak{e}_L.
\end{aligned}
\end{equation*}
Restricting to a compact set of $z$'s 
we take  
\begin{equation*}
d = \rho_0^{-1/D} (a^D\rho_0)^{-s},
\qquad 
R = \rho_0^{-1/D} (a^D\rho_0)^{t}
\end{equation*}
for some $s,t> 0$ to be determined. Then, 
$\mcE_{\textnormal{pt}}(\Gamma) \geq - C(z) L^D a^D\rho_0^{2+2/D} \left[(a^D\rho_0)^{\sigma} - \mathfrak{e}_L\right]$
with $C(z)$ bounded uniformly on compact sets of $z$'s and 
\begin{equation*}
\begin{aligned}
\sigma 
  & = \frac{1}{8}
  \min 
    \Bigl\{
  16t - (4D+16)\alpha \,, \, 
  2 - (4D+8)t + 2Ds \,, \, 
  \\ & \hphantom{= \frac{1}{8} \min \Bigl\{} \, 
  -(4D+8)t + (2D+8)s - (D+4) \alpha\,, \,
  2 - 2Ds\,, \, 
  (8-2D)s - (D+4) \alpha 
  \Bigr\}        
\end{aligned}
\end{equation*}
Choosing 
\begin{equation*}
s = \begin{cases}
\frac{3}{11} + \frac{30}{11}\alpha & D=1, 
\\
\frac{1}{4} + \frac{3}{4} \alpha & D=2,
\end{cases}   
\qquad 
t = \begin{cases}
\frac{1}{11} + \frac{10}{11}\alpha & D=1,
\\
\frac{1}{16} + \frac{21}{16} \alpha & D=2,
\end{cases}
\end{equation*}
we find 
\begin{equation*}
\sigma = \begin{cases}
\frac{2}{11} - \frac{30}{44}\alpha & D=1,
\\
\frac{1}{8} - \frac{3}{8} \alpha & D=2.
\end{cases}
\end{equation*}
This concludes the proof of \Cref{lem.mcE.1st.dimensions}
\end{proof}

\bigskip

{\it Acknowledgments.} Financial support by the Austrian Science Fund (FWF) through grant DOI: 10.55776/I6427 (as part of the SFB/TRR 352) is gratefully acknowledged.

\printbibliography

\end{document}